\documentclass[11pt]{scrartcl}
\usepackage[utf8]{inputenc}
\usepackage{hyperref}
\usepackage{amsmath,amssymb,amstext}
\usepackage{braket}
\usepackage{amsthm}
\usepackage{dsfont}
\usepackage{color}
\usepackage[title]{appendix}

\newtheorem{vor}{Assumption}[section]
\newtheorem{theorem}[vor]{Theorem}

\newtheorem{lem}[vor]{Lemma}
\newtheorem{cor}[vor]{Corollary}

\theoremstyle{definition}
\newtheorem{defi}[vor]{Definition}

\newtheorem{note}[vor]{Remark}
\numberwithin{equation}{section}

\begin{document}
\title{Validity of Bogoliubov's approximation for translation-invariant Bose gases}
\author{\textsc{Morris Brooks and Robert Seiringer}}
\date{}
\maketitle

\begin{abstract} 
\textsc{Abstract}. We verify Bogoliubov's approximation for translation-invariant Bose gases in the mean field regime, i.e. we prove that the ground state energy $E_N$ is given by $E_N=Ne_\mathrm{H}+\inf \sigma\left(\mathbb{H}\right)+o_{N\rightarrow \infty}(1)$, where $N$ is the number of particles, $e_\mathrm{H}$ is the minimal Hartree energy and $\mathbb{H}$ is the Bogoliubov Hamiltonian. As an intermediate result we show the existence of approximate ground states $\Psi_N$, i.e. states satisfying $\langle H_N\rangle_{\Psi_N}=E_N+o_{N\rightarrow \infty}(1)$, exhibiting complete Bose--Einstein condensation with respect to one of the Hartree minimizers.
\end{abstract}

\section{Introduction and Main Results}

We study the Hamiltonian $H_N$ acting on the Hilbert space $L^2_{\mathrm{sym}}(\mathbb{R}^{N\times d})\simeq\bigotimes_\mathrm{s}^N L^2(\mathbb{R}^{d})$ of $N$ identical bosons in $\mathbb{R}^d$ for $d\geq 1$, given by
\begin{align}
\label{Equation: Hamilton Operator}
H_N:=\sum_{i=1}^N T_i+\frac{1}{N-1}\sum_{i<j}v(x_i-x_j),
\end{align}
where $T$ is a non-negative and translation-invariant operator defined on the single particle space $L^2(\mathbb{R}^d)$ and the interaction potential $v$ is an even function. Typically we will think of $T$ as the non-relativistic energy $T=-\Delta$ or the pseudo relativistic energy $T =\sqrt{m^2-\Delta}-m$, and of the interaction $v$ as being attractive. The most prominent features of this model are the mean field scaling $\frac{1}{N-1}$ of the interaction energy and the invariance of $H_N$ under translations, which especially means that the Hamiltonian $H_N$ describes an unconfined system of $N$ bosons. By choosing a product state $\Psi:=u^{\otimes^N}$ as a test function, we obtain the trivial upper bound on the ground state energy $E_N:=\inf \sigma\left(H_N\right)$ per particle
\begin{align*}
N^{-1}E_N\leq N^{-1}\braket{H_N}_\Psi=\braket{T}_u+\frac{1}{2}\int\int |u(x)|^2 v(x-y)|u(y)|^2 \mathrm{d}x\mathrm{d}y=:\mathcal{E}_\mathrm{H}[u],
\end{align*}
where $\mathcal{E}_\mathrm{H}[u]$ is referred to as the Hartree energy functional. This upper bound is independent of the particle number $N$ due to the scaling by $\frac{1}{N-1}$ of the interaction. It is known under quite general assumptions on $v$ and $T$ that the upper bound 
\begin{align}
\label{Equation: Hartree Energy}
e_\mathrm{H}:=\inf_{\|u\|=1}\mathcal{E}_\mathrm{H}[u]
\end{align}
on the ground state energy per particle is asymptotically correct in the mean field limit $N\rightarrow \infty$, see \cite{LNR}. Furthermore, the Bogoliubov approximation \cite{B} predicts that the next order term in the approximation $E_N\approx N\, e_\mathrm{H}$ is of order one and given by the ground state energy of the corresponding Bogoliubov Hamiltonian $\mathbb{H}$, which is formally the second quantization of the Hessian $\mathrm{Hess}|_{u_0}\mathcal{E}_\mathrm{H}$ at a minimizer $u_0$. In the past decade, this conjecture has been proven for a variety of mean field models \cite{GS,LNSS,NS,S}, and also for systems with more singular interactions \cite{DN,BBSS2,BBSS1,BCS,BSS,NT}.  However, the rigorous verification of Bogoliubov's approximation has so far been restricted to confined systems only. In the case of translation-invariant models, we face the problem that minimizers of the Hartree energy functional $\mathcal{E}_\mathrm{H}$ are not unique and that the Hessian $\mathrm{Hess}|_{u_0}\mathcal{E}_\mathrm{H}$ at a minimizer $u_0$ does not exhibit a gap, i.e. we do not have an inequality of the form $\mathrm{Hess}|_{u_0}\mathcal{E}_\mathrm{H}\geq c$ with $c>0$. Novel ideas and techniques are required in order to deal with these translation-invariance specific problems, which we will develop in the course of this paper allowing us to verify Bogoliubov's prediction $E_N=N\, e_\mathrm{H}+\inf \sigma\left(\mathbb{H}\right)+o_N(1)$ for translation-invariant systems. As an intermediate step, we will construct a sequence of approximate ground states $\Psi_N$ satisfying complete Bose--Einstein condensation, which we believe to be of independent interest.

Note that the situation is different for time-dependent problems, where it is already well-known that fluctuations around a product state $u^{\otimes^N}$ evolve according to a (time-dependent) Bogoliubov operator, even for translation-invariant systems \cite{LNS}.\\

Due to the translation-invariance, it is clear that $H_N$ has no ground state and therefore we have to restrict our attention to sequences of approximate ground states $\Psi_N$. We will use the convention that states $\Psi$ are normed Hilbert space elements, i.e. $\|\Psi\|=1$. In our first result we show the existence of a sequence of approximate ground states $\Psi_N$, with the property that $\Psi_N$ is close to a product state $u_0^{\otimes^N}$ where $u_0$ minimizes the Hartree energy $\mathcal{E}_\mathrm{H}$. In this context, close means that the sequence $\Psi_N$ satisfies complete Bose--Einstein condensation with respect to the state $u_0$, i.e. the corresponding one particle density matrices $\gamma_N^{(1)}$ satisfy $\braket{\gamma_N^{(1)}}_{u_0}\underset{N\rightarrow \infty}{\longrightarrow}1$. In general we define the $k$-particle density matrix $\gamma_\Psi^{(k)}$ corresponding to a state $\Psi\in \bigotimes_\mathrm{s}^N L^2\left(\mathbb{R}^{d}\right)$ by the equation $\mathrm{Tr}\left[\gamma_\Psi^{(k)}\ B\right]=\braket{B\otimes 1\otimes \dots\otimes 1}_\Psi$ for all bounded $k$-particle operators $B$. This means in particular that we use the normalization convention $\mathrm{Tr}\left[\gamma^{(k)}_N\right]=1$. In order to prove complete Bose--Einstein condensation, we need certain assumptions concerning the kinetic energy operator $T$ and the Hartree theory, as well as a relative bound of the interaction potential $v$ in terms of the kinetic energy.\\

\begin{vor}
\label{Assumption: Part I}
The kinetic energy is given by $T:=\left(m^2-\Delta\right)^s-m^{2s}$ with $m> 0$ and $s\in (0,1]$, the interaction potential $v$ satisfies $\lim_{|x|\rightarrow \infty} v(x)=0$ and the chain of inequalities
\begin{align}
\label{Equation: Relative Bounds}
-\lambda T-\Lambda\leq v\leq |v|\leq \Lambda(T+1)
\end{align}
for some $\lambda\in (0,2)$ and $\Lambda\in (0,\infty)$. Furthermore, the Hartree energy defined in Eq.~(\ref{Equation: Hartree Energy}) is strictly negative, i.e. $e_\mathrm{H}<0$, and there exists a real-valued function $u_0\in L^2\left(\mathbb{R}^d\right)$ that minimizes the Hartree energy, i.e. $e_\mathrm{H}=\mathcal{E}_\mathrm{H}[u_0]$, and satisfies $\int_{[x_r\leq t]}|u_0(x)|^2\ \mathrm{d}x=\frac{1}{2}$ if and only if $t=0$, where $x_r$ is the $r$-th component of the vector $x\in \mathbb{R}^d$. Up to a complex phase, all other Hartree minimizers are given by translations of $u_0$, i.e. all minimizers are of the form $e^{i\theta}u_{0,t}$ with $\theta\in [0,2\pi), t\in \mathbb{R}^d$ and $u_{0,t}(x):=u_0(x-t)$.\\
\end{vor}

By the translation-invariance of the Hartree energy, any shift of a Hartree minimizer $u_0(x-t)$ is again a minimizer. Therefore, we can always choose the Hartree minimizer such that it is centered around zero, i.e. such that $\int_{[x_r\leq 0]}|u_0(x)|^2\ \mathrm{d}x=\frac{1}{2}$ for all $r\in \{1,\dots,d\}$. In particular, in case the minimizers $u$ of $\mathcal{E}_H$ satisfy $u>0$, the existence of a $u_0$ satisfying $\int_{[x_r\leq t]}|u_0(x)|^2\ \mathrm{d}x=\frac{1}{2}$ if and only if $t=0$ is always granted. Furthermore,  most of our proofs do not depend on the concrete structure $T=\left(m^2-\Delta\right)^s-m^{2s}$ of the kinetic energy, and it is sufficient to assume instead that the operator $T$ is of the translation-invariant form $T=t(i\nabla)$ for some $t$ with $t(p)\underset{|p|\rightarrow \infty}{\longrightarrow}\infty$ such that the Hartree approximation $\frac{1}{N}E_N\underset{N\rightarrow \infty}{\longrightarrow}e_\mathrm{H}$ as well as the IMS localization formula in Lemma \ref{Lemma: IMS Localization} hold.

With Assumption \ref{Assumption: Part I} at hand, we obtain our first main result Theorem \ref{Theorem: Bose--Einstein condensation of Ground States}, which we will prove in Section \ref{Section: Bose--Einstein condensation of Ground States}.\\

\begin{theorem}
\label{Theorem: Bose--Einstein condensation of Ground States}
Given Assumption \ref{Assumption: Part I}, there exists a sequence of states $\Psi_N\in \bigotimes_\mathrm{s}^N L^2\left(\mathbb{R}^d\right)$ with $\braket{H_N}_{\Psi_N}=E_N+o_{N\rightarrow \infty}(1)$, exhibiting complete Bose--Einstein condensation with respect to the state $u_0$, i.e.
\begin{align}
\label{Equation: Bose--Einstein condensation}
\braket{\gamma_N^{(1)}}_{u_0}\underset{N\rightarrow \infty}{\longrightarrow} 1.
\end{align}
\end{theorem}

Since Assumption \ref{Assumption: Part I} implies the validity of the Hartree approximation in the form $\frac{1}{N}E_N\underset{N\rightarrow \infty}{\longrightarrow}e_\mathrm{H}$, see \cite{LNR}, it is clear that the product state $u_0^{\otimes^N}$, which trivially satisfies perfect Bose--Einstein condensation, approximates the ground state energy to leading order, i.e. $\braket{H_N}_{u_0^{\otimes^N}}=E_N+o_{N\rightarrow \infty}(N)$. In Theorem \ref{Theorem: Bose--Einstein condensation of Ground States} we improve this result by constructing a Bose--Einstein condensate that approximates $E_N$ even up to terms $o_{N\rightarrow \infty}(1)$. Note, however, that Theorem \ref{Theorem: Bose--Einstein condensation of Ground States} claims nothing about the rate of convergence in Eq.~(\ref{Equation: Bose--Einstein condensation}). One can improve this result a posteriori by using the trial states in our proof of the upper bound in Theorem \ref{Theorem: Main Theorem}, which yields for any given sequence $c_N\underset{N\rightarrow \infty}{\longrightarrow}\infty$ a sequence of approximate ground states $\widetilde{\Psi}_N$ satisfying
\begin{align*}
|\braket{\widetilde{\gamma}_N^{(1)}}_{u_0}-1|\leq \frac{c_N}{N}.
\end{align*}
It follows from our proof of the lower bound in Theorem \ref{Theorem: Main Theorem} that this result is optimal in the sense that any sequence with $|\braket{\widetilde{\gamma}_N^{(1)}}_{u_0}-1|=O_{N\rightarrow \infty}\left(\frac{1}{N}\right)$ cannot be a sequence of approximate ground states.

Furthermore it follows from the proof of Theorem \ref{Theorem: Bose--Einstein condensation of Ground States} that for any sequence $c_N\underset{N\rightarrow \infty}{\longrightarrow}\infty$, there exist states $\Psi'_N$ exhibiting complete Bose--Einstein condensation with $\braket{H_N}_{\Psi'_N}\leq E_N+\frac{c_N}{N}$. Again it is a consequence of our proof of the lower bound that this result is optimal in the sense that any sequence with $\braket{H_N}_{\Psi'_N}= E_N+O_{N\rightarrow \infty}\left(\frac{1}{N}\right)$ does not satisfy complete Bose--Einstein condensation.\\

\textbf{Proof strategy of Theorem \ref{Theorem: Bose--Einstein condensation of Ground States}}. With Assumption \ref{Assumption: Part I} at hand, we can apply the results in \cite{LNR} which tell us that the Hartree asymptotics $\frac{1}{N}E_N\underset{N\rightarrow \infty}{\longrightarrow}e_\mathrm{H}$ holds true and that any sequence of approximate ground states $\Psi_N$ has a subsequence such that the $k$-particle density matrices converge weakly to a mixture of not necessarily normed Hartree minimizers. This means that there exists a probability measure $\mu$ supported on functions $u$ with $\|u\|\leq 1$ and $\mathcal{E}_\mathrm{H}[u]=\inf_{\|v\|=\|u\|}\mathcal{E}_\mathrm{H}[v]$, such that the $k$-particle density matrix of the subsequence $\Psi_{N_j}$ satisfies
\begin{align}
\label{Equation: Quantum De Finetti}
\mathrm{Tr}\left[\gamma_{N_j}^{(k)}\ K\right]\underset{j\rightarrow \infty}{\longrightarrow}\int \mathrm{Tr}\left[\big(\! \ket{u}\bra{ u}\! \big)^{\otimes^k} K\right] \mathrm{d}\mu(u)
\end{align}
for any compact $k$ particle operator $K$. The proofs in \cite{LNR} rely on the quantum de Finetti theorem (see also \cite{St,HM}), which identifies states on the infinite symmetric tensor product as the convex hull of product states. In order to prove Theorem \ref{Theorem: Bose--Einstein condensation of Ground States}, we have to construct a sequence of approximate ground states $\Psi_N$ such that the corresponding measure $\mu$ in Eq.~(\ref{Equation: Quantum De Finetti}) is equal to the delta measure $\delta_{u_0}$. In particular this means that $\mu$ has to be supported on the set of normed elements $\|u\|=1$, or equivalently we have to make sure that mass cannot escape to infinity. For confined systems satisfying a binding inequality, it has been shown in \cite{LNR} that $\mu$ is always supported on normed elements. For translation-invariant systems this is no longer the case, since one can always find $y_N\in \mathbb{R}^{d}$ such that $\widetilde{\Psi}_N\underset{N\rightarrow \infty}{\rightharpoonup} 0$ where 
\begin{align*}
\widetilde{\Psi}_N\left(x^{(1)},\dots,x^{(N)}\right):=\Psi_N\left(x^{(1)}-y_N,\dots,x^{(N)}-y_N\right)
\end{align*}
for all $\left(x^{(1)},\dots,x^{(N)}\right)\in \mathbb{R}^{N\times d}$, and therefore the corresponding measure is supported on $\{0\}$ only. While one could  circumvent this issue by factoring out the center-of-mass variable,  we avoid doing this since there is no straightforward analogue of product states and Bose--Einstein condensation in the space of relative coordinates. Alternatively we overcome this problem by localizing a sequence of approximate ground states $\Psi_N$ only to configurations that are centered around zero. It turns out that the median of a configuration $x=\left(x^{(1)},\dots,x^{(N)}\right)\in \mathbb{R}^{N\times d}$, respectively a regularized version of the median, is the right statistical quantity to measure whether a configuration is centered around the origin or not. Furthermore, we will energetically rule out configurations where the mass is split up in two or multiple parts, e.g. we will rule out configurations where $\frac{N}{2}$ particles are very far from the other $\frac{N}{2}$ particles. We conclude that the mass is concentrated at the origin and therefore it does not escape to infinity.

 In order to identify the support of the measure $\mu$ in Eq.~(\ref{Equation: Quantum De Finetti}), note that all Hartree minimizers are up to a complex phase translations of the minimizer $u_0$, which is a function centered around zero. Consequently, up to this complex phase, $u_0$ is the only minimizer with the property of being centered around zero. Using the support property of $\Psi_N$, this already suggests that the measure $\mu$ should be supported on states of the form $\{e^{i\theta}u_0:\theta\in [0,2\pi)\}$ only. Since $\ket{e^{i\theta}u_0}\bra{e^{i\theta}u_0}=\ket{u_0}\bra{u_0}$ defines the same density matrix for all complex phases $e^{i\theta}$, this support property of the measure $\mu$ implies the convergence of the density matrix $\gamma_{N}^{(k)}$ to a single condensate $\big(\! \ket{u_0}\bra{ u_0}\! \big)^{\otimes^k}$.\\

Having a sequence of approximate ground states at hand that satisfies complete Bose--Einstein condensation is a crucial prerequisite in identifying the sub-leading term in the energy asymptotics $E_N=N\ e_\mathrm{H}+o(N)$. In the following, let $u_0,u_1,\dots,u_d,u_{d+1},\dots$ be a real orthonormal basis of $L^2\!\left(\mathbb{R}^d\right)$, where $u_0$ is the Hartree minimizer from Assumption \ref{Assumption: Part I} and $u_1,\dots,u_d$ a basis of the vector space spanned by the partial derivatives $\braket{\partial_{x_1}u_0,\dots,\partial_{x_d}u_0}$. Since the functional $\mathcal{E}_\mathrm{H}$ is invariant under a phase change $u\mapsto e^{i\theta}u$, we can restrict ourself to states $u$ with $\braket{u_0,u}\geq 0$. Then, the Hessian $\mathrm{Hess}|_{u_0}\mathcal{E}_\mathrm{H}$ of the Hartree energy is a real quadratic form defined on $\{u_0\}^\perp\subset L^2\left(\mathbb{R}^d\right)$, and consequently there exist coefficients $Q_{i,j},G_{i,j}\in \mathbb{C}$, $i,j\in \mathbb{N}$, such that $\mathrm{Hess}|_{u_0}\mathcal{E}_\mathrm{H}[z]=\sum_{i,j=1}^\infty \left(Q_{i,j} \overline{z}_i z_j+\overline{G}_{i,j}z_i z_j+G_{i,j}\overline{z}_i\overline{z}_j\right)$, where $z_i$ are the coordinates of $z\in \{u_0\}^\perp$. In order to define the Bogoliubov operator $\mathbb{H}$, let $a_i,a_i^\dagger$ be the annihilation/creation operators corresponding to the state $u_i\in L^2\left(\mathbb{R}^d\right)$. Following \cite{LNSS} we formally define $\mathbb{H}$ as the second quantization of the Hessian $\mathrm{Hess}|_{u_0}\mathcal{E}_\mathrm{H}$, i.e. 
\begin{align}
\label{Equation: Bogoliubov}
\mathbb{H}:=\sum_{i,j=1}^\infty \left(Q_{i,j}\ a^\dagger_i a_j+\overline{G}_{i,j}\ a_i  a_j+G_{i,j}\ a_i^\dagger a_j^\dagger\right).
\end{align}
For a rigorous construction see Definition \ref{Definition: Bogoliubov Operator}.

Note that due to the translation-invariance, the Hessian $\mathrm{Hess}|_{u_0}\mathcal{E}_\mathrm{H}$ is degenerate in the directions $u_j$ for $j\in \{1,\dots,d\}$, i.e. $\mathrm{Hess}|_{u_0}\mathcal{E}_\mathrm{H}[u_j]=0$. The following Assumption makes sure that $\mathrm{Hess}|_{u_0}\mathcal{E}_\mathrm{H}$ is non-degenerate in all other directions.
\begin{vor}
\label{Assumption: Part II}
The partial derivatives of $u_0$ are in the form domain of $T$, 
 and there exists a constant $\eta>0$ such that
\begin{align}
\label{Equation: Non-degenerate Hessian}
\mathrm{Hess}|_{u_0}\mathcal{E}_\mathrm{H}[z]\geq \eta\, \|z\|^2
\end{align}
for all $z$ of the form $z=i\sum_{j=1}^d s_j u_j+z_{>d}$ with $s_j\in \mathbb{R}$ and $z_{>d}\in \{u_0,\partial_{x_1}u_0,\dots,\partial_{x_d}u_0\}^\perp$. Furthermore, the Hartree minimizer $u_0$ is an element of $H^2(\mathbb{R}^d)$.\\
\end{vor}

With the Assumption \ref{Assumption: Part II} at hand, we arrive at our second main Theorem, which identifies the sub-leading term in the energy asymptotics as the ground state energy $\inf \sigma\left(\mathbb{H}\right)$ of the Bogoliubov operator $\mathbb{H}$.\\

\begin{theorem}
\label{Theorem: Main Theorem}
Let $E_N$ be the ground state energy of the Hamiltonian $H_N$ defined in Eq.~(\ref{Equation: Hamilton Operator}), $e_\mathrm{H}$ the Hartree energy defined in Eq.~(\ref{Equation: Hartree Energy}) and let $\mathbb{H}$ be the Bogoliubov operator defined in Eq.~(\ref{Equation: Bogoliubov}). Given Assumption \ref{Assumption: Part I} and Assumption \ref{Assumption: Part II}, we have
\begin{align}
\label{Equation: Energy asymptotics}
E_N=N\, e_\mathrm{H}+\inf \sigma\left(\mathbb{H}\right)+o_{N\rightarrow \infty}\left(1\right).\\
\nonumber
\end{align}
\end{theorem}

Examples of systems satisfying both Assumptions \ref{Assumption: Part I} and \ref{Assumption: Part II}, and hence our Theorem \ref{Theorem: Main Theorem} applies to, are as follows.\\

\textbf{Example (I)}. Let us first consider a system of $N$ non-relativistic bosons in $\mathbb{R}^3$ interacting with each other via a Newtonian potential
\begin{align*}
H_N:=-\sum_{i=1}^N \Delta_i-\frac{g}{N-1}\sum_{i<j}\frac{1}{|x_i-x_j|}
\end{align*}
with $g>0$. Existence and uniqueness of the Hartree minimizer $u_0$, in the sense of Assumption \ref{Assumption: Part I}, have been shown in \cite{Li}. Moreover, $u_0$ is strictly positive and smooth, hence satisfies all the other requirements of Assumptions~\ref{Assumption: Part I} and~\ref{Assumption: Part II}. The non-degeneracy of the Hessian follows from the results in \cite{Le} by standard arguments, see for instance \cite{FLS}. Furthermore, it is clear by a scaling argument that $e_\mathrm{H}<0$ and that we can bound the interaction energy in terms of the kinetic energy by $\frac{1}{|x|}\leq -\epsilon \Delta+\frac{1}{4\epsilon}$ for all $\epsilon>0$.\\

\textbf{Example (II)}. As a second example let us consider a system of $N$ pseudo-relativistic bosons in $\mathbb{R}^3$ with positive mass $m>0$, interacting with each other via a Newtonian potential
\begin{align*}
H_N:=\sum_{i=1}^N \left(\sqrt{m^2-\Delta_i}-m\right)-\frac{g}{N-1}\sum_{i<j}\frac{1}{|x_i-x_j|},
\end{align*}
where we assume that the coupling strength satisfies $g\in (0,g_*)$ for a suitable positive constant $g_*>0$. It has been shown in \cite{LY} that there exists a Hartree minimizer $u_0$ as long as the coupling $g$ is below a critical value, in which case the Hartree approximation $\lim_{N\rightarrow \infty} N^{-1}E_N=e_\mathrm{H}$ holds true. The chain of operator inequalities in Assumption \ref{Assumption: Part I} holds as long as the coupling is below the critical value $\frac{4}{\pi}$, see \cite{H,K}. By restricting the attention to possibly smaller couplings $g\in (0,g_*)$ it has been shown in \cite{Le,GZ} that minimizers $u_0$ are unique in the sense of Assumption \ref{Assumption: Part I}. Furthermore it follows from the results in \cite{Le,GZ} that the Hessian is non-degenerate in the sense of Assumption \ref{Assumption: Part II} for couplings $g$ below a critical value. We will verify this explicitly in Appendix \ref{Appendix: Non-degenerate Hessian}, using an argument similar to the one in \cite{FLS} for non-relativistic systems. (The argument in \cite{FLS} is based on scaling the coordinates and hence not directly applicable in the pseudo-relativistic case.) \\

\textbf{Example (III)}. As a third example let us consider the exactly solvable model of $N$ non-relativistic bosons on the real line $\mathbb{R}$, interacting with each other via an attractive delta potential
\begin{align*}
H_N:=-\sum_{i=1}^N \partial_i^2-\frac{\lambda}{N-1}\sum_{i<j}\delta(x_i-x_j),
\end{align*}
where $\lambda>0$, see \cite{M} for an explicit expression of the ground state energy. In this case the Hartree energy $\mathcal{E}_\mathrm{H}$ is given by
\begin{align*}
\mathcal{E}_\mathrm{H}[u]=\int_{-\infty}^\infty |u'(x)|^2\mathrm{d}x-\frac{\lambda}{2}\int_{-\infty}^\infty |u(x)|^4\mathrm{d}x.
\end{align*}
For $d=1$ we have $\delta\leq -\epsilon\, \partial^2+\frac{1}{4\epsilon}$ for all $\epsilon>0$ in the sense of quadratic forms, and therefore Eq.~(\ref{Equation: Relative Bounds}) in Assumption \ref{Assumption: Part I} holds. By a scaling argument it is clear that $e_\mathrm{H}<0$ and minimizers of the Hartree energy are unique in the sense of Assumption \ref{Assumption: Part I}, see \cite{Kw} where the uniqueness of solutions to the corresponding Euler-Lagrange equation 
is verified. Furthermore the coercivity assumption in Eq.~(\ref{Equation: Non-degenerate Hessian}) is a consequence of the slightly different coercivity result in \cite{W} (arguing, e.g., as in Appendix~\ref{Appendix: Non-degenerate Hessian}).\\

We remark that in Examples (I) and (III), the value of the coupling constant, and hence also the factor $1/(N-1)$ in front of the interaction term, is irrelevant, since it can be replaced by any other value by a simple scaling of the coordinates. This does not apply to Example (II), however. \\

\textbf{Proof strategy of Theorem \ref{Theorem: Main Theorem}}. We will verify the upper bound in our main result (\ref{Equation: Energy asymptotics}) analogously to the proof of the energy asymptotics for confined systems in \cite{LNSS}. The more difficult lower bound will be based on the correspondence between the Hartree energy $\mathcal{E}_\mathrm{H}$ and the Hamiltonian $H_N$. This correspondence becomes evident when we rewrite $H_N$ in the language of second quantization. For this purpose, let us define the rescaled creation operators $b_j^\dagger:=\frac{1}{\sqrt{N}}a_{u_j}^\dagger$, where we suppress the $N$ dependence in our notation for simplicity. Then we can write
\begin{align}
\label{Equation: H_N in second quantization}
N^{-1}H_N=\sum_{i,j=0}^\infty T_{i,j}\ b_i^\dagger b_j+\frac{N}{N-1}\frac{1}{2}\sum_{ij,k\ell}\hat{v}_{ij,k\ell}\ b_i^\dagger b_j^\dagger b_k b_{\ell},
\end{align}
where $T_{i,j}$ are the matrix entries of the operator $T$ with respect to the basis $\{u_i:i\in \mathbb{N}_0\}$ and $\hat{v}_{ij,k\ell}$ are the ones of the two body multiplication operator $\hat{v}=v(x-y)$ with respect to the basis $\{u_i\otimes u_j:i,j\in \mathbb{N}_0\}$. Up to the factor $\frac{N}{N-1}$, the Hartree energy $\mathcal{E}_\mathrm{H}[u]$ 
\begin{align*}
\mathcal{E}_\mathrm{H}\left[u\right]=\sum_{i,j=0}^\infty T_{i,j}\ \overline{c_i}\ c_j+\frac{1}{2}\sum_{ij,k\ell}\hat{v}_{ij,k\ell}\ \overline{c_i}\ \overline{c_j}\ c_k\ c_{\ell}
\end{align*}
is represented by the same symbolic expression as in Eq.~(\ref{Equation: H_N in second quantization}), i.e. we plug in the complex numbers $c_i$ instead of the operators $b_i$. Before investigating the next order term in the energy asymptotics, let us discuss the next order expansion of the commutative counterpart $\mathcal{E}_\mathrm{H}[u]= e_\mathrm{H}+o\left(\|u-u_0\|\right)$, which is given by the Hessian of the functional $\mathcal{E}_\mathrm{H}$. Since the Hartree energy is defined on the infinite dimensional manifold $\{u\in L^2\left(\mathbb{R}^d\right):\|u\|=1,\braket{u_0,u}\geq 0\}\subset L^2\left(\mathbb{R}^d\right)$, it is convenient to introduce the embedding
\begin{align}
\label{Equation: Embedding}
\iota:\begin{cases}\{z\in \{u_0\}^\perp:\|z\|\leq 1\}\longrightarrow \{u\in L^2\left(\mathbb{R}^d\right):\|u\|=1,\braket{u_0,u}\geq 0\},\\
\ \ z\mapsto \iota(z):=\sqrt{1-\|z\|^2}\ u_0+z.\end{cases}
\end{align}
Using the chart $\iota$, we can express the Hessian as $\mathrm{Hess}|_{u_0}\mathcal{E}_\mathrm{H}=D^2|_0 \left(\mathcal{E}_\mathrm{H}\circ \iota\right)$ and the second order expansion at $z=0$ is given by
\begin{align*}
\mathcal{E}_\mathrm{H}[\iota(z)]=e_\mathrm{H}+\mathrm{Hess}|_{u_0}\mathcal{E}_\mathrm{H}[z]+o\left(\|z\|^2\right).
\end{align*}
In contrast to confined systems, the Hessian for translation-invariant systems is always degenerate in the directions $u_1,\dots,u_d$, i.e. $\mathrm{Hess}|_{u_0}\mathcal{E}_\mathrm{H}\left[u_j\right]=0$ for $j\in \{1,\dots,d\}$. It is important to observe that the manifold of minimizers $\mathcal{M}:=\{z: \mathcal{E}_\mathrm{H}[\iota(z)]=e_\mathrm{H}\}$ is not contained in the null space of the Hessian $\{z:\mathrm{Hess}|_{u_0}\mathcal{E}_\mathrm{H}[z]=0\}$. Therefore, we do not have the crucial estimate $\mathcal{E}_\mathrm{H}[\iota(z)]\geq e_\mathrm{H}+(1-\epsilon)\mathrm{Hess}|_{u_0}\mathcal{E}_\mathrm{H}[z]$, $0<\epsilon<1$, not even in an arbitrary small neighborhood of zero. In order to obtain such an inequality, we will introduce yet another transformation $F$ on the ball $\{z\in \{u_0\}^\perp:\|z\|\leq 1\}$, such that $D|_0 F$ is the identity and such that $F$ flattens the manifold of minimizers $\mathcal{M}$, i.e. $\mathcal{E}_\mathrm{H}\left[(\iota\circ F)\left(\sum_{j=1}^d t_j u_j\right)\right]=e_\mathrm{H}$ for all $t_j\in \mathbb{R}$. For a concrete construction of $F$ see Eq.~(\ref{Equation: Definition F}) in Section \ref{Section: Energy Asymptotic}. Under the assumption that the Hessian is only degenerate in the directions $u_j$, see Assumption \ref{Assumption: Part II}, we obtain for any fixed $\epsilon>0$ and $z$ small enough the important estimate
\begin{align}
\label{Inequality: Commutative}
\mathcal{E}_\mathrm{H}[\left(\iota\circ F\right)(z)]\geq e_\mathrm{H}+(1-\epsilon)\mathrm{Hess}|_{u_0}\mathcal{E}_\mathrm{H}[z].
\end{align} 
Returning to the Hamiltonian $H_N$, we will introduce non-commutative counterparts to the embedding $\iota$ and the transformation $F$. The counterpart to $\iota$ is the excitation map $U_N$ introduced in \cite{LNSS}, where it has already been used to verify the next order approximation of the ground state energy for confined systems. It is defined as
\begin{align}
\label{Equation: Excitation Map}
U_N\left(u_0^{\otimes^{i_0}}\otimes_\mathrm{s} u_1^{\otimes^{i_1}}\otimes_\mathrm{s}\dots\otimes_\mathrm{s} u_m^{\otimes^{i_m}}\right):=u_1^{\otimes^{i_1}}\otimes_\mathrm{s}\dots\otimes_\mathrm{s} u_m^{\otimes^{i_m}}
\end{align}
for non-negative integers $i_0+\dots+i_m=N$, mapping the $N$ particle space $\bigotimes_\mathrm{s}^N L^2\left(\mathbb{R}^d\right)$ into the truncated Fock space $\mathcal{F}_{\leq N}\left(\{u_0\}^\perp\right):=\bigoplus_{n\leq N}\bigotimes_\mathrm{s}^n \{u_0\}^\perp$ over modes orthogonal to $u_0$, where the symmetric tensor product $\otimes_\mathrm{s}$ is defined as
\begin{align*}
\psi_k\! \otimes_s\! \psi_\ell\! \left(\! x^{(1)},\dots,x^{(k+\ell)}\! \right)\! : =\! \frac{1}{\sqrt{\ell! k!(k\! +\! \ell)!}}\!  \sum_{\sigma\in S_{k+\ell}}\! \! \! \! \psi_k\! \left(\! x^{(\sigma_1)},\dots,x^{(\sigma_k)}\! \right)\! \psi_\ell\! \left(\! x^{(\sigma_{k+1})},\dots,x^{(\sigma_{k+\ell})}\! \right)
\end{align*}
for $\psi_k\in \bigotimes_\mathrm{s}^k L^2\!\left(\mathbb{R}^3\right)$ and $\psi_\ell\in \bigotimes_\mathrm{s}^\ell L^2\!\left(\mathbb{R}^3\right)$, and $S_{n}$ is the set of permutations on $\{1,\dots,n\}$. Regarding the transformation $F$, we construct the counterpart $\mathcal{W}_N$ in Definition \ref{Definition: Unitary Transformation} as a certain transformation reminiscent of the Gross transformation in \cite{G,N}, operating on the space $\mathcal{F}\left(\{u_0\}^\perp\right)$. Based on these correspondences and the observation that the Bogoliubov operator is the non-commutative analogue of the Hessian $\mathrm{Hess}|_{u_0}\mathcal{E}_\mathrm{H}$, we obtain the following inequality analogous to Eq.~(\ref{Inequality: Commutative})
\begin{align}
\label{Inequality: Non-Commutative}
\left(\mathcal{W}_N U_N\right) N^{-1}H_N \left(\mathcal{W}_N U_N\right)^{-1}\gtrsim e_\mathrm{H}+(1-\epsilon)N^{-1}\mathbb{H}.
\end{align}
We write $\gtrsim$ for two reasons: There are errors of order $o\left(\frac{1}{N}\right)$ coming from the non-commutative nature of $H_N$; moreover Eq.~(\ref{Inequality: Non-Commutative}) only holds for states $\Psi$ that satisfy a strengthened version of Bose--Einstein condensation of the form $U_N \Psi \in \mathcal{F}_{\leq M_N}\left(\{u_0\}^\perp\right)$ with $M_N\ll N$, which corresponds to the fact that Inequality (\ref{Inequality: Commutative}) only holds for small $z$. The rigorous verification of inequality (\ref{Inequality: Non-Commutative}) will be the content of Sections \ref{Section: Energy Asymptotic} and \ref{Section: Results in the transformed picture}.\\

Our construction of $\mathcal{W}_N$ and the proof of Inequality (\ref{Inequality: Non-Commutative}) do not rely on the specific structure of $H_N$ or $L^2(\mathbb{R}^d)$, and they can be generalized for various mean field models with continuous symmetries. The essential assumption is that the dimension of the symmetry group agrees with the nullity of the Hessian, i.e. the Hessian is as non-degenerate as possible in the presence of a continuous symmetry, see Assumption \ref{Assumption: Part II}.\\

\textbf{Outline}. The paper is structured as follows. In Section \ref{Section: Bose--Einstein condensation of Ground States} we construct a sequence of approximate ground states satisfying complete Bose--Einstein condensation, which verifies our first main Theorem \ref{Theorem: Bose--Einstein condensation of Ground States}. The methods and results of Section \ref{Section: Bose--Einstein condensation of Ground States} can be read independently of the rest of the paper, which is dedicated to the proof of our second main Theorem \ref{Theorem: Main Theorem}. In Section \ref{Secction: Fock Space Formalism}, we will introduce the relevant Fock spaces as well as a useful notation for second quantized operators, which we believe to be intuitive and natural for our problem. With the basic notions at hand, we will follow the strategy in \cite{LNSS} and reformulate our problem in a Fock space language using the excitation map $U_N$. In Section \ref{Section: Energy Asymptotic} we will discuss the energy asymptotics of $H_N$, starting with a precise definition of the Bogoliubov operator $\mathbb{H}$ in Subsection \ref{Subsection: Construction of the Bogoliubov Operator}, the verification of the upper bound in Subsection \ref{Subsection: Upper Bound} and the proof of the lower bound in Subsection \ref{Subsection: Lower Bound}, up to the proof of the main technical inequality Eq.~(\ref{Inequality: Non-Commutative}). The proof of the latter is the content of Section \ref{Section: Results in the transformed picture}.

\section{Bose--Einstein Condensation of Ground States}
\label{Section: Bose--Einstein condensation of Ground States}
In this section we will prove Theorem \ref{Theorem: Bose--Einstein condensation of Ground States} by constructing a sequence $\Psi_N$ of approximate ground states satisfying complete Bose--Einstein condensation. The concrete construction of $\Psi_N$ will be part of Subsection \ref{Subsection: Localization of the Ground State}, where we introduce a suitable localization method and verify that mass does not escape to infinity. In the following Subsection \ref{Subsection: Convergence to a single Condensate}, we will use this to verify complete Bose--Einstein condensation of the sequence $\Psi_N$.

\subsection{Localization of the Ground State}
\label{Subsection: Localization of the Ground State}
In the following we are constructing a sequence of states $\Psi_N$, i.e. elements satisfying $\|\Psi_N\|=1$, localized only to configurations $x\in \mathbb{R}^{N\times d}$ centered at zero, such that $\braket{H_N}_{\Psi_N}=E_N+o_{N\rightarrow \infty}(1)$. For such a sequence we will verify that mass cannot escape to infinity. As it turns out, the regularized median $M_N$, which we will define in the subsequent Definition \ref{Definition: Localization}, is the right statistical quantity to measure the center 
\begin{align*}
x_\mathrm{center}:=\left(M_{N,k}\left(x_1^{(1)},\dots,x_1^{(N)}\right),\dots,M_{N,k}\left(x_d^{(1)},\dots,x_d^{(N)}\right)\right)\in \mathbb{R}^d
\end{align*}
of a configuration $x=\left(x^{(1)},\dots,x^{(N)}\right)\in \mathbb{R}^{N\times d}$, where $x^{(j)}=\left(x^{(j)}_1,\dots,x^{(j)}_d\right)\in \mathbb{R}^d$ is the coordinate vector of the $j$-th particle.

\begin{defi}[Localization]
\label{Definition: Localization}
Given $N\in \mathbb{N}$ and $k$ such that $k+\frac{N}{2}\in \mathbb{N}$, we define the regularized median $M_{N,k}:\mathbb{R}^N\longrightarrow \mathbb{R}$ as the unique permutation-invariant function that is defined for all $x^{(1)}\leq \dots\leq x^{(N)}$ as
\begin{align*}
M_{N,k}\left(x^{(1)},\dots,x^{(N)}\right):=\frac{1}{2k+1}\sum_{j=\frac{N}{2}-k}^{\frac{N}{2}+k}x^{(j)}.
\end{align*}
\end{defi}

In the IMS-type estimate of the following Lemma \ref{Lemma: IMS Localization}, which has been proven in \cite[Lemma 7]{LL}, we will make use of the specific structure of the operator $T=\left(m^2-\Delta\right)^s-m^{2s}$. Note that this is the only place where the specific structure is relevant for us.

\begin{lem}
\label{Lemma: IMS Localization}
Let $T=\left(m^2-\Delta\right)^s-m^{2s}$ be as in Assumption \ref{Assumption: Part I} and let $\{\chi_i:i\in I\}$ be a family of $W^{1,\infty}\! \left(\mathbb{R}^d\right)$ functions with $\sum_i \chi_i^2=1$. With the definition $C:=m^{2(s-1)}s$ we have for all states $u\in L^2\left(\mathbb{R}^d\right)$ 
\begin{align*}
\sum_{i\in I} \braket{T}_{\chi_i u}\leq \braket{T}_u+C\ \left\| \sum_{i\in I}|\nabla\chi_i|^2 \right\|_\infty.
\end{align*}
\end{lem}

\begin{lem}
\label{Lemma: Localized Ground State}
Let $E_N$ denote the ground state energy of $H_N$ and let $k_N$ be a sequence with $\sqrt{N}\ll k_N\ll N$ such that $k_N+\frac{N}{2}\in \mathbb{N}$. Then there exists a sequence of states $\Psi_N$ in $L_{\mathrm{sym}}^2\left(\mathbb{R}^{N\times d}\right)$ with $\braket{H_N}_{\Psi_N}-E_N\underset{N\rightarrow \infty}{\longrightarrow} 0$ and a sequence $0<\alpha_N\ll 1$, such that
\begin{align*}
\left|M_{N,k_N}\left(x^{(1)}_r,\dots,x^{(N)}_r\right)\right|\leq \alpha_N
\end{align*}
for all $x\in \mathrm{supp}\left(\Psi_N\right)\subset \mathbb{R}^{N\times d}$ and $r\in \{1,\dots,d\}$.
\end{lem}
\begin{proof}
Let $0< \alpha_N\leq 1$ be a sequence with $\frac{\sqrt{N}}{k_N}\ll \alpha_N\ll 1$ and let $\nu_\ell:\mathbb{R}\rightarrow \mathbb{R}$, $\ell\in \mathbb{Z}$, be a family of $C^\infty$ functions with $\sum_{\ell\in \mathbb{Z}}\nu_\ell^2=1$, $\mathrm{supp}(\nu_\ell)\subset (\ell-1,\ell+1)$ and $\nu_{\ell}(x)=\nu_0(x-\ell)$. Then we define the family of functions $\chi_{\ell,r}:\mathbb{R}^{N\times d}\longrightarrow \mathbb{R}$ with $\ell\in \mathbb{Z}$ and $r\in \{1,\dots,d\}$ as
\begin{align*}
\chi_{\ell,r}\left(x\right):=\nu_\ell\left(\frac{1}{\alpha_N} M_{N,k_N}\left(x_r^{(1)},\dots,x_r^{(N)}\right)\right)
\end{align*}
and for $\ell=(\ell_1,\dots,\ell_d)\in \mathbb{Z}^d$ we define $\chi_{\ell}:=\chi_{\ell_1,1} \dots \chi_{\ell_d,d}$. First of all $\sum_{\ell\in \mathbb{Z}^d}\chi_\ell^2=\left(\sum_{\ell_1\in \mathbb{Z}}\chi_{\ell_1,1}^2\right) \dots \left(\sum_{\ell_d\in \mathbb{Z}}\chi_{\ell_d,d}^2\right)=1$. Furthermore, for any $x\in \mathbb{R}^{N\times d}$ the family of smooth functions $\{\chi_\ell:\ell\in \mathbb{Z}^d\}$ satisfies $\# \{\ell\in \mathbb{Z}^d:\chi_{\ell}(x)\neq 0\}=\# \prod_{r=1}^d \{z\in \mathbb{Z}:\chi_{z,r}(x)\neq 0\}\leq 2^d$. With the definition $C_d:=2^d C$, where $C$ is the constant from Lemma \ref{Lemma: IMS Localization}, we obtain any state $\Psi\in L^2\left(\mathbb{R}^{N\times d}\right)$
\begin{align*}
\sum_{j=1}^N\braket{T_j}_{\Psi}\!\geq\! \sum_{j=1}^N\sum_{\ell\in \mathbb{Z}^d}\braket{T_j}_{\chi_\ell \Psi}\!-\!C_d\sum_{j=1}^N\sup_{\ell\in \mathbb{Z}^d}\left\||\nabla_j \chi_\ell|^2\right\|_\infty\!\geq\! \sum_{j=1}^N\sum_{\ell\in \mathbb{Z}^d}\braket{T_j}_{\chi_\ell \Psi}-N\frac{C_d d}{\alpha_N^2 k_N^2}\|\nu'_0\|^2_{\infty},
\end{align*}
where we used the fact that $|\nabla_j \chi_\ell|^2\leq \sum_{r=1}^d|\partial_j \chi_{\ell_r,r}|^2\leq \sum_{r=1}^d\frac{1}{\alpha^2_N}\|\nu'_{\ell_r}\|^2_\infty \|\partial_{j} M_{N,k_N}\|^2_\infty$, $\|\partial_{j} M_{N,k_N}\|_\infty\leq \frac{1}{k_N}$ and $\|\nu'_z\|_\infty=\|\nu'_0\|_\infty$ for any $z\in \mathbb{Z}$. By our choice of $\alpha_N$ it is clear that $\epsilon_N:=N\frac{C_d d}{\alpha_N^2 k_N^2}\|\nu'_0\|^2_{\infty}\underset{N\rightarrow \infty}{\longrightarrow}0$. In the following let $\Phi_N$ be a sequence of states with $\braket{H_N}_{\Phi_N}-E_N\underset{N\rightarrow \infty}{\longrightarrow} 0$, and let us define $\rho_{N,\ell}:=\|\chi_\ell\Phi_N\|^2$ as well as $\Phi_{N,\ell}:=\rho_{N,\ell}^{-\frac{1}{2}}\ \chi_\ell\Phi_N$. Since $\Phi_N$ is a state, it is clear that $\sum_\ell \rho_{N,\ell}=1$. We have the estimate
\begin{align*}
\sum_{\ell\in \mathbb{Z}^d} \rho_{N,\ell} \braket{H_N}_{\Phi_{N,\ell}}\leq \sum_{j=1}^N\braket{T_j}_{\Phi_N}+\epsilon_N+\frac{1}{N-1}\sum_{i<j}\braket{v(x_i-x_j)}_{ \Phi_N}=\braket{H_N}_{\Phi_N}+\epsilon_N,
\end{align*}
and therefore there exists at least one $l\in \mathbb{Z}^d$ such that $\braket{H_N}_{\Phi_{N,\ell}}\leq \braket{H_N}_{\Phi_N}+\epsilon_N$. We can finally define $\Psi_N\left(x^{(1)},\dots,x^{(N)}\right):=\Phi_{N,\ell}\left(x^{(1)}+\xi,\dots,x^{(N)}+\xi\right)$ with $\xi:=\alpha_N \ell$. By translation-invariance of $H_N$, we have $\braket{H_N}_{\Psi_N}\leq \braket{H_N}_{\Phi_N}+\epsilon_N$ and consequently $\braket{H_N}_{\Psi_N}-E_N\underset{N\rightarrow \infty}{\longrightarrow} 0$. Furthermore, $\Psi_N\left(x^{(1)},\dots,x^{(N)}\right)\neq 0$ implies for all $r\in \{1,\dots,d\}$
\begin{align*}
\frac{1}{\alpha_N}M_{N,k_N}\left(x_r^{(1)}\!+\!\xi_r,\dots,x_r^{(N)}\!+\!\xi_r\right)=\frac{1}{\alpha_N}M_{N,k_N}\left(x_r^{(1)},\dots,x_r^{(N)}\right)\!+\!\ell_r\in \mathrm{supp}(\nu_{\ell_r}),
\end{align*}
and therefore $M_{N,k_N}\left(x_r^{(1)},\dots,x_r^{(N)}\right)\in (-\alpha_N,\alpha_N)$.
\end{proof}

Recall the inequality $-\left(\lambda T+\Lambda\right)\leq v\leq |v|\leq \Lambda(T+1)$ from Assumption \ref{Assumption: Part I}. Let us denote with $\hat{v}:=v(x-y)$ the two body multiplication operator associated to the interaction potential $v$. Due to the translation-invariance of $T$, we can promote the one body operator inequality from above to the two body operator inequality 
\begin{align*}
-\left(\lambda T+\Lambda\right)\otimes 1_{ L^2\left(\mathbb{R}^{d}\right)}\leq \hat{v}\leq |\hat{v}|\leq \Lambda(T+1)\otimes 1_{ L^2\left(\mathbb{R}^{d}\right)}.
\end{align*}
As an immediate consequence of this inequality we have the following Lemma.\\

\begin{lem}
\label{Lemma: Balance of Energy}
Given Assumption \ref{Assumption: Part I}, there exist constants $c$ and $\delta>0$ such that
\begin{align*}
\delta\sum_{j=1}^N (T_j-c)\leq H_N\leq \delta^{-1}\sum_{j=1}^N (T_j+c),
\end{align*}
as well as $\frac{1}{N-1}\sum_{i<j}|v(x_i-x_j)|\leq c\left(H_N+N\right)$.
\end{lem}

\begin{defi}
\label{Definition: Omega}
Let us define $n_{N,r,L}:\mathbb{R}^{N\times d}\longrightarrow \mathbb{R}$ as the density of particles $x^{(j)}\in \mathbb{R}^d$ that satisfy $x^{(j)}_r\geq L$, i.e. for a configuration $x=\left(x^{(1)},\dots,x^{(N)}\right)\in \mathbb{R}^{N\times d}$ with $x^{(j)}=\left(x^{(j)}_1,\dots,x^{(j)}_d\right)\in \mathbb{R}^d$ we define
\begin{align*}
n_{N,r,L}(x):=\frac{1}{N}\sum_{j=1}^N\mathds{1}_{[L,\infty)}\left(x^{(j)}_r\right).
\end{align*}
Furthermore, let $\Omega_{N,r,L,\delta}$ be the set of all $x\in \mathbb{R}^{N\times d}$ that satisfy $n_{N,r,L}(x)\geq \delta$ and $M_{N,k_N}\left(x^{(1)}_r,\dots,x^{(N)}_r\right)\leq  \xi_0$, where $k_N$ is the sequence introduced in Lemma \ref{Lemma: Localized Ground State} and $\xi_0$ is some fixed positive number. Let $E_{N,r,L,\delta}$ denote the ground state energy of $H_N$ restricted to states $\Phi$ with $\mathrm{supp}(\Phi)\subset \Omega_{N,r,L,\delta}$.
\end{defi}

\begin{lem}
\label{Lemma: Localized Energy}
Given Assumption \ref{Assumption: Part I}, there exist for all $\delta>0$ constants $\gamma_\delta>0$, $L_0(\delta)$ and $N_0(\delta)$, such that for all $r\in \{1,\dots,d\}$, $L\geq L_0(\delta)$ and $N\geq N_0(\delta)$
\begin{align}
\label{Equation: Energy in a local area}
E_{N,r,L,\delta}\geq E_N+\gamma_\delta N.
\end{align}
\end{lem}
\begin{proof}
According to Definition \ref{Definition: Omega}, for any configuration $x=\left(x^{(1)},\dots,x^{(N)}\right)\in \Omega_{N,r,L,\delta}$ there are at least $\frac{N}{2}-k_N$ particles $x^{(j)}$ such that $x^{(j)}_r\leq \xi_0$ and at least $\delta N$ particles $x^{(k)}$ such that $x^{(k)}_r\geq L$. Heuristically, this means that $\frac{N}{2}$ particles do not interact with $\delta N$ particles in case $L-\xi_0$ is large compared to the range of the interaction $v$. Since the interaction in Eq.~(\ref{Equation: Hamilton Operator}) scales like $\frac{1}{N}$, the absence of $\frac{N}{2}\times \delta N$ interaction pairs corresponds to an increase in energy of order $N$. In order to make this rigorous, i.e. in order to verify Eq.~(\ref{Equation: Energy in a local area}), we will apply the ideas of geometric localization from \cite{Lew,LNR}. In the first step, we decompose the energy $\braket{H_N}_{\Psi}$ of a state $\Psi$ into a term $E_-$ covering contributions from the left side $x^{(j)}_r\leq \xi+R$ with $\xi>\xi_0$ and $\xi+R<L$, a term $E_+$ covering contributions from the right side $x^{(j)}_r\geq \xi$ and a localization error depending on the length $R$ of the overlap $[\xi,\xi+R]$ of the two regions, which can be neglected for large separations $R\gg 1$. In the second step, we will verify that the sum of the local energies $E_- + E_+$ is indeed larger than the ground state energy $E_N$ by a contribution of order $N$, which corresponds to the observation that $E_- + E_+$ does not involve any interactions between particles on the left side and particles on the right side.\\

In the following let us fix an $r\in \{1,\dots,d\}$, and let $f_-,f_+:\mathbb{R}\longrightarrow [0,1]$ be smooth functions with $f_-^2+f_+^2=1$, $f_-(t)=1$ for $t\leq 0$ and $f_+(t)=1$ for $t\geq 1$. Then we define for $\xi\in \mathbb{R}$ and $R>0$ the functions $f_{\xi,R,\pm}:\mathbb{R}^d\longrightarrow [0,1]$ as $f_{\xi,R,\pm}(x):=f_\pm\left(\frac{x_r-\xi}{R}\right)$. This family of functions clearly satisfies $f_{\xi,R,-}(x)=1$ for $x_r\leq \xi$, $f_{\xi,R,-}(x)=0$ for $x_r\geq \xi+R$, $f_{\xi,R,+}(x)=1$ for $x_r\geq \xi+R$ and $f_{\xi,R,+}(x)=0$ for $x_r\leq \xi$. Furthermore, there exists a constant $k>0$ such that $\left|\nabla f_{\xi,R,\pm}\right|^2\leq \frac{k}{R^2}$. By Lemma \ref{Lemma: IMS Localization} we have the IMS localization formula $T\geq f_{\xi,R,-} T f_{\xi,R,-}+f_{\xi,R,+} T f_{\xi,R,+}-\frac{K}{R^2}$, $K:=2kC$. For a state $\Psi\in \bigotimes_\mathrm{s}^N L^2\left(\mathbb{R}^d\right)$, let us denote with $\gamma^{(k)}$ its reduced density matrices and with $\rho^{(k)}$ the corresponding density functions, and let us further define the localized objects $\gamma^{(k)}_{\xi,R,\pm}:=f_{\xi,R,\pm}^{\otimes^k} \gamma^{(k)} f_{\xi,R,\pm}^{\otimes^k}$ and the corresponding density functions $\rho^{(k)}_{\xi,R,\pm}(x_1,\dots,x_k):=\rho^{(k)}(x_1,\dots,x_k)f_{\xi,R,\pm}(x_1)^2\dots f_{\xi,R,\pm}(x_k)^2$. Then,
\begin{align*}
&\frac{1}{N}\braket{H_N}_\Psi=\mathrm{Tr}\left[\gamma^{(1)}\ T\right]+\frac{1}{2}\int\int \rho^{(2)}(x,y)v(x-y)\mathrm{d}x\mathrm{d}y\\
&\ \ \ =\!\mathrm{Tr}\left[\gamma^{(1)} T\right]\!+\!\frac{1}{2}\int\int\! \rho^{(2)}(x,y)\left[f_{\xi,R,-}^2\!+\!f_{\xi,R,+}^2\right](x)\left[f_{\xi,R,-}^2\!+\!f_{\xi,R,+}^2\right](y)\ v(x-y)\mathrm{d}x\mathrm{d}y\\
&\ \ \ \geq E_- + E_+ +\!\int\int\! \rho^{(2)}(x,y)f_{\xi,R,-}(x)^2f_{\xi,R,+}(y)^2 v(x\!-\!y)\mathrm{d}x\mathrm{d}y\!-\!\frac{K}{R^2},
\end{align*}
where we define
\begin{align}
\label{Equation: E left/right}
E_\pm=\mathrm{Tr}\left[\gamma^{(1)}_{\xi,R,\pm}\ T\right]+\frac{1}{2}\int\int \rho_{\xi,R,\pm}^{(2)}(x,y)v(x-y)\mathrm{d}x\mathrm{d}y.
\end{align}
 Note that we have $v_R:=\sup_{|x|\geq R} |v(x)|\underset{R\rightarrow \infty}{\longrightarrow}0$ by Assumption \ref{Assumption: Part I}, and therefore we can estimate the localization error $\left|\int\int \rho^{(2)}(x,y)f_{\xi,R,-}(x_r)^2f_{\xi,R,+}(y_r)^2 v(x-y)\right|$ by
\begin{align*}
\int\int_{[|x_r-y_r|< R]}& \rho^{(2)}(x,y)f_{\xi,R,-}(x)^2f_{\xi,R,+}(y)^2 |v(x-y)|\mathrm{d}x\mathrm{d}y+v_R\int\int \rho^{(2)}(x,y)\mathrm{d}x\mathrm{d}y\\
&\leq \int\int_{[|x_r-\xi|< R]} \rho^{(2)}(x,y)|v(x-y)|\mathrm{d}x\mathrm{d}y+v_R,
\end{align*}
where we used the fact that $x\in \mathrm{supp}\left(f_{\xi,R,-}\right)$, $y\in \mathrm{supp}\left(f_{\xi,R,+}\right)$ and $|x_r-y_r|< R$ is only possible in case $|x_r-\xi|< R$. Let us now define for $n\in \mathbb{N}$ and $m\leq n$ the points $\xi_m:=\xi_0+2Rm$. Clearly, the intervals $[|x_r-\xi_m|< R]$ are disjoint and therefore Lemma \ref{Lemma: Balance of Energy} yields
\begin{align*}
\sum_{m=1}^n\!\underset{[|x_r\!-\!\xi_m|< R]}{\int\int}\! \rho^{(2)}(x,y)|v(x\!-\!y)|\mathrm{d}x\mathrm{d}y\!\leq\! \int\int \rho^{(2)}(x,y)|v(x\!-\!y)|\mathrm{d}x\mathrm{d}y\!\leq\! \frac{2c}{N}\braket{H_N}_\Psi+2c.
\end{align*}
Hence, there exists an $m_*\leq n$ such that $\int\int_{[|x_r-\xi_{m_*}|< R]} \rho^{(2)}(x,y)|v(x-y)|\leq \frac{2c}{n N}\braket{H_N}_\Psi+\frac{2c}{n}$. We conclude that for $n\in \mathbb{N}$, there exists a $\xi$ with $\xi_0\leq \xi\leq \xi_0+2nR$ such that
\begin{align}
\frac{1+\frac{2c}{n}}{N}\braket{H_N}_\Psi&\geq E_- + E_+ -\frac{K}{R^2}-v_R-\frac{2c}{n}\label{Inequality: Localized Energy}.
\end{align}

Let us now investigate the local energy contributions $E_\pm$. As a first step, we follow the framework in \cite{LNR} and define the mixed $\ell$ particle states
\begin{align*}
G_{\ell,\pm}=\binom{N}{\ell}\mathrm{Tr}_{\ell+1\rightarrow N}\left[f_{\xi,R,\pm}^{\otimes^\ell}\otimes f_{\xi,R,\mp}^{\otimes^{N-\ell}}\ \ket{\Psi}\bra{\Psi}\  f_{\xi,R,\pm}^{\otimes^\ell}\otimes f_{\xi,R,\mp}^{\otimes^{N-\ell}}\right],
\end{align*}
where we used the notion $\mathrm{Tr}_{\ell+1\rightarrow N}\left[\ .\ \right]$ for the partial trace over the indices $\ell+1,\dots,N$. These mixed states satisfy $\mathrm{Tr}[G_{\ell,-}]=\mathrm{Tr}[G_{N-\ell,+}]$ as well as $\sum_{\ell=0}^N \mathrm{Tr}[G_{\ell,-}]=1$. Furthermore, it was shown in \cite{LNR} that we can use these mixed states to express the localized density matrices as
\begin{align}
\label{Equation: Mixed Density Matrix Hierarchy}
f_{\xi,R,\pm}^{\otimes^k}\ \gamma^{(k)}\ f_{\xi,R,\pm}^{\otimes^k}=\binom{N}{k}^{-1}\sum_{\ell=k}^N\binom{\ell}{k}G_{\ell,\pm}^{(k)},
\end{align}
where $G_{\ell,\pm}^{(k)}$ is the $k$-th reduced density matrix of $G_{\ell,\pm}$. In the following, let us assume that the state $\Psi$ satisfies $\mathrm{supp}\left(\Psi\right)\subset \Omega_{N,r,L_0,\delta}$ with $\delta>0$ and $L_0>\xi_0+R$, i.e. all $x\in \mathrm{supp}\left(\Psi\right)$ satisfy $M_{N,k_N}(x)\leq \xi_0$ and $n_{N,r,L_0}(x)\geq \delta$. The first condition $M_{N,k_N}(x)\leq \xi_0$ implies that at most $\frac{N}{2}+k_N$ indices $j$ satisfy $x^{(j)}_r> \xi_0$ and the second condition $n_{N,r,L_0}(x)\geq \delta$ is equivalent to the fact that at most $\lceil(1-\delta) N\rceil$ indices satisfy $x^{(j)}_r< L_0$. Let us denote $N_*(N):=\max\left(\frac{N}{2}+k_N,\lceil(1-\delta) N\rceil\right)$. From the support properties of $f_{\xi,R,\pm}$ we obtain for all $\xi$ with $\xi_0<\xi<L_0-R$ and $x\in \mathrm{supp}\left(\Psi\right)$, that $f_{\xi,R,+}\left(x^{(1)}\right) \dots f_{\xi,R,+}\left(x^{(\ell)}\right)=0$ for all $\ell>N_*(N)$ and $f_{\xi,R,-}\left(x^{(\ell+1)}\right) \dots f_{\xi,R,-}\left(x^{(N)}\right)=0$ for all $N-\ell>N_*(N)$. Hence, we obtain for all $\ell$ with either $\ell>N_*(N)$ or $\ell<N-N_*(N)$, and $\xi$ with $\xi_0<\xi<L_0-R$
\begin{align*}
&\binom{N}{\ell}^{-1}\mathrm{Tr}\left[G_{\ell,+}\right]=\mathrm{Tr}\left[f_{\xi,R,+}^{\otimes^\ell}\otimes f_{\xi,R,-}^{\otimes^{N-\ell}}\ \ket{\Psi}\bra{\Psi}\ f_{\xi,R,+}^{\otimes^\ell}\otimes f_{\xi,R,-}^{\otimes^{N-\ell}}\right]\\
&\ =\!\underset{\mathrm{supp}\left(\Psi\right)}{\int}\! f_{\xi,R,+}\left(x^{(1)}\right)^2\! \dots
  f_{\xi,R,+}\left(x^{(\ell)}\right)^2 f_{\xi,R,-}\left(x^{(\ell+1)}\right)^2\! \dots f_{\xi,R,-}\left(x^{(N)}\right)^2\!\left|\Psi\right|^2 \mathrm{d}x=0,
\end{align*}
and since $G_{\ell,+}\geq 0$ this implies $G_{\ell,+}=0$ for all such $\ell$. Using $\mathrm{Tr}[G_{\ell,-}]=\mathrm{Tr}[G_{N-\ell,+}]$, we also obtain $G_{\ell,-}=0$ for all $\ell$ with $\ell>N_*(N)$, respectively $\ell<N-N_*(N)$. 

Let us define rescaled versions $H_\ell^{(\lambda)}:=\sum_{j=1}^\ell T_j+\frac{1}{\ell-1}\sum_{i<j}^\ell \lambda v(x_{i}-x_{j})$ of the Hamiltonian $H_N$ and let us denote the corresponding ground state energy by $E_\ell^{(\lambda)}:=\inf \sigma\left(H_\ell^{\lambda}\right)$. Note that there exists a $\delta$-dependent $\kappa_\delta<1$ and $N_1\in \mathbb{N}$, such that $\frac{N_*(N)-1}{N-1}\leq \kappa_\delta$ for all $N\geq N_1$. Applying Eq.~(\ref{Equation: Mixed Density Matrix Hierarchy}) together with the identity $\mathrm{Tr}\left[G_{\ell,\pm}^{(1)}\ T\right]+\frac{l-1}{N-1}\frac{1}{2}\mathrm{Tr}\left[G_{\ell,\pm}^{(2)}\ \hat{v}\right]=\mathrm{Tr}\left[\frac{1}{\ell}H_\ell^{\left(\frac{\ell-1}{N-1}\right)}\ G_{\ell,\pm}\right]$ yields for all $N\geq N_1$ and $\xi$ with $\xi_0<\xi <L_0-R$
\begin{align*}
&E_{\pm}=\mathrm{Tr}\left[\gamma^{(1)}_{\xi,R,\pm}\ T\right]+\frac{1}{2}\int\int \rho_{\xi,R,\pm}^{(2)}(x,y)v(x-y)\mathrm{d}x\mathrm{d}y=\frac{1}{N}\sum_{\ell=N-N_*(N)}^{N_*(N)}\mathrm{Tr}\left[H_\ell^{\left(\frac{\ell-1}{N-1}\right)}\ G_{\ell,\pm}\right]\\
&\ \ \ \ \geq \frac{1}{N}\sum_{\ell=N-N_*(N)}^{N_*(N)} E_\ell^{\left(\frac{l-1}{N-1}\right)}\ \mathrm{Tr}\left[G_{\ell,\pm}\right]\geq \frac{1}{N}\sum_{l=N-N_*(N)}^N \kappa_\delta E_{\ell}\ \mathrm{Tr}\left[G_{\ell,\pm}\right]\\
&\ \ \ \ \geq \kappa_\delta \min_{\ell\geq N-N_*(N)}\left(\frac{1}{\ell}E_\ell\right)\frac{1}{N}\sum_{l=0}^N \ell\ \mathrm{Tr}\left[G_{\ell,\pm}\right],
\end{align*}
where we used $H_k^{(\lambda_1)}\geq \frac{\lambda_1}{\lambda_2}H_k^{(\lambda_2)}$ for all $\lambda_1\leq \lambda_2$ as well as the fact that $E_\ell=E_\ell^{(1)}<0$, which is a direct consequence of the assumption $e_\mathrm{H}<0$. Observe that
\begin{align*}
\frac{1}{N}\sum_{l=0}^N \ell\ \mathrm{Tr}\left[G_{\ell,-}\right]&+\frac{1}{N}\sum_{l=0}^N \ell\ \mathrm{Tr}\left[G_{\ell,+}\right]=\frac{1}{N}\sum_{l=0}^N \ell\ \mathrm{Tr}\left[G_{\ell,-}\right]+\frac{1}{N}\sum_{l=0}^N (N-\ell)\ \mathrm{Tr}\left[G_{\ell,-}\right]=1,
\end{align*}
and consequently we obtain for all $N\geq N_1$ and $\xi$ with $\xi_0<\xi<L_0-R$ the estimate
\begin{align}
\label{Equation: Local Energies}
E_- + E_+\geq \kappa_\delta\min_{\ell\geq N-N_*(N)} \frac{1}{\ell}E_\ell,
\end{align}
where $E_\pm$ is defined in Eq.~(\ref{Equation: E left/right}). Furthermore, Assumption \ref{Assumption: Part I} enables us to apply the results in \cite{LNR}, which tell us that $\lim_\ell\frac{1}{\ell}E_\ell=e_\mathrm{H}$, and since $N-N_*(N)\underset{N\rightarrow \infty}{\longrightarrow}\infty$, we obtain that $\min_{\ell\geq N-N_*(N)} \frac{1}{\ell}E_\ell\underset{N\rightarrow \infty}{\longrightarrow}e_\mathrm{H}$ as well. For $R>0$ and $n\in \mathbb{N}$, let us define $L_0:=\xi_0+(2n+1)R$. Combining Inequalities (\ref{Inequality: Localized Energy}) and (\ref{Equation: Local Energies}), we obtain
\begin{align*}
\liminf_{N\rightarrow \infty}\frac{1}{N}\left[\left(1+\frac{2c}{n}\right)E_{N,r,L_0,\delta}-E_N\right]\geq (\kappa_\delta-1)e_\mathrm{H}-\frac{K}{R^2}-v_R-\frac{2c}{n}.
\end{align*}
Since $\kappa_\delta<1$ and $e_\mathrm{H}<0$, we can choose $R_\delta$ and $n_\delta$ large enough, such that $\beta_\delta:=(\kappa_\delta-1)e_\mathrm{H}-\frac{K}{R_\delta^2}-v_{R_\delta}-\frac{2c}{n_\delta}>0$. With the choice $L_0(\delta):=\xi_0+(2n_\delta+1)R_\delta$ we conclude
\begin{align*}
\liminf_{N\rightarrow \infty}\frac{1}{N}\left[E_{N,r,L_0(\delta),\delta}-E_N\right]&\geq \liminf_{N\rightarrow \infty}\frac{1}{N}\left(\min\left[\left(1+\frac{2c}{n_\delta}\right)E_{N,r,L_0(\delta),\delta},\ 0\right]-E_N\right)\\
&\geq \min\left(\beta_\delta,-e_\mathrm{H}\right)>0.
\end{align*}
\end{proof}

\begin{cor}
\label{Corollary: Density Expectation}
Let Assumption \ref{Assumption: Part I} hold and $\Psi_N$ be a sequence as in Lemma \ref{Lemma: Localized Ground State}. Then,
\begin{align*}
\underset{L\rightarrow \infty}{\lim}\underset{N\rightarrow \infty}{\limsup}\braket{n_{N,r,L}}_{\Psi_N}=0
\end{align*}
for any $r\in \{1,\dots,d\}$.
\end{cor}
\begin{proof}
In the following, let $\chi:\mathbb{R}\longrightarrow [0,1]$ be a function with $\chi(x)=0$ for $x\leq 1$ and $\chi(x)=1$ for $x\geq 2$, such that $\chi$ and $\sqrt{1-\chi^2}$ are $C^\infty$. Then we define
\begin{align*}
f_{N,r,L,\delta}(x):=\chi\left(\frac{1}{\delta N}\sum_{j=1}^N \chi\left(\frac{2 x_r^{(j)}}{L}\right)\right),
\end{align*}
$g_{N,r,L,\delta}(x):=\sqrt{1-f_{N,r,L,\delta}^2}$ and $\alpha:=\|\chi'\|_\infty^2 \left(\|\chi'\|_\infty^2+\|\sqrt{1-\chi^2}'\|_\infty^2\right)$. Note that we have $\mathrm{supp}\left(f_{N,r,L,\delta}\Psi_N\right)\subset \Omega_{N,r,\frac{L}{2},\delta}$. Therefore the localization formula from Lemma \ref{Lemma: IMS Localization} and the result from Lemma \ref{Lemma: Localized Energy} tell us that there exists a $\gamma_\delta>0$ such that for all $L\geq 2L_0(\delta)$ and $N\geq N_0(\delta)$
\begin{align*}
\braket{H_N}_{\Psi_N}&\geq \braket{H_N}_{f_{N,r,L,\delta}\Psi_N}+\braket{H_N}_{g_{N,r,L,\delta}\Psi_N}-\frac{4C}{\delta^2 N L^2}\alpha\\
&\geq \left(E_N+\gamma_\delta N\right)\|f_{N,r,L,\delta}\Psi_N\|^2+E_N\left(1-\|f_{N,r,L,\delta}\Psi_N\|^2\right)-\frac{4C}{\delta^2 N L^2}\alpha.
\end{align*}
Consequently, $0\leq \|f_{N,r,L,\delta}\Psi_N\|^2\leq \frac{\braket{H_N}_{\Psi_N}-E_N+\frac{4C}{\delta^2 N L^2}\alpha}{\gamma_\delta N}\underset{N\rightarrow \infty}{\longrightarrow}0$. Furthermore, note that $x\in \mathrm{supp}\left(g_{N,r,L,\delta}\right)$ implies $n_{N,r,L}(x)\leq \frac{1}{N}\sum_{j=1}^N \chi\left(\frac{2 x_r^{(j)}}{L}\right)\leq 2\delta$ and therefore
\begin{align*}
0\leq \braket{n_{N,r,L}}_{\Psi_N}=\braket{n_{N,r,L}}_{f_{N,r,L,\delta}\Psi_N}+\braket{n_{N,r,L}}_{g_{N,r,L,\delta}\Psi_N}\leq \|f_{N,r,L,\delta}\Psi_N\|^2+2\delta\underset{N\rightarrow \infty}{\longrightarrow}2\delta
\end{align*}
for all $L\geq 2L(\delta)$. Hence $\underset{L\rightarrow \infty}{\lim}\underset{N\rightarrow \infty}{\limsup}\braket{n_{N,r,L}}_{\Psi_N}=0$.
\end{proof}

\subsection{Convergence to a Single Condensate}
\label{Subsection: Convergence to a single Condensate}

It was shown in \cite{LNR} that under quite general assumptions, including ours, on the decay and regularity of the interaction potential $v$, there exists for any sequence of states $\Phi_N$ with $\braket{H_N}_{\Phi_N}=E_N+o(N)$ a probability measure $\nu$ supported on the set of (not necessarily normed) Hartree minimizers $\{u\in \mathcal{H}: \mathcal{E}_\mathrm{H}[u]=e_\mathrm{H}(\|u\|)\}$, where $e_\mathrm{H}(s):=\inf_{\|v\|=s}\mathcal{E}_\mathrm{H}[v]$, such that a subsequence of the sequence $\gamma_{\Phi_N}^{(k)}$ converges weakly to the state $\int \big(\! \ket{u}\bra{ u}\! \big)^{\otimes^k} \ \mathrm{d}\nu(u)$ for all $k\in \mathbb{N}$, i.e.
\begin{align}
\label{Equation: Quantum de Finetti}
\mathrm{Tr}\left[\gamma_{\Phi_{N_j}}^{(k)}\ B\right]\underset{j\rightarrow \infty}{\longrightarrow}\int\mathrm{Tr}\left[\big(\! \ket{u}\bra{ u}\! \big)^{\otimes^k} B\right] \mathrm{d}\nu(u)
\end{align}
for any compact $k$ particle operator $B$. In Lemma \ref{Lemma: Strong Convergence}, we will lift this weak convergence to a strong one for the sequence of approximate ground states $\Psi_N$ constructed in Lemma \ref{Lemma: Localized Ground State}, by using the fact that mass cannot escape to infinity as a consequence of Corollary \ref{Corollary: Density Expectation}. In this context, strong convergence means that Eq.~(\ref{Equation: Quantum de Finetti}) holds for all bounded $k$ particle operators $B$, and not only compact ones. In particular, $\|u\|=1$ on the support of $\nu$.

\begin{lem}[Strong Convergence]
\label{Lemma: Strong Convergence}
Let $\Psi_N$ be the sequence from Lemma \ref{Lemma: Localized Ground State} and let $\gamma_N^{(k)}$ denote the corresponding reduced density matrices. Given Assumption \ref{Assumption: Part I}, there exists a probability measure $\mu$ supported on $\mathbb{R}^d$ and a subsequence $N_j$, such that for any bounded $k$ particle operator $B$
\begin{align*}
\mathrm{Tr}\left[\gamma_{N_j}^{(k)}\ B\right]\underset{j\rightarrow \infty}{\longrightarrow} \int_{\mathbb{R}^d} \mathrm{Tr}\left[\big(\! \ket{u_{0,t}}\bra{ u_{0,t}}\! \big)^{\otimes^k} B\right] \mathrm{d}\mu(t),
\end{align*}
where $u_{0,t}$ is defined in Assumption \ref{Assumption: Part I}.
\end{lem}
\begin{proof}
As was shown in \cite{LNR}, any sequence of approximate ground states, such as $\Psi_{N}$, has a subsequence $N_j$ that converges weakly to a convex combination of product states over Hartree minimizers, i.e. there exists a probability measure $\nu$ supported on the set of Hartree minimizers $u$ with $\|u\|\leq 1$, such that Eq.~(\ref{Equation: Quantum de Finetti}) holds for any compact $k$ particle operator $B$. As the central step of this proof, we will verify that the measure $\nu$ satisfies the identity $\int \|u\|^2\mathrm{d}\nu(u)=1$. By Corollary \ref{Corollary: Density Expectation}, we know that
\begin{align*}
\underset{L\rightarrow \infty}{\lim}\ \underset{j\rightarrow \infty}{\limsup}\ \mathrm{Tr}\left[\gamma^{(1)}_{N_j}\ \mathds{1}_{[x_r> L]}\right]=\underset{L\rightarrow \infty} {\lim}\ \underset{j\rightarrow \infty}{\limsup}\braket{n_{N_{j},r,L}}_{\Psi_{N_{j}}}=0.
\end{align*}
Since the reflected states $x\mapsto \Psi_N(-x)$ still satisfy the conditions of Corollary \ref{Corollary: Density Expectation}, we obtain $\underset{L\rightarrow \infty}{\lim}\ \underset{j\rightarrow \infty}{\limsup}\ \mathrm{Tr}\left[\gamma^{(1)}_{N_j}\ \mathds{1}_{[x_r<-L]}\right]=0$ as well. Consequently, 
\begin{align*}
\underset{L\rightarrow \infty}{\lim}\ \underset{j\rightarrow \infty}{\liminf}\ \mathrm{Tr}\left[\gamma^{(1)}_{N_j}\ \mathds{1}_{[-L,L]^d}\right]=1.
\end{align*}
Since the operator $\mathds{1}_{[-L,L]^d}$ is not compact, we cannot immediately apply the convergence (\ref{Equation: Quantum de Finetti}) for $B:=\mathds{1}_{[-L,L]^d}$. In order to obtain a convergence in a stronger sense, note that by Lemma \ref{Lemma: Balance of Energy} we have a uniform bound on the kinetic energy of $\gamma_{N_j}^{(1)}$ , i.e. there exists a constant $C<\infty$, such that
\begin{align*}
\mathrm{Tr}\left[(T+1)^{\frac{1}{2}}\ \gamma_{N_j}^{(1)}\ (T+1)^{\frac{1}{2}}\right]\leq C
\end{align*}
for all $j\in \mathbb{N}$. Since the trace class operators are the dual space of the compact operators, there exists by the Banach-Alaoglu theorem a trace class operator $\gamma$ and a subsequence, which we will still denote by $N_j$ for the sake of readability, such that for any compact one particle operator $K$
\begin{align*}
\mathrm{Tr}\left[(T+1)^{\frac{1}{2}}\ \gamma_{N_j}^{(1)}\ (T+1)^{\frac{1}{2}}\ K\right]\underset{j\rightarrow \infty}{\longrightarrow}\mathrm{Tr}\left[\gamma\ K\right].
\end{align*}
This in particular yields $\mathrm{Tr}\left[ \gamma_{N_j}^{(1)}\ B\right]\underset{j\rightarrow \infty}{\longrightarrow}\mathrm{Tr}\left[(T+1)^{-\frac{1}{2}}\ \gamma\ (T+1)^{-\frac{1}{2}}\ B\right]$ for any compact $B$, and consequently $(T+1)^{-\frac{1}{2}}\ \gamma\ (T+1)^{-\frac{1}{2}}=\int \ket{u}\bra{ u} \mathrm{d}\nu(u)$ by Eq.~(\ref{Equation: Quantum de Finetti}). Since the kinetic energy is of the form $T=t\left(i \nabla\right)$ with $t(p)\underset{|p|\rightarrow \infty}{\longrightarrow}\infty$, the operator $K:=\left(T+1\right)^{-\frac{1}{2}}\ \mathds{1}_{[-L,L]^d}\ \left(T+1\right)^{-\frac{1}{2}}$ is compact. Collecting all the results we have obtained so far yields
\begin{align*}
1&=\underset{L\rightarrow \infty}{\lim}\ \underset{j\rightarrow \infty}{\liminf}\ \mathrm{Tr}\left[\gamma^{(1)}_{N_j}\ \mathds{1}_{[-L,L]^d}\right]=\underset{L\rightarrow \infty}{\lim}\ \underset{j\rightarrow \infty}{\liminf}\ \mathrm{Tr}\left[(T+1)^{\frac{1}{2}}\ \gamma_{N_j}^{(1)}\ (T+1)^{\frac{1}{2}}\ K\right]\\
&=\underset{L\rightarrow \infty}{\lim}\mathrm{Tr}\left[(T+1)^{-\frac{1}{2}}\ \gamma\ (T+1)^{-\frac{1}{2}}\ \mathds{1}_{[-L,L]^d}\right]=\underset{L\rightarrow \infty}{\lim}\int \mathrm{Tr}\left[\ket{u}\bra{ u}\ \mathds{1}_{[-L,L]^d}\right]\mathrm{d}\nu(u)\\
&=\int \mathrm{Tr}\big[\ket{u}\bra{ u}\big]\mathrm{d}\nu(u)=\int \|u\|^2\mathrm{d}\nu(u).
\end{align*}
As an immediate consequence we obtain that $\nu$ is supported on Hartree minimizers $u$ with $\|u\|=1$. By Assumption \ref{Assumption: Part I}, we know that all such Hartree minimizers are given by $e^{i\theta}u_{0,t}$ with $\theta\in [0,2\pi)$ and $t\in \mathbb{R}^d$. Recall that $\ket{e^{i\theta}u_{0,t}}\bra{e^{i\theta}u_{0,t}}=\ket{u_{0,t}}\bra{u_{0,t}}$ defines the same density matrix for all complex phases $e^{i\theta}$. Therefore, defining the measure $\mu(A):=\nu\big(\{u_{0,t}:t\in A,\theta\in [0,2\pi)\}\big)$ yields
\begin{align}
\label{Equation: Weak Convergence with t}
\mathrm{Tr}\left[\gamma_{N_j}^{(k)}\ B\right]\underset{j\rightarrow \infty}{\longrightarrow} \int_{\mathbb{R}^d} \mathrm{Tr}\left[\big(\! \ket{u_{0,t}}\bra{ u_{0,t}}\! \big)^{\otimes^k} B\right] \mathrm{d}\mu(t)
\end{align}
for all compact operators $B$. Since $\lim_j \mathrm{Tr}\left[\gamma_{N_j}^{(1)}\right]=1=\int_{\mathbb{R}^d}\mathrm{Tr}\big[ \ket{u_{0,t}}\bra{u_{0,t}}\big]\mathrm{d}\mu(t)$, this convergence holds even in the strong sense, see \cite{We}, i.e. the convergence (\ref{Equation: Quantum de Finetti}) holds for all bounded operator $B$.
\end{proof}

\begin{lem}
\label{Lemma: Using the median}
Let $\Psi_N$ be the sequence from Lemma \ref{Lemma: Localized Ground State}. For any $\epsilon>0$ and $r\in \{1,\dots,d\}$, consider the bounded two particle operator $B_{\epsilon,r}:=\mathds{1}_{[x_r\leq \epsilon]}\ \mathds{1}_{[y_r\geq  -\epsilon]}+\mathds{1}_{[y_r\leq \epsilon]}\ \mathds{1}_{[x_r\geq  -\epsilon]}$. Then 
\begin{align*}
\underset{N\rightarrow \infty}{\liminf}\ \mathrm{Tr}\left[\gamma_N^{(2)}\ B_{\epsilon,r}\right]\geq \frac{1}{2}.
\end{align*}
\end{lem}
\begin{proof}
With the help of the function $f_{N,\epsilon,r}:=\frac{2}{N(N-1)}\sum_{i\neq j}\mathds{1}_{\left[x^{(i)}_r\leq \epsilon\right]}\mathds{1}_{\left[x^{(j)}_r\geq -\epsilon\right]}$ we have
\begin{align*}
\mathrm{Tr}\left[\gamma_N^{(2)} B_{\epsilon,r} \right]=\underset{\mathbb{R}^{N\times d}}{\int} f_{N,\epsilon,r}(x)|\Psi_N|^2 \mathrm{d}x.
\end{align*}
Let $\alpha_N$ and $k_N$ be the sequences introduced in Lemma \ref{Lemma: Localized Ground State} and let $N$ be large enough such that $\alpha_N<\epsilon$. Then, $\left|M_{N,k_N}\left(x_r^{(1)},\dots,x_r^{(N)}\right)\right|<\epsilon$ for all $x\in \mathrm{supp}\left(\Psi_N\right)$, and therefore at least $\frac{N}{2}-k_N$ particles satisfy $x_r\leq\epsilon$ and at least $\frac{N}{2}-k_N$ particles satisfy $-\epsilon\leq x_r$. Consequently
\begin{align*}
f_{N,\epsilon,r}(x)\geq \frac{2}{N(N-1)}\left(\frac{N}{2}-k_N\right)^2\underset{N\rightarrow \infty}{\longrightarrow}\frac{1}{2},
\end{align*}
and therefore $\underset{N\rightarrow \infty}{\liminf} \int_{\mathbb{R}^{N\times d}} f_{N,\epsilon,r}(x)|\Psi_N|^2\ \mathrm{d}x\geq \frac{1}{2}$.
\end{proof}

\begin{lem}
\label{Lemma: Delta Measure}
The measure $\mu$ from Lemma \ref{Lemma: Strong Convergence} is supported on $\{0\}\subset \mathbb{R}^d$, i.e. $\mu=\delta_0$.
\end{lem}
\begin{proof}
Let us define the density function $\rho(x):=|u_0(x)|^2$, as well as the marginal density function $\rho_r(x_r):=\int \rho(x)\ \mathrm{d}x_1\dots\mathrm{d}x_{r-1}\mathrm{d}x_{r+1}\dots\mathrm{d}x_d$ and the marginal measure $\mu_r(A):=\mu\left([x_r\in A]\right)$. Note that the two particle density function corresponding to $\big(\! \ket{u_{0,t}}\bra{ u_{0,t}}\!\big)^{\otimes^2}$ is given by $\rho(x-t)\rho(y-t)$, and therefore Lemmata \ref{Lemma: Strong Convergence} and \ref{Lemma: Using the median} imply
\begin{align*}
&\frac{1}{2}\leq \lim_j \mathrm{Tr}\left[\gamma_{N_j}^{(2)}\ B_{\epsilon,r}\right]=\int\limits_{\mathbb{R}^d} \mathrm{Tr}\left[\big(\! \ket{u_{0,t}}\bra{ u_{0,t}}\!\big)^{\otimes^2}B_{\epsilon,r}\right] \mathrm{d}\mu(t)\\
&=2\int\limits_\mathbb{R} \left(\int\limits_{-\infty}^{t_r+\epsilon} \rho_r(x_r) \mathrm{d}x_r\right)\left(\ \int\limits_{t_r-\epsilon}^\infty \rho_r(x_r) \mathrm{d}x_r\right) \mathrm{d}\mu_r(t_r)\\
&=2\int\limits_\mathbb{R} f_r\left(t_r+\epsilon\right)\left(1-f_r\left(t_r-\epsilon\right)\right) \ \mathrm{d}\mu_r(t_r)\underset{\epsilon\rightarrow 0}{\longrightarrow} 2\int\limits_\mathbb{R} f_r(t_r)\left(1-f_r(t_r)\right) \ \mathrm{d}\mu_r(t_r) 
\end{align*}
with the definition $f_r(s):=\int_{-\infty}^{s} \rho_r(x_r)\ \mathrm{d}x_r$, where we have used dominated convergence and continuity of $f_r$. Hence we obtain the inequality
\begin{align*}
\int_\mathbb{R} f_r(t_r)\left(1-f_r(t_r)\right) \ \mathrm{d}\mu_r(t_r)\geq \frac{1}{4}.
\end{align*}
Since the function $h(q):=q(1-q)$ is bounded by $\frac{1}{4}$ and attains its maximum only for $q=\frac{1}{2}$, we conclude $f_r(s)=\frac{1}{2}$ $\mu_r$-almost everywhere. On the other hand, by Assumption \ref{Assumption: Part I} we know that $\int_{-\infty}^s\rho_r(x_r)\ \mathrm{d}x_r=\frac{1}{2}$ if and only if $s= 0$ and therefore $f_r(s)\neq \frac{1}{2}$ for all $s\neq 0$. This together with the fact $f_r(s)=\frac{1}{2}$ $\mu_r$-almost everywhere, implies $\mu_r=\delta_0$. Since this holds for all marginal measures $\mu_r$ with $r\in \{1,\dots,d\}$, we conclude $\mu=\delta_0$.
\end{proof}

By choosing the bounded one particle operator $B$ as the projection onto the state $u_0$, Theorem \ref{Theorem: Bose--Einstein condensation of Ground States} is a direct consequence of Lemmata \ref{Lemma: Strong Convergence} and \ref{Lemma: Delta Measure}.

\section{Fock Space Formalism}
\label{Secction: Fock Space Formalism}
In order to prove Theorem \ref{Theorem: Main Theorem}, we will make use of the correspondence between the Hartree energy $\mathcal{E}_\mathrm{H}$ and the Hamiltonian $H_N$. For a rigorous treatment of this correspondence, we first need to formulate our problem in the language of second quantization. In the subsequent Definition \ref{Definition: Basic Fock} we will define the necessary formalism including the relevant Fock spaces with the corresponding creation and annihilation operators. Following \cite{LNSS}, we will use the excitation map $U_N$ in order to arrive at an operator $U_N H_N U_N^{-1}$ that only depends on modes $a_i$, $i>0$,  describing excitations, and not on the mode $a_0$ corresponding to the condensate $u_0$. The usefulness of this stems from the fact that all the modes $a_i$, $i>0$, can be thought of as being small due to Bose--Einstein condensation.\\

Before we start introducing the Fock space formalism, let us fix some notation. In the following we will repeatedly use the notation $A\cdot B$ for the composition of an operator $B:\mathcal{H}_1\longrightarrow \mathcal{H}_2$ with an operator $A:\mathcal{H}_2\longrightarrow \mathcal{H}_3$, especially when we want to stress that the involved operators map different Hilbert spaces. In order to have a consistent notation, we will occasionally write expectation values as operator products by identifying an element $u\in L^2(\mathbb{R^d})$ with a linear map $\mathbb{C}\longrightarrow L^2(\mathbb{R^d})$, e.g. we write $u^\dagger\cdot T\cdot u$ for the expectation value $\braket{T}_u$. Furthermore, recall the real orthonormal basis $u_0,u_1,\dots,u_d,u_{d+1},\dots$ from the introduction, where $u_0$ is the Hartree minimizer from Assumption \ref{Assumption: Part I} and $u_1,\dots,u_d$ form a basis of the vector space spanned by the partial derivatives $\partial_{x_1}u_0,\dots,\partial_{x_d}u_0$. Moreover, let us define the spaces
\begin{align*}
\mathcal{H}:&=L^2\big(\mathbb{R}^d\big),\\
\mathcal{H}_0:&=\{u_0\}^\perp\subset \mathcal{H}.
\end{align*}

\begin{defi}
\label{Definition: Basic Fock}
Let us denote with $a_j:=a_{u_j}$ the annihilation operator corresponding to $u_j\in \mathcal{H}$ and $\mathcal{N}_{\geq k}:=\sum_{j=k}^\infty a_j^\dagger a_j$. In the following, we will repeatedly use the Fock spaces $\mathcal{F}:=\mathcal{F}\left(\mathcal{H}\right)$, $\mathcal{F}_0:=\mathcal{F}\left(\mathcal{H}_0\right)$ and $\mathcal{F}_{\leq M}:=\mathds{1}_{[\mathcal{N}\leq M]} \mathcal{F}_0\subset \mathcal{F}_0$, where $\mathcal{N}:=\mathcal{N}_{\geq 1}$. For any $k\in \mathbb{N}_0$ we define the operator $ a_{\geq k}:\mathrm{dom}\left(\sqrt{\mathcal{N}_{\geq k}}\, \right)\longrightarrow \mathcal{F}\otimes \mathcal{H}$ as
\begin{align*}
a_{\geq k}:=\sum_{j=k}^\infty  a_j\otimes u_j,
\end{align*}
as well as the re-scaled operator $b_{\geq k}:=\sum_{j=k}^\infty  b_j\otimes u_j:=\frac{1}{\sqrt{N}}a_{\geq k}$, and the re-scaled and restricted operator $\mathbb{L}:=\frac{1}{N}\mathcal{N}\big|_{\mathcal{F}_{\leq N}}:\mathcal{F}_{\leq N}\longrightarrow \mathcal{F}_{\leq N}$, where we suppress the $N$ dependence of $b_{\geq k}$ and $\mathbb{L}$ in our notation. Furthermore, given two operators $X=\sum_{i=0}^\infty  X_i\otimes u_i:\mathrm{dom}(X)\longrightarrow \mathcal{F}\otimes \mathcal{H}$ and $Y=\sum_{i=0}^\infty  Y_i\otimes u_i:\mathrm{dom}\left(Y\right)\longrightarrow \mathcal{F}\otimes \mathcal{H}$ defined on subsets $\mathrm{dom}(X),\mathrm{dom}(Y)\subset \mathcal{F}$, we define the product operator $X\ \underline{\otimes}\ Y:\mathcal{D}\longrightarrow  \mathcal{F}\otimes \mathcal{H}\otimes \mathcal{H}$, with $\mathcal{D}:=\{\Psi\in \mathcal{F}:\sum_{i,j=0}^\infty \|X_i Y_j \Psi\|^2<\infty\}$, as
\begin{align*}
X\ \underline{\otimes}\ Y:=X\otimes 1_{\mathcal{H}}\cdot Y=\sum_{i,j=0}^\infty  \left(X_i Y_j\right)\otimes u_i\otimes u_j,
\end{align*}
where we use the convention that tensor products are performed before operator products, i.e. $X\otimes 1_{\mathcal{H}}\cdot Y:=\left(X\otimes 1_{\mathcal{H}}\right)\cdot Y$. 
\end{defi}

\begin{note}
Recall that $T$ is an operator acting on the one particle space $\mathcal{H}$ and $\hat{v}:=v(x-y)$ is an operator acting on the two particle space $\mathcal{H}\otimes \mathcal{H}$. Then, $1_\mathcal{F}\otimes T$ is an operator on $\mathcal{F}\otimes \mathcal{H}$ and $1_\mathcal{F}\otimes \hat{v}$ operates on $\mathcal{F}\otimes \mathcal{H}\otimes \mathcal{H}$. With this, we have a convenient way to express double and four fold sums of creation and annihilation operators
\begin{align*}
b_{\geq 0}^\dagger\cdot\ &1_\mathcal{F}\otimes T\cdot b_{\geq 0}=\sum_{i,j=0}^\infty T_{i,j}\ b_i^\dagger b_j,\\
\left(b_{\geq 0}\ \underline{\otimes}\ b_{\geq 0}\right)^\dagger\cdot \ &1_\mathcal{F}\otimes \hat{v}\cdot b_{\geq 0}\ \underline{\otimes}\ b_{\geq 0}=\sum_{ij,k\ell=0}^\infty\hat{v}_{ij,k\ell}\ b_i^\dagger b_j^\dagger b_k b_\ell.
\end{align*}
 In order to avoid issues with operator domains, we will define products of the form $\big(b_{\geq 0}\ \underline{\otimes}\ b_{\geq 0}\big)^\dagger\cdot 1_\mathcal{F}\otimes \hat{v}\cdot b_{\geq 0}\ \underline{\otimes}\ b_{\geq 0}$ as quadratic forms, i.e. we define the quadratic form
\begin{align*}
\Big\langle \big(b_{\geq 0}\ \underline{\otimes}\ b_{\geq 0}\big)^\dagger\cdot \ \big(1_\mathcal{F}\otimes \hat{v}\big)\cdot \big(b_{\geq 0}\ \underline{\otimes}\ b_{\geq 0}\big)\Big\rangle_\Psi:=\Big\langle 1_\mathcal{F}\otimes \hat{v}\Big\rangle_{b_{\geq 0} \underline{\otimes} b_{\geq 0}\Psi}.
\end{align*}
For the sake of readability, we will suppress the tensor with the identity in our notation, i.e. we will simply write $b_{\geq 0}^\dagger\cdot\ T\cdot b_{\geq 0}$ and $\left(b_{\geq 0}\ \underline{\otimes}\ b_{\geq 0}\right)^\dagger\cdot \ \hat{v}\cdot b_{\geq 0}\ \underline{\otimes}\ b_{\geq 0}$.\\
\end{note}

In the following, we will make use of the fact that we can express the Hamiltonian in Eq.~(\ref{Equation: Hamilton Operator}) in terms of the rescaled creation and annihilation operators as
\begin{align}
\label{Equation: Hamilton New}
N^{-1}\!H_N&=b_{\geq 0}^\dagger\cdot\ T\cdot b_{\geq 0}+\frac{N}{2(N-1)}\left(b_{\geq 0}\ \underline{\otimes}\ b_{\geq 0}\right)^\dagger\cdot \ \hat{v}\cdot b_{\geq 0}\ \underline{\otimes}\ b_{\geq 0}.
\end{align}
Since the Hamiltonian $H_N$ is only defined on the subset $\bigotimes_\mathrm{s}^N \mathcal{H}\subset \mathcal{F}$, the equation above only holds in this subspace of fixed particle number $N$. In order to focus on excitations above the condensate, we follow the strategy in \cite{LNSS} and map the Hamiltonian $H_N$ to an operator which acts on the truncated Fock space $\mathcal{F}_{\leq N}$ of modes orthogonal to $u_0$ with the help of the excitation map $U_N$. We will think of this map $U_N$ as the quantum counterpart to the embedding of the disc $\{z\in \{u_0\}^\perp:\|z\|\leq 1\}$ into the sphere $\{u\in \mathcal{H}:\|u\|=1\}$ via the map $\iota$ defined in Eq.~(\ref{Equation: Embedding}). The proof of the  following properties of $U_N$ is elementary and  is left to the reader.

\begin{lem}
\label{Lemma: Excitation Map}
Recall the definition of the operator $\mathbb{L}$ in Definition \ref{Definition: Basic Fock} and the excitation map $U_N:\bigotimes_\mathrm{s}^N \mathcal{H}\longrightarrow \mathcal{F}_{\leq N}$ from Eq.~(\ref{Equation: Excitation Map})
\begin{align*}
U_N\left(u_0^{\otimes^{i_0}}\otimes_\mathrm{s} u_1^{\otimes^{i_1}}\otimes_\mathrm{s}\dots\otimes_\mathrm{s} u_m^{\otimes^{i_m}}\right):=u_1^{\otimes^{i_1}}\otimes_\mathrm{s}\dots\otimes_\mathrm{s} u_m^{\otimes^{i_m}},
\end{align*}
for non-negative integers $i_0+\dots+i_m=N$. Under conjugation with this unitary map $U_N$, we have for all $i,j\geq 1$ the following transformation laws
\begin{align*}
U_N\, b_0^\dagger b_0 \, U_N^{-1}&=1-\mathbb{L},\\
U_N\, b_j^\dagger b_0\, U_N^{-1}&=b^\dagger_j \sqrt{1-\mathbb{L}},\\
U_N\, b_j^\dagger b_i\, U_N^{-1}&=b_j^\dagger b_i.
\end{align*}
\end{lem}

We can summarize the transformation laws from Lemma \ref{Lemma: Excitation Map} as follows: In any product of the form $b_i^\dagger b_j$ we exchange $b_0$ with the operator $\sqrt{1-\mathbb{L}}$. In analogy to this, the zero component of the embedding $\iota(z)$ defined in Eq.~(\ref{Equation: Embedding}) is given by $u_0^\dagger\cdot \iota(z)=\sqrt{1-\|z\|^2}$. In order to express $U_N  H_N U_N^{-1}$, let us first compute
\begin{align*}
U_N  &\left(b_{\geq 0}^\dagger\cdot T\cdot b_{\geq 0}\right) U_N^{-1}=U_N\, \mathfrak{Re}\left[T_{0,0}\ b_0^\dagger b_0\!+\!2\sum_{i=1}^\infty T_{i,0}\ b_i^\dagger b_0+\sum_{i,j=1}^\infty T_{i,j}\ b_i^\dagger b_j\right] U_N^{-1}\\
&=\mathfrak{Re}\left[T_{0,0}\ (1-\mathbb{L})+2\sum_{i=1}^\infty T_{i,0}\ b_i^\dagger \sqrt{1-\mathbb{L}}+\sum_{i,j=1}^\infty T_{i,j}\ b_i^\dagger b_j\right]\\
&=\mathfrak{Re}\left[u_0^\dagger\cdot  T\cdot  u_0\ (1-\mathbb{L})+2\ b_{\geq 1}^\dagger\cdot T\cdot   u_0\cdot \sqrt{1-\mathbb{L}}+b_{\geq 1}^\dagger\cdot T\cdot b_{\geq 1}\right],
\end{align*}
where the real part of an operator is defined as $\mathfrak{Re}\left[X\right]:=\frac{X+X^\dagger}{2}$. Similarly, we can express the transformed operator $U_N \left(\frac{N}{2(N-1)}\left(b_{\geq 0}\ \underline{\otimes}\ b_{\geq 0}\right)^\dagger\cdot \ \hat{v}\cdot b_{\geq 0}\ \underline{\otimes}\ b_{\geq 0}\right) U_N^{-1}$ as
\begin{align}
\label{Equation: Parameterized Hamiltonian}
\nonumber \mathfrak{Re}&\Big[\frac{1}{2}\big( u_0\ \underline{\otimes}\   u_0\big)^\dagger\cdot \hat{v}\cdot  u_0\ \underline{\otimes}\   u_0\ f_0\left(\mathbb{L}\right)+2\ \big(b_{\geq 1}\ \underline{\otimes}\   u_0\big)^\dagger\cdot \hat{v}\cdot  u_0\ \underline{\otimes}\   u_0\, f_1\left(\mathbb{L}\right)\\
\nonumber&\ \ \ \ \ \ +\big(b_{\geq 1}\ \underline{\otimes}\  b_{\geq 1}\big)^\dagger\cdot \hat{v}\cdot  u_0\ \underline{\otimes}\   u_0\, f_2\left(\mathbb{L}\right)+\big(b_{\geq 1}\ \underline{\otimes}\  u_0 \big)^\dagger\cdot \hat{v}\cdot b_{\geq 1}\ \underline{\otimes}\   u_0\, f_3\left(\mathbb{L}\right)\\
\nonumber&\ \ \ \ \ \ +\big( u_0\ \underline{\otimes}\  b_{\geq 1}\big)^\dagger\cdot \hat{v}\cdot b_{\geq 1}\ \underline{\otimes}\   u_0\, f_4\left(\mathbb{L}\right)+\!2\ \big(b_{\geq 1}\ \underline{\otimes}\  b_{\geq 1}\big)^\dagger\cdot \hat{v}\cdot b_{\geq 1}\ \underline{\otimes}\   u_0\, f_5\left(\mathbb{L}\right)\\
&\ \ \ \ \ \ +\frac{1}{2}\big(b_{\geq 1}\ \underline{\otimes}\  b_{\geq 1}\big)^\dagger\cdot \hat{v}\cdot b_{\geq 1}\ \underline{\otimes}\  b_{\geq 1}\, f_6\left(\mathbb{L}\right)\Big],
\end{align}
with $f_0(x):=\frac{N}{N-1}(1-x)(1-x-N^{-1})$, $f_1(x):=\frac{N}{N-1}(1-x-N^{-1})\sqrt{1-x}$, $f_2(x):=\frac{N}{N-1}\sqrt{1-x-N^{-1}}\sqrt{1-x}$, $f_3(x):=f_4(x):=\frac{N}{N-1}(1-x)$, $f_5(x):=\frac{N}{N-1}\sqrt{1-x}$ and $f_6(x):=\frac{N}{N-1}$. In order to keep the notation compact, let us name the essential building blocks involved in the expressions above.

\begin{defi}
\label{Definition: A and B}
We define $A_0:= u_0^\dagger\cdot  T\cdot  u_0$, $A_1:=2\ b_{\geq 1}^\dagger\cdot T\cdot   u_0$ and $A_2:=b_{\geq 1}^\dagger\cdot T\cdot b_{\geq 1}$, as well as $B_0:=\frac{1}{2}\big( u_0\ \underline{\otimes}\   u_0\big)^\dagger\cdot \hat{v}\cdot  u_0\ \underline{\otimes}\   u_0$ and
\begin{align*}
B_1:&=2\ \big(b_{\geq 1}\ \underline{\otimes}\   u_0\big)^\dagger\cdot \hat{v}\cdot  u_0\ \underline{\otimes}\   u_0,\ \ \ \ \ B_4:=\big( u_0\ \underline{\otimes}\  b_{\geq 1}\big)^\dagger\cdot \hat{v}\cdot b_{\geq 1}\ \underline{\otimes}\   u_0,\\
B_2:&=\big(b_{\geq 1}\ \underline{\otimes}\  b_{\geq 1}\big)^\dagger\cdot \hat{v}\cdot  u_0\ \underline{\otimes}\   u_0,\ \ \ \ \ \ B_5:=2\ \big(b_{\geq 1}\ \underline{\otimes}\  b_{\geq 1}\big)^\dagger\cdot \hat{v}\cdot b_{\geq 1}\ \underline{\otimes}\   u_0,\\
B_3:&=\big(b_{\geq 1}\ \underline{\otimes}\  u_0 \big)^\dagger\cdot \hat{v}\cdot b_{\geq 1}\ \underline{\otimes}\   u_0,\ \ \ \ \ \  B_6:=\frac{1}{2}\big(b_{\geq 1}\ \underline{\otimes}\  b_{\geq 1}\big)^\dagger\cdot \hat{v}\cdot b_{\geq 1}\ \underline{\otimes}\  b_{\geq 1}.
\end{align*}
\end{defi}
With these building blocks at hand, we can express the transformed Hamiltonian as
\begin{align}
\label{Equation: U representation of H}
U_N N^{-1}\!H_N U_N^{-1}= \sum_{r=0}^2\mathfrak{Re}\left[A_r \sqrt{1-\mathbb{L}}^{2-r}\right]+\sum_{r=0}^6\mathfrak{Re}\left[B_r f_r\left(\mathbb{L}\right)\right].
\end{align}

In the subsequent Lemma \ref{Lemma: Auxiliary Estimates} we will derive estimates for operator expressions of the form $B_r\, f\left(\mathbb{L}\right)$. Such estimates will be useful for the identification of lower order terms in the energy asymptotics in Eq.~(\ref{Equation: Energy asymptotics}).

\begin{lem}
\label{Lemma: Auxiliary Estimates}
Let us denote with $\pi_{M}$ the orthogonal projection onto $\mathcal{F}_{\leq M}$. Given Assumption \ref{Assumption: Part I}, there exists a constant $c$ such that for functions $f:[0,1]\longrightarrow \mathbb{R}$
\begin{align}
\label{Equation: B1 estimate}
\pm\pi_M\, \mathfrak{Re}\left[B_1 f\left(\mathbb{L}\right)\right] \pi_M\leq c \sup_{x\leq \frac{M}{N}}|f(x)|\ \sqrt{\frac{M}{N}}
\end{align}
for all $M\leq N$, and for all $t>0$ and $i\in \{2,3,4\}$ we have
\begin{align*}
\pm\pi_M\, \mathfrak{Re}\left[B_i f\left(\mathbb{L}\right)\right]\pi_M&\leq c \sup_{x\leq \frac{M}{N}}|f(x)|\ \sqrt{\frac{M}{N}}\left(t+t^{-1}\ b_{\geq 1}^\dagger \cdot (T+1)\cdot b_{\geq 1}\right),\\
\pm\pi_M\, \mathfrak{Re}\left[B_5 f\left(\mathbb{L}\right)\right] \pi_M&\leq c \sup_{x\leq \frac{M}{N}}|f(x)|\ \frac{M}{N}\left(t+t^{-1}\ b_{\geq 1}^\dagger \cdot (T+1)\cdot b_{\geq 1}\right),\\
-\frac{M}{2N}\ b^\dagger_{\geq 1}\cdot \left(\lambda\ T+\Lambda\right)\cdot b_{\geq 1} &\leq \pi_M\, \mathfrak{Re}\left[B_6\right] \pi_M\leq \frac{M}{2N}b^\dagger_{\geq 1}\cdot (\Lambda T+\Lambda)\cdot b_{\geq 1},
\end{align*}
where the constants $\lambda,\Lambda$ are as in Assumption \ref{Assumption: Part I}.
\end{lem}
\begin{proof}
Using the Cauchy--Schwarz inequality as in Lemma \ref{Lemma: O_* results} with $Q:=1_{\mathcal{F}_0}\otimes \hat{v}$, $A:=b_{\geq 1}\ \underline{\otimes}\  u_0\, \pi_M$ and $B:=2 u_{0}\ \underline{\otimes}\  u_0\, f\left(\mathbb{L}\right) \pi_M$, and defining $k:=\left(u_{0}\otimes u_0\right)^\dagger \cdot |\hat{v}|\cdot u_{0}\otimes u_0$, we obtain for any $s>0$
\begin{align*}
\pm\pi_M &\, \mathfrak{Re}\left[B_1 f\left(\mathbb{L}\right)\right] \pi_M=\pm \mathfrak{Re}\left[A^\dagger\cdot Q\cdot B\right]\leq s\ A^\dagger\cdot |Q|\cdot A+s^{-1}\ B^\dagger\cdot |Q|\cdot B\\
&=s\ \pi_M \left(b_{\geq 1}\ \underline{\otimes}\  u_0\right)^\dagger \cdot  |\hat{v}|\cdot b_{\geq 1}\ \underline{\otimes}\  u_0\, \pi_M+s^{-1}4k\, \pi_M f\left(\mathbb{L}\right)^2 \pi_M.
\end{align*}
By Assumption \ref{Assumption: Part I}, $|\hat{v}|\leq \Lambda\ 1_{\mathcal{H}}\otimes (T+1)$. Let $K:=\Lambda\ u_0^\dagger\cdot (T+1)\cdot u_0$, then
\begin{align*}
\pi_M \left(b_{\geq 1}\ \underline{\otimes}\  u_0\right)^\dagger \cdot  |\hat{v}|\cdot b_{\geq 1}\ \underline{\otimes}\  u_0\cdot \pi_M\leq K\ \pi_M\, b_{\geq 1}^\dagger\cdot b_{\geq 1}\, \pi_M\leq K\ \frac{M}{N}.
\end{align*}
Using $\pi_M f\left(\mathbb{L}\right)^2 \pi_M\leq \left(\sup_{x\leq \frac{M}{N}}|f(x)|\right)^2$ and choosing $s:=\sqrt{\frac{N}{M}}\sup_{x\leq \frac{M}{N}}|f(x)|$ yields Eq.~(\ref{Equation: B1 estimate}). The other inequalities can be derived similarly.
\end{proof}

The following two Lemmata will be useful tools in the verification of the lower bound of the energy asymptotics in Theorem \ref{Theorem: Lower Bound}.

\begin{lem}
\label{Lemma: Decomposition}
There exist constants $c,\delta>0$, such that for $N\geq 2$ 
\begin{align}
\label{Equation: Relative Bound}
\delta\ b_{\geq 1}^\dagger\cdot  T\cdot b_{\geq 1}-c \leq U_NN^{-1}\!H_NU_N^{-1} \leq c\left(b_{\geq 1}^\dagger\cdot  T\cdot b_{\geq 1}+1\right).
\end{align} 
Let us further denote with $P_n$ the orthogonal projection onto $\mathds{1}_{[\mathcal{N}=n]}\mathcal{F}_0$. Then there exists a constant $k$, such that for $N\geq 2$
\begin{align*}
\sum_{n=0}^N P_n \left(U_N N^{-1}\!H_N U_N^{-1}\right) P_n\leq k\left(U_N N^{-1}\!H_N U_N^{-1}+k\right).
\end{align*}
\end{lem}
\begin{proof}
Recall from Lemma \ref{Lemma: Balance of Energy} that $N^{-1}H_N\geq \frac{\delta}{N}\sum_{j=1}^N T_j-\delta c=\delta\ b_{\geq 0}^\dagger\cdot T\cdot b_{\geq 0}-\delta c$. Therefore we have the estimate
\begin{align*}
&U_NN^{-1}\!H_NU_N^{-1}\geq \delta \left( u_0\cdot \sqrt{1-\mathbb{L}}+b_{\geq 1}\right)^\dagger\cdot T\cdot \left( u_0\cdot \sqrt{1-\mathbb{L}}+b_{\geq 1}\right)-\delta c\\
&\geq \frac{\delta}{2}b_{\geq 1}^\dagger\cdot T\cdot b_{\geq 1}-\delta u_0^\dagger\cdot T\cdot u_0\ (1-\mathbb{L})-\delta c\geq \tilde{\delta}b_{\geq 1}^\dagger\cdot T\cdot b_{\geq 1}-\tilde{c},
\end{align*}
with $\tilde{\delta}:=\frac{\delta}{2}$ and $\tilde{c}:=\delta u_0^\dagger\cdot T\cdot u_0+\delta c$. The upper bound in Eq.~(\ref{Equation: Relative Bound}) follows analogously. In order to verify the second inequality note that the map $A\mapsto \sum_n P_n A P_n$ is monotone and $\sum_n P_n \left(b_{\geq 1}^\dagger\cdot  T\cdot b_{\geq 1}\right) P_n=b_{\geq 1}^\dagger\cdot  T\cdot b_{\geq 1}\sum_n P_n^2=b_{\geq 1}^\dagger\cdot  T\cdot b_{\geq 1}$. Hence,
\begin{align*}
&\sum_{M=0}^N P_n \left(U_N N^{-1}\!H_N U_N^{-1}\right) P_n\!\leq\! \sum_{M=0}^N P_n \left(c\ b_{\geq 1}^\dagger\cdot  T\cdot b_{\geq 1}\!+c\right) P_n\\
&\ \ =c\ b_{\geq 1}^\dagger\cdot  T\cdot b_{\geq 1}+c\leq \delta^{-1}c\ U_N N^{-1}\!H_N U_N^{-1}+(c+\delta^{-1}c^2).
\end{align*}
\end{proof}

In the subsequent Lemma we are going to verify that we can exchange the $N$-dependent functions $f_i$ in Eq.~(\ref{Equation: U representation of H}) with $N$-independent functions $\sqrt{1-x}^{\ \beta_i}$, for suitable $\beta_i$, without changing the operator substantially. This will be convenient in the lower bound of the energy asymptotics, since there we have to verify an operator Taylor approximation, which will be more convenient to do for the functions $\sqrt{1-x}^{\ \beta_i}$ than for the functions $f_i$.
\begin{lem}
\label{Lemma: Estimates Lower Bound}
Let $\beta_0:=4,\beta_1:=3$, $\beta_2:=\beta_3:=\beta_4:=2$, $\beta_5:=1$ and $\beta_6:=0$, and let us define the operators $\widetilde{A}_N$ and $\widetilde{B}_N$ acting on $\mathcal{F}_0$ as
\begin{align}
\label{Equation: tilde A}
\widetilde{A}_N:&=\sum_{r=0}^2\mathfrak{Re}\left[A_r \sqrt{1-\mathbb{L}}^{2-r}\right],\\
\label{Equation: tilde B}
\widetilde{B}_N:&=\sum_{r=0}^6\mathfrak{Re}\left[B_r \sqrt{1-\mathbb{L}}^{\beta_r}\right].
\end{align}
Then, given Assumption \ref{Assumption: Part I}, there exists a constant $K$ such that for all $M\leq N$
\begin{align}
\label{Equation: Corollary: Estimates Lower Bound - first line}
\pm \pi_M \left(U_N N^{-1}\!H_N U_N^{-1}\! -\!\widetilde{A}_N\!-\!\widetilde{B}_N\right)\pi_M\leq \frac{C}{N}\sqrt{\frac{M}{N}} \left(b_{\geq 1}^\dagger\cdot  T\cdot b_{\geq 1}+1\right).
\end{align}
\end{lem}
\begin{proof}
According to Eq.~(\ref{Equation: U representation of H}), we have 
\begin{align}
\label{Equation: Tranformed Operator}
U_N N^{-1}\! H_N U_N^{-1}-\widetilde{A}_N-\widetilde{B}_N=\sum_{r=0}^6 \mathfrak{Re}\left[B_r \left(f_r\left(\mathbb{L}\right)-\sqrt{1-\mathbb{L}}^{\beta_r}\right)\right],
\end{align}
with the functions $f_0,\dots,f_6$ from Eq.~(\ref{Equation: Parameterized Hamiltonian}). Note that for all $N\geq 2$
\begin{align*}
\pm& \pi_M B_0 \left(f_0(\mathbb{L})\!-\!(1-\mathbb{L})^2\right)\pi_M=\pm \frac{1}{2}\hat{v}_{00,00}\ \pi_M \left(f_0(\mathbb{L})\!-\!(1-\mathbb{L})^2\right) \pi_M\\
&\leq \frac{1}{2}|\hat{v}_{00,00}|\ \sup_{x\leq \frac{M}{N}}|f_0(x)-(1-x)^2|\leq \frac{1}{2}|\hat{v}_{00,00}|\frac{M}{(N-1)N}.
\end{align*}
Furthermore, $f_r\left(x\right)=\sqrt{1-x}^{\beta_r}+O\left(\frac{1}{N}\right)$ and therefore we obtain with Lemma \ref{Lemma: Auxiliary Estimates} and the choice $t=1$
\begin{align*}
\pm \pi_M B_r \left(f_r(\mathbb{L})-\sqrt{1-x}^{\ \beta_r}\right)\pi_M\leq \frac{C}{N}\sqrt{\frac{M}{N}} \left(b_{\geq 1}^\dagger\cdot  T\cdot b_{\geq 1}+1\right),
\end{align*}
for a constant $C$ and $r\in \{1,\dots,6\}$. 
\end{proof}

\section{Asymptotics of the Ground State Energy}
\label{Section: Energy Asymptotic}
We start by making the formal definition of the Bogoliubov Hamiltonian $\mathbb{H}$ in Eq.~(\ref{Equation: Bogoliubov}) rigorous in Subsection \ref{Subsection: Construction of the Bogoliubov Operator}. In the following Subsection \ref{Subsection: Upper Bound}, we will verify the upper bound in the energy asymptotics in Eq.~(\ref{Equation: Energy asymptotics}). We will then discuss the proof of the lower bound in Subsection \ref{Subsection: Lower Bound}, while the verification of the main technical Theorem \ref{Theorem: Decomposition} for the lower bound will be postponed to Section \ref{Section: Results in the transformed picture}.\\

\subsection{Construction of the Bogoliubov Operator $\mathbb{H}$}
\label{Subsection: Construction of the Bogoliubov Operator}
In the following Lemma \ref{Lemma: Hessian} we will identify the Hessian $\mathrm{Hess}|_{u_0}\mathcal{E}_\mathrm{H}$, and give a precise definition of the Bogoliubov operator in the subsequent Definition \ref{Definition: Bogoliubov Operator}. Furthermore, we shall see that the operator $\mathbb{H}$ is indeed semi-bounded. In the following let us denote with $\mathrm{dom}\left[A\right]:=\mathrm{dom}\big(\sqrt{A}\, \big)$ the form domain of an operator $A\geq 0$.

\begin{lem}
\label{Lemma: Hessian}
Given Assumption \ref{Assumption: Part I}, the Hessian of the Hartree energy $\mathcal{E}_\mathrm{H}$ at the Hartree minimizer $u_0$ is given by
\begin{align}
\label{Equation: Formula for the Hessian}
\frac{1}{2}\mathrm{Hess}|_{u_0}\mathcal{E}_\mathrm{H}[z]=z^\dagger\cdot Q_\mathrm{H}\cdot z+G_\mathrm{H}^\dagger\cdot z\otimes z+\left(z\otimes z\right)^\dagger\cdot G_\mathrm{H},
\end{align}
where $G_\mathrm{H}:=\frac{1}{2}\hat{v}\cdot u_0\otimes u_0\in \overline{\mathcal{H}_0\otimes_\mathrm{s} \mathcal{H}_0}^{\|.\|_*}$ is in the closure of $\mathcal{H}_0\otimes_\mathrm{s} \mathcal{H}_0$ with respect to the norm $\|G\|_*:=\|1_\mathcal{H}\otimes (T+1)^{-\frac{1}{2}}\cdot G\|$, and the operator $Q_\mathrm{H}$ is defined by the equation
\begin{align*}
z^\dagger\cdot Q_\mathrm{H}\cdot z:&=z^\dagger\cdot T\cdot z+\left(z\otimes u_0\right)^\dagger\cdot \hat{v}\cdot z\otimes u_0-\mu_\mathrm{H}\, z^\dagger\cdot z+\left(u_0\otimes z\right)^\dagger\cdot \hat{v}\cdot z\otimes u_0
\end{align*}
for all $z\in \mathcal{H}_0\cap \mathrm{dom}\left[T\right]$, with $\mu_\mathrm{H}:=u_0^\dagger\cdot T\cdot u_0+(u_0\otimes u_0)^\dagger\cdot \hat{v}\cdot u_0\otimes u_0$. Furthermore, $Q_\mathrm{H}$ is non-negative and satisfies $\nu^{-1}(T|_{\mathcal{H}_0}+1)\leq Q_\mathrm{H}+1\leq \nu (T|_{\mathcal{H}_0}+1)$ for some constant $\nu>0$.\\
\end{lem}
\begin{note}
By Assumption \ref{Assumption: Part I}, we know that $\hat{v}\cdot u_0\otimes u_0\in \overline{\mathcal{H}_0\otimes_\mathrm{s} \mathcal{H}_0}^{\|.\|_*}$, which follows from the fact that $1_\mathcal{H}\otimes (T+1)^{-\frac{1}{2}}\cdot \hat{v}\cdot 1_\mathcal{H}\otimes (T+1)^{-\frac{1}{2}}$ is a bounded operator and that $u_0\in \mathrm{dom}\left[T\right]$. For such elements $G\in \overline{\mathcal{H}_0\otimes_\mathrm{s} \mathcal{H}_0}^{\|.\|_*}$, we have that $G_\mathrm{reg}:=1_\mathcal{H}\otimes (T+1)^{-\frac{1}{2}}\cdot G$ is an element of $\mathcal{H}_0\otimes_\mathrm{s} \mathcal{H}_0$ and therefore we can define for all $z\in  \mathrm{dom}\left[T\right]$
\begin{align*}
G^\dagger\cdot z\otimes z:=G_\mathrm{reg}^\dagger\cdot z\otimes \left((T+1)^{\frac{1}{2}}\cdot z\right).
\end{align*}
In a similar fashion, we define the operator $G^\dagger\cdot b_{\geq 1}\underline{\otimes}b_{\geq 1}:=G_\mathrm{reg}^\dagger\cdot b_{\geq 1}\underline{\otimes}\left((T+1)^{\frac{1}{2}}\cdot b_{\geq 1}\right)$.
\end{note}

\begin{proof}[Proof of Lemma \ref{Lemma: Hessian}]
With the help of the embedding $\iota$ defined in Eq.~(\ref{Equation: Embedding}), we can express the Hessian as $\mathrm{Hess}|_{u_0}\mathcal{E}_\mathrm{H}[z]=D^2|_0 \left(\mathcal{E}_\mathrm{H}\circ \iota\right)(z)$, where $D^2|_{z_0}f(z)$ denotes the second derivative of a function $f$ in the direction $z$ evaluated at $z_0$. An explicit computation yields Eq.~(\ref{Equation: Formula for the Hessian}). Regarding the second part of the Lemma, observe that $Q_\mathrm{H}\geq0$ follows from the fact that we can always find a phase $\theta_z$ such that
\begin{align*}
z^\dagger\cdot Q_\mathrm{H}\cdot z=\frac{1}{2}\mathrm{Hess}|_{u_0}\mathcal{E}_\mathrm{H}[e^{i\theta_z}z]\geq 0.
\end{align*}
Furthermore, note that $|v|\leq \Lambda (T+1)$ implies $\pm \left(1_{\mathcal{H}_0}\otimes u_0\right)^\dagger\cdot \hat{v}\cdot 1_{\mathcal{H}_0}\otimes u_0\leq  c\ 1_{\mathcal{H}_0}$ with $c:=u_0^\dagger\cdot \Lambda(T+1)\cdot u_0$ and
\begin{align*}
\pm\! \left(u_0\!\otimes\! 1_{\mathcal{H}_0}\right)^\dagger\!\cdot\! \hat{v}\!\cdot \!1_{\mathcal{H}_0}\!\otimes\! u_0&\leq\! \frac{1}{2}\left(u_0\!\otimes \!1_{\mathcal{H}_0}\right)^\dagger\!\cdot\! |\hat{v}|\!\cdot\! u_0\!\otimes \!1_{\mathcal{H}_0}\!+\!\frac{1}{2}\left( 1_{\mathcal{H}_0}\!\otimes\! u_0\right)^\dagger\!\cdot\! |\hat{v}|\!\cdot\! 1_{\mathcal{H}_0}\!\otimes\! u_0 \leq\! c\ 1_{\mathcal{H}_0}.
\end{align*}
Hence $Q_\mathrm{H}\geq 0$ implies $Q_\mathrm{H}+1\geq T|_{\mathcal{H}_0}+1-(2c+|\mu|+1)\geq T|_{\mathcal{H}_0}+1-(1+2c+\mu)(Q_\mathrm{H}+1)$, and therefore $(2+2c+\mu)(Q_\mathrm{H}+1)\geq T|_{\mathcal{H}_0}+1$. Furthermore $T\geq 0$ implies
\begin{align*}
Q_\mathrm{H}+1\leq T+2c+|\mu|\leq (1+2c+|\mu|)(T|_{\mathcal{H}_0}+1).
\end{align*}
\end{proof}

\begin{defi}
\label{Definition: Bogoliubov Operator}
Let the selfadjoint operator $Q_\mathrm{H}$ and $G_\mathrm{H}\in \overline{\mathcal{H}_0\otimes_\mathrm{s} \mathcal{H}_0}^{\|.\|_*}$ be as in Lemma \ref{Lemma: Hessian}. Then we define the Bogoliubov operator $\mathbb{H}$ as
\begin{align}
\label{Equation: Definition Bogoliubov}
\mathbb{H}:= a_{\geq 1}^\dagger\cdot  Q_\mathrm{H}\cdot  a_{\geq 1}+G_\mathrm{H}^\dagger\cdot  a_{\geq 1}\ \underline{\otimes}\  a_{\geq 1}+\left( a_{\geq 1}\ \underline{\otimes}\  a_{\geq 1}\right)^\dagger\cdot  G_\mathrm{H}.
\end{align}
\end{defi}

\begin{theorem}
\label{Theorem: Bogoliubov Ground State Energy}
The quadratic form on the right side of Eq.~(\ref{Equation: Definition Bogoliubov}) is semi-bounded from below and closeable, and consequently defines by Friedrichs extension a selfadjoint operator $\mathbb{H}$ with $\inf \sigma\left(\mathbb{H}\right)>-\infty$. Furthermore there exists a sequence of states $\Psi_M\in \mathrm{dom}\left[a_{\geq 1}^\dagger\cdot (T+1)\cdot a_{\geq 1}\right]\cap \mathcal{F}_{\leq M}$, $\|\Psi_M\|=1$, such that 
\begin{align*}
\braket{\mathbb{H}}_{\Psi_M}\underset{M\rightarrow \infty}{\longrightarrow}\inf \sigma\left(\mathbb{H}\right).
\end{align*}
 Additionally there exists a constant $r_*>0$ such that for all $r<r_*$ the operator $\mathbb{H}-r\mathbb{A}$ satisfies $\inf \sigma\left(\mathbb{H}-r\mathbb{A}\right)>-\infty$ as well, where
\begin{align}
\label{Equation: NT}
\mathbb{A}:=-\frac{1}{4}\sum_{j=1}^d \left(a_j-a_j^\dagger\right)^2+ a_{>d}^\dagger\cdot (T+1)\cdot  a_{>d}.
\end{align}
\end{theorem}
The proof of Theorem \ref{Theorem: Bogoliubov Ground State Energy} is being carried out in Appendix \ref{Appendix A}. We emphasize that $\mathbb{H}$ is degenerate, in the sense that $z^\dagger\cdot  Q_\mathrm{H}\cdot  z+G_\mathrm{H}^\dagger \cdot  z \otimes  z+\left(z \otimes  z\right)^\dagger\cdot  G_\mathrm{H}=0$ for any $z$ in the vector space spanned by $\{u_1,\dots,u_d\}$, and therefore we cannot directly apply the results in \cite{NNS}. We also  note that the semi-boundedness of Bogoliubov operators with degeneracies has been verified in \cite{KM} under the additional assumption that $Q_\mathrm{H}$ is bounded.

\subsection{Upper Bound}
\label{Subsection: Upper Bound}
With the essential definitions at hand, we will derive the upper bound in Theorem \ref{Theorem: Upper Bound} using the representation of $U_N H_N U_N^{-1}$ derived in the previous section. We follow the strategy presented in \cite{LNSS}, by sorting the operator $U_N H_N U_N^{-1}$ in terms of different powers in $b_{\geq 1}$ and identifying the zero component as the Hartree energy $N e_\mathrm{H}$ defined in Eq.~(\ref{Equation: Hartree Energy}) and the second order component as the Bogoliubov operator $\mathbb{H}$ defined in Eq.~(\ref{Equation: Definition Bogoliubov}).\\

\begin{lem}
\label{Lemma: Upper Bound Bogoliubov}
Let Assumption \ref{Assumption: Part I} hold. Then there exists a constant $C$ such that
\begin{align*}
\pm \pi_{M} \left(U_N N^{-1}\!H_N U_N^{-1}-e_\mathrm{H}-N^{-1}\mathbb{H}\right) \pi_{M}\leq C\ \left(\frac{M}{N}\right)^{\frac{3}{2}}\ \left(1+ a_{\geq 1}^\dagger\cdot  (T+1)\cdot  a_{\geq 1}\right)
\end{align*}
for all $M\leq N$.
\end{lem}
While Lemma \ref{Lemma: Upper Bound Bogoliubov} will be useful for proving the upper bound in Theorem \ref{Theorem: Upper Bound}, it is insufficient for proving the corresponding lower bound. This is due to the fact that Bose--Einstein condensation only provides the rough a priori information $M=o(N)$, see also the proof of Theorem \ref{Theorem: Lower Bound}.

\begin{proof}
Observe that $u_0$ minimizes the Hartree energy, and therefore
\begin{align*}
e_\mathrm{H}=\inf_{\|u\|=1} \mathcal{E}_\mathrm{H}[u]=u_0^\dagger\cdot T\cdot u_0+\frac{1}{2} \left(u_0\otimes u_0\right)^\dagger\cdot \hat{v}\cdot u_0\otimes u_0=A_0+B_0,
\end{align*}
where $A_i$ and $B_i$ are defined in Definition \ref{Definition: A and B}. Since $\mathcal{E}_\mathrm{H}[u_0]\leq \mathcal{E}_\mathrm{H}[u]$ for $\|u\|=1$, we obtain by differentiation in any direction $z\perp u_0$
\begin{align*}
0=D|_{u_0} \mathcal{E}_\mathrm{H}(z)=u_0^\dagger\cdot T\cdot z+z^\dagger\cdot T\cdot u_0+(z\otimes u_0)^\dagger\cdot \hat{v}\cdot u_0\otimes u_0+(u_0\otimes u_0)^\dagger\cdot \hat{v}\cdot z\otimes u_0,
\end{align*}
and consequently $u_j^\dagger \cdot T\cdot u_0+\left(u_j\otimes u_0\right)^\dagger\cdot \hat{v}\cdot u_0\otimes u_0=0$ for all $j\geq 1$. Hence,
\begin{align*}
A_1+B_1=2\left(b_{\geq 1}^\dagger\cdot  T\cdot  u_0+\left(b_{\geq 1}\ \underline{\otimes}\  u_0\right)^\dagger\cdot \hat{v}\cdot  u_0\ \underline{\otimes}\  u_0\right)=0.
\end{align*}
By Definition \ref{Definition: Bogoliubov Operator} and Lemma \ref{Lemma: Hessian}, we have
\begin{align*}
N^{-1}\mathbb{H}=\mathfrak{Re}\left[A_2+B_2+B_3+B_4-\mu_\mathrm{H}\ b_{\geq 1}^\dagger\cdot b_{\geq 1}\right],
\end{align*}
and consequently we can write for any $M\leq N$, using Eq.~(\ref{Equation: U representation of H}),
\begin{align*}
\pi_{M} \left(U_N N^{-1}\!H_N U_N^{-1}-e_\mathrm{H}-N^{-1}\mathbb{H}\right)\pi_{M}=\pi_{M} \mathfrak{Re}\left[X\right] \pi_{M}
\end{align*}
 with
\begin{align*}
X:=B_0 \left(f_0\left(\mathbb{L}\right)\!-\!1\!+\!2\mathbb{L}\right)\!+\!B_1\left(f_1\left(\mathbb{L}\right)\!-\!\sqrt{1\!-\!\mathbb{L}}\right)\!+\!\sum_{r=2}^4B_r \left(f_r\left(\mathbb{L}\right)\!-\!1\right)\!+\!\sum_{r=5}^6 B_r f_r\left(\mathbb{L}\right),
\end{align*}
where we used $A_1 \sqrt{1-\mathbb{L}}=-B_1 \sqrt{1-\mathbb{L}}$. In order to estimate the first contribution, note that $|f_0(x)-1+2x|\leq 2\frac{M^2}{(N-1)N}$ for all $0\leq x\leq \frac{M}{N}$ and therefore
\begin{align*}
\pm \pi_M B_0 \left(f_0\left(\mathbb{L}\right)-1+2\mathbb{L}\right) \pi_M=\pm \frac{1}{2}\hat{v}_{00,00} \pi_M\left(f_0\left(\mathbb{L}\right)-1+2\mathbb{L}\right)\pi_M\leq |\hat{v}_{00,00}|\frac{M^2}{(N-1)N}.
\end{align*}
Recalling that $b_{\geq 1}=\frac{1}{\sqrt{N}}a_{\geq 1}$ and Lemma \ref{Lemma: Auxiliary Estimates} yields 
\begin{align*}
\pm \pi_M B_1\left(f_1\left(\mathbb{L}\right)-\sqrt{1-\mathbb{L}}\right)\pi_M\leq c\frac{M}{N-1}\sqrt{\frac{M}{N}}.
\end{align*}
Furthermore we obtain for $r\in \{2,3,4\}$ by Lemma \ref{Lemma: Auxiliary Estimates} with the choice $t=\frac{1}{\sqrt{N}}$, together with the bound $\sup_{x\leq \frac{M}{N}}|f_r(x)-1|\leq C\frac{M}{N}$ for a constant $C>0$,
\begin{align*}
\pm \pi_M B_r \left(f_r\left(\mathbb{L}\right)-1\right)\pi_M\leq c\, C \frac{M^{\frac{3}{2}}}{N^2}\left(1+ a_{\geq 1}^\dagger\cdot  (T+1)\cdot  a_{\geq 1}\right).
\end{align*}
The estimates for $B_5 f_5\left(\mathbb{L}\right)$ and $B_6 f_6\left(\mathbb{L}\right)$ can be obtained analogously.
\end{proof}
\begin{theorem}[Upper Bound]
\label{Theorem: Upper Bound}
Let $E_N$ be the ground state energy of $H_N$, $e_\mathrm{H}$ the Hartree energy defined in Eq.~(\ref{Equation: Hartree Energy}) and let $\mathbb{H}$ be the Bogoliubov operator defined in Eq.~(\ref{Equation: Definition Bogoliubov}). Given Assumption \ref{Assumption: Part I}, we have the upper bound
\begin{align*}
E_N\leq N\ e_\mathrm{H}+\inf \sigma\left(\mathbb{H}\right)+o_{N\rightarrow \infty}\left(1\right).
\end{align*}
\end{theorem}
\begin{proof}
Let $\nu$ be the constant from Lemma \ref{Lemma: Hessian}, such that the inequality $Q_\mathrm{H}+1\leq \nu\left(T|_{\mathcal{H}_0}+1\right)$ holds. For all $\epsilon>0$, we know by Theorem \ref{Theorem: Bogoliubov Ground State Energy} that there exists a state $\Psi\in \mathcal{F}_M$ with $M<\infty$ such that $\kappa:=\braket{ a_{\geq 1}^\dagger\cdot  (T+1)\cdot  a_{\geq 1}}_\Psi<\infty$ and $\braket{\mathbb{H}}_\Psi\leq \inf \sigma\left(\mathbb{H}\right)+\epsilon$. Applying Lemma \ref{Lemma: Upper Bound Bogoliubov} yields the estimate
\begin{align*}
\braket{H_N}_{U_N^{-1}\Psi}&\leq N\ e_\mathrm{H}+\braket{\mathbb{H}}_\Psi+C\ M\sqrt{\frac{M}{N}}\ \left(1+\braket{ a_{\geq 1}^\dagger\cdot  (T+1)\cdot  a_{\geq 1}}_\Psi\right)\\
&\leq N\ e_\mathrm{H}+\inf \sigma\left(\mathbb{H}\right)+\epsilon+C\ M\sqrt{\frac{M}{N}}\ \left(1+\kappa\right).
\end{align*}
\end{proof}

\subsection{Lower Bound}
\label{Subsection: Lower Bound}
In the following, we will give the proof of the lower bound in the energy asymptotics in Eq.~(\ref{Equation: Energy asymptotics}). First of all let us define the operators $q,p:\mathrm{dom}\left[\mathcal{N}\right]\longrightarrow \mathcal{F}_0\otimes  \mathcal{H}_0$ as
\begin{align}
\label{Equation: Definition q}q:=\sum_{j=1}^d q_j\otimes u_j:=\frac{1}{2}\sum_{j=1}^d   \left(b_j+b_j^\dagger\right)\otimes u_j,\\
\label{Equation: Definition p}p:=\!\sum_{j=1}^d  p_j\otimes u_j:=\frac{1}{2i}\sum_{j=1}^d  \left(b_j-b_j^\dagger\right)\otimes u_j,
\end{align}
which satisfy the commutation relations $[p_k,q_\ell]=\frac{1}{2iN}\delta_{k,\ell}$. Recall that due to the translation-invariance of $\mathcal{E}_\mathrm{H}$, the Hessian $\mathrm{Hess}|_{u_0}\mathcal{E}_\mathrm{H}$ is degenerate on the real subspace $\{\sum_{j=1}^d t_ju_j: t_j\in \mathbb{\mathbb{R}}\}$. Therefore the Bogoliubov operator $\mathbb{H}$, which we have defined in Eq.~(\ref{Equation: Definition Bogoliubov}) as the second quantization of the Hessian $\mathrm{Hess}|_{u_0}\mathcal{E}_\mathrm{H}$, is degenerate with respect to the operator $q$, i.e. it can be expressed only in terms of $p$, $b_{>d}$ and $b_{>d}^\dagger$. Due to this degeneracy, we cannot directly apply the strategy pursued in \cite{LNSS} where the residuum of the Bogoliubov approximation is being estimated by the Bogoliubov operator itself. The problem is that the residuum $U_N H_N U_N^{-1}-N e_\mathrm{H}-\mathbb{H}$ includes contributions depending significantly on the modes $q_j$, like $q_j^3$, which we cannot compare with the Bogoliubov operator $\mathbb{H}$ due to its degeneracy. Furthermore, it is insufficient to compare the residuum with the (rescaled) particle number operator $\frac{1}{N}\mathcal{N}$, which  indeed dominates terms like $q_j^3$, since we only have the a priori information $\braket{\mathcal{N}}_{U_N \Psi_N}=o(N)$ provided by Bose--Einstein condensation. The novel idea of this Subsection and the subsequent Section \ref{Section: Results in the transformed picture} is to apply a further unitary transformation $\mathcal{W}_N$ such that the residuum $\mathcal{W}_N U_N H_N U_N^{-1} \mathcal{W}_N^{-1}-Ne_\mathrm{H}-\mathbb{H}$ no longer includes this kind of contributions and consequently we can compare the residuum with the Bogoliubov operator $\mathbb{H}$. This leads to the important inequality in Eq.~(\ref{Inequality: Non-Commutative}). As a consequence we observe that, in contrast to the particle number operator $\mathcal{N}$, the Bogoliubov operator satisfies $\braket{\mathbb{H}}_{U_N \Psi_N}=O(1)$, which, a posteriori, justifies estimating the residuum by the Bogoliubov operator.\\

Before we are going to construct a unitary map $\mathcal{W}_N$  satisfying Eq.~(\ref{Inequality: Non-Commutative}), we are solving the corresponding problem on a classical level, i.e. we are going to construct a map $F$ which satisfies Eq.~(\ref{Inequality: Commutative}). We will then define $\mathcal{W}_N$ as the quantum counterpart to $F$. 

\begin{defi}
\label{Definition: F}
For any $y\in \mathbb{R}^d$, let us recall the functions $u_{0,y}(x):=u_{0}(x-y)$ defined in Assumption \ref{Assumption: Part I} and let us define the map $\lambda:\mathbb{R}^d\longrightarrow \mathbb{R}^d$
\begin{align*}
\lambda(y):=\left(u_j^\dagger\cdot u_{0,y}\right)_{j=1}^d\in \mathbb{R}^d.
\end{align*}
Note that $u_j$ and $u_{0,y}$ are real-valued functions, and therefore $\lambda$ is indeed $\mathbb{R}^d$-valued. Since $y\mapsto u_{0,y}$ is a $C^2\left(\mathbb{R}^d,\mathcal{H}\right)$ function by Assumption \ref{Assumption: Part II}, $D_y \lambda(0)$ has full rank and $\lambda(0)=0$, there exists a local inverse $\lambda^{-1}:B_{2\delta}(0)\longrightarrow \mathbb{R}^d$ for $\delta>0$ small enough, where $B_r(0)\subset \mathbb{R}^d$ denotes the ball of radius $r$ centered around the origin. Let $0\leq \sigma\leq 1$ be a smooth function with $\sigma|_{B_{\delta}(0)}=1$ and $\mathrm{supp}(\sigma)\subset B_{2\delta}(0)$. Then we define the function $f:\mathbb{R}^d\longrightarrow \mathcal{H}$
\begin{align}
\label{Equation: Definition f}
f(t):&=\sigma(t)\left[u_{0,\lambda^{-1}(t)}-\left(u_0^\dagger\cdot u_{0,\lambda^{-1}(t)}\right) u_0-\sum_{j=1}^d t_j\ u_j\right]=\sum_{j=d+1}^\infty f_j(t)u_j,
\end{align}
with $f_j(t):=\sigma(t)u_j^\dagger\cdot u_{0,\lambda^{-1}(t)}$. Note that $t\mapsto f(t)$ is a $C^2\left(\mathbb{R}^d,\mathcal{H}_0\right)$ function, due to the regularity of $y\mapsto u_{0,y}$. Furthermore, $f(0)=0$. We can now define the map $F:\mathcal{H}_0\longrightarrow \mathcal{H}_0$ for all $z=\sum_{j=1}^d \left(t_j+is_j\right)u_j+z_{>d}\in \mathcal{H}_0$ with $t,s\in \mathbb{R}^d$ and $z_{>d}\in \{u_1,\dots,u_d\}^\perp$ as
\begin{align}
\label{Equation: Definition F}
F\left(z\right):=\sum_{j=1}^d\left(t_j+i s'_j\right) u_j+z_{>d}+f(t),
\end{align}
where $s'_j:=s_j-\mathfrak{Im}\left[\partial_jf(t)^\dagger\cdot z_{>d}\right]$. \\
\end{defi}

The essential property of $F$ is that $\iota\circ F$, where $\iota$ is the embedding defined in Eq.~(\ref{Equation: Embedding}), maps the set $\{\sum_{j=1}^d t_j u_j:|t|<\delta\}$ into the set of Hartree minimizers
\begin{align*}
\iota\circ F\left(\sum_{j=1}^d t_j u_j\right)=\iota\left(\sum_{j=1}^d t_j u_j+f(t)\right)=u_{0,\lambda^{-1}(t)},
\end{align*}
for all $|t|<\delta$. This also implies the central inequality Eq.~(\ref{Inequality: Commutative}), as will be demonstrated in the introduction of Section \ref{Section: Results in the transformed picture}.

 The arguments so far are based only on the fact that $F$ shifts the component $z_{>d}$ by an amount $f(t)$. The identity $(\iota\circ F)\left(\sum_{j=1}^d t_j u_j\right)=u_{0,\lambda^{-1}(t)}$ would still hold if we used $s_j$ instead of $s_j'$ in Eq.~(\ref{Equation: Definition F}). Nevertheless, it is natural that $F$ shifts the $s$ component as well, since this shift makes sure that $\mathrm{d} F$ preserves the symplectic form $\omega(z_1,z_2):=\mathfrak{Re}\left[z_1\right]^\dagger\cdot \mathfrak{Im}\left[z_2\right]-\mathfrak{Im}\left[z_1\right]^\dagger\cdot \mathfrak{Re}\left[z_2\right]$. Therefore it makes sense to look for a quantum counterpart $\mathcal{W}_N$, which we are going to define in the subsequent Definition \ref{Definition: Unitary Transformation}. In analogy to $F$ preserving the symplectic form $\omega$, the unitary map $\mathcal{W}_N$ is preserving the commutator bracket.

\begin{defi}[Unitary Transformation $\mathcal{W}_N:\mathcal{F}_0 \longrightarrow \mathcal{F}_0$] 
\label{Definition: Unitary Transformation}
Based on the fact that the operators $q_1,\dots,q_d$ defined in Eq.~(\ref{Equation: Definition q}) commute, we can assign to a function $h:\mathbb{R}^d\longrightarrow \mathcal{H}_0$ with components $h_j(t):=u_j^\dagger\cdot h(t)$ an operator $h(q):\mathcal{F}_0\longrightarrow \mathcal{F}_0\otimes \mathcal{H}_0$
\begin{align*}
h(q):=\sum_{j=1}^\infty h_j(q_1,\dots,q_d)\otimes u_j,
\end{align*}
where the operators $h_j(q_1,\dots,q_d)$ are well defined via functional calculus. Let $f$ be the function defined in Eq.~(\ref{Equation: Definition f}), then we can define the unitary map $\mathcal{W}_N:\mathcal{F}_0 \longrightarrow \mathcal{F}_0$ as
\begin{align}
\label{Equation: Definition W}
\mathcal{W}_N:=\mathrm{exp}\left[N  f(q)^\dagger\cdot b_{>d}-N b_{>d}^\dagger\cdot  f(q)\right]=\mathrm{exp}\left[N\sum_{j=d+1}^\infty f_j(q_1,\dots,q_d)\left(b_j-b_j^\dagger\right)\right],
\end{align}
where we have used that $u_j^\dagger\cdot f(t)=0$ for $j\in \{1,\dots,d\}$. Note that $q_1,\dots,q_d$ and $b_{>d}$ have an $N$ dependence, which we suppress in our notation. Furthermore, we define the transformed operators
\begin{align*}
p'_j:&=\mathcal{W}_N\, p_j\, \mathcal{W}_N^{-1},\\
 p':&=\mathcal{W}_N\,  p\, \mathcal{W}_N^{-1}=\sum_{j=1}^d p'_j\otimes u_j,\\
\mathbb{L}':&=\mathcal{W}_N\, \mathbb{L}\, \mathcal{W}_N^{-1},
\end{align*}
where $p$ is defined in Eq.~(\ref{Equation: Definition p}) and $\mathbb{L}$ is defined in Definition \ref{Definition: Basic Fock}. Note that the domain of $\mathbb{L}'$ is $\mathcal{W}_N \mathcal{F}_{\leq N}$, since $\mathbb{L}$ is only defined on $\mathcal{F}_{\leq N}$.
\end{defi}

That the unitary map $\mathcal{W}_N$ is indeed a quantum counterpart to the classical map $F$ defined in Eq.~(\ref{Equation: Definition F}) can be seen from the transformation laws described in the following Lemma \ref{Lemma: Transformation Laws}.

\begin{lem}[Transformation Laws]
\label{Lemma: Transformation Laws}
We have the following transformation laws
\begin{align*}
&\mathcal{W}_N\, b_j\, \mathcal{W}_N^{-1}=b_j+f_j(q) \mbox{ for } j>d,\\
&\mathcal{W}_N\, q_j\, \mathcal{W}_N^{-1}=q_j \mbox{ for } j\in \{1,\dots,d\},\\
&p'_j=p_j-\mathfrak{Im}\left[\partial_{u_j}  f(q)^\dagger\cdot b_{>d}\right] \mbox{ for } j\in \{1,\dots,d\},
\end{align*}
and therefore $\mathcal{W}_N\, b_{\geq 1}\, \mathcal{W}_N^{-1}= q+i p'+b_{>d}+ f(q)$.
\end{lem}
The proof of Lemma \ref{Lemma: Transformation Laws} is elementary and is left to the reader. Before we state the main Theorems of this subsection, let us define what it means for a sequence of operators $X_N$ to be asymptotically small compared to another sequence $Y_N$, in a suitable sense that is specific to our problem.\\

\begin{defi}
\label{Definition: Smallness}
We say that sequences of operators $X_N,Y_N$ with $Y_N\geq 0$ satisfy
\begin{align*}
X_N=o_*(Y_N),
\end{align*}
in case for all $\epsilon>0$ there exists a $\delta>0$, such that $\big|\braket{X_N}_\Psi\big|\leq \epsilon\ \braket{Y_N}_\Psi$ for all $M,N$ with $\frac{M}{N}\leq \delta$ and all elements $\Psi\in \mathcal{W}_N \mathcal{F}_{\leq M}$. Furthermore, we say that sequences of operators $X_N,Y_N$ with $Y_N\geq 0$ satisfy
\begin{align*}
X_N=O_*(Y_N),
\end{align*}
in case there exists a constant $C$ and $\delta_0>0$, such that $\big|\braket{X_N}_\Psi\big|\leq C\ \braket{Y_N}_\Psi$ for all $M,N$ with $\frac{M}{N}\leq \delta_0$ and all $\Psi\in \mathcal{W}_N \mathcal{F}_{\leq M}$.\
\end{defi}

\begin{note}
\label{Remark: Projection Formalism}
Let us denote with $\pi_{M,N}:=\mathcal{W}_N\, \pi_M\, \mathcal{W}_N^{-1}$ the orthogonal projection onto the subspace $\mathcal{W}_N\mathcal{F}_{\leq M}\subset \mathcal{F}_0$. Then the statement $X_N=O_*\left(Y_N\right)$ holds true if and only if there exists a constant $C$ and $\delta_0>0$, such that
\begin{align}
\label{Inequality: pi-2}
\pi_{M,N}\, \mathfrak{Re}\left[\lambda X_N\right] \pi_{M,N}\leq C\  \pi_{M,N} Y_N \pi_{M,N}
\end{align}
for all $\lambda\in \mathbb{C}$ with $|\lambda|=1$ and $\frac{M}{N}\leq \delta_0$. Similarly, $X_N=o_*\left(Y_N\right)$ is equivalent to the existence of a function $\epsilon:\mathbb{R}^+\longrightarrow \mathbb{R}^+$ with $\underset{\delta\rightarrow 0}\lim\ \epsilon(\delta)=0$, such that
\begin{align}
\label{Inequality: pi-1}
\pi_{M,N}\, \mathfrak{Re}\left[\lambda X_N\right] \pi_{M,N}\leq \epsilon\left(\frac{M}{N}\right) \pi_{M,N} Y_N \pi_{M,N}
\end{align}
for all $\lambda\in \mathbb{C}$ with $|\lambda|=1$ and $M\leq N$.\\
\end{note}

\begin{theorem}
\label{Theorem: Decomposition}
Recall the $o_*(\cdot)$ notation from Definition \ref{Definition: Smallness}, the Hartree energy $e_\mathrm{H}$ defined in Eq.~(\ref{Equation: Hartree Energy}) and the Bogoliubov operator $\mathbb{H}$ defined in Eq.~(\ref{Equation: Definition Bogoliubov}), and let us define
\begin{align}
\label{Equation: Definition T}
\mathbb{T}_N:=p^\dagger\cdot  p+b_{>d}^\dagger\cdot  (T+1)\cdot b_{>d}+\frac{1}{N}.
\end{align}
Then, given Assumptions \ref{Assumption: Part I} and \ref{Assumption: Part II}, we have
\begin{align*}
\left(\mathcal{W}_N U_N\right) N^{-1}\!H_N \left(\mathcal{W}_N U_N\right)^{-1}=e_\mathrm{H}+N^{-1}\mathbb{H}+o_*\left(\mathbb{T}_N\right).
\end{align*}
\end{theorem}

The proof of Theorem \ref{Theorem: Decomposition}, which in particular gives rise to a rigorous version of the key inequality Eq.~(\ref{Inequality: Non-Commutative}), will be the content of Section \ref{Section: Results in the transformed picture}. With Theorem \ref{Theorem: Decomposition} at hand we can verify the lower bound in the main Theorem \ref{Theorem: Main Theorem}.

\begin{theorem}[Lower Bound]
\label{Theorem: Lower Bound}
Let $E_N$ be the ground state energy of $H_N$, $e_\mathrm{H}$ the Hartree energy defined in Eq.~(\ref{Equation: Hartree Energy}) and let $\mathbb{H}$ be the Bogoliubov operator defined in Eq.~(\ref{Equation: Definition Bogoliubov}). Given Assumptions \ref{Assumption: Part I} and \ref{Assumption: Part II}, we have the lower bound
\begin{align*}
E_N\geq N\, e_\mathrm{H}+\inf \sigma\left(\mathbb{H}\right)+o_{N\rightarrow \infty}\left(1\right).
\end{align*}
\end{theorem}
\begin{proof}
According to Theorem \ref{Theorem: Bose--Einstein condensation of Ground States}, there exists a sequence of states $\Psi_N\in \bigotimes_\mathrm{s}^N \mathcal{H}$, $\|\Psi_N\|=1$, such that $\braket{H_N}_{\Psi_N}\leq E_N+\alpha_N$ with $\alpha_N\underset{N\rightarrow \infty}{\longrightarrow }0$ and
\begin{align*}
\epsilon_N:=\braket{b_{\geq 1}^\dagger\cdot b_{\geq 1}}_{U_N \Psi_N}=\braket{b_{\geq 1}^\dagger\cdot b_{\geq 1}}_{\Psi_N}\underset{N\rightarrow \infty}{\longrightarrow }0.
\end{align*}
Let us abbreviate $\widetilde{H}_N:=U_N H_N U_N^{-1}$ and let $\pi_M$ be the orthogonal projection onto the space $\mathcal{F}_{\leq M}$ as before. Furthermore, let $0\leq f,g\leq 1$ be smooth functions with $f^2+g^2=1$, $f(x)=1$ for $x\leq \frac{1}{2}$ and $f(x)=0$ for $x\geq 1$, and let us define $f_M(x):=f\left(\frac{x}{M}\right)$ and $g_M(x):=g\left(\frac{x}{M}\right)$. Then the generalized IMS localization formula in \cite[Theorem A.1]{LS}, in the form stated in \cite[Proposition 6.1]{LNSS}, tells us that 
\begin{align*}
\widetilde{H}_N= f_M\left(\mathcal{N}\right)\, \widetilde{H}_N\, f_M\left(\mathcal{N}\right)+g_M\left(\mathcal{N}\right) \,\widetilde{H}_N\, g_M\left(\mathcal{N}\right)-R_{M,N},
\end{align*}
with $R_{M,N}\leq \frac{R}{M^2}\sum_{n=0}^\infty P_n \left(\widetilde{H}_N-E_N\right) P_n$, where $P_n$ is the orthogonal projection onto $\mathcal{F}_{\leq n}\cap \mathcal{F}_{\leq n-1}^\perp$, $\mathcal{N}=\sum_{j=1}^\infty a_j^\dagger a_j$ and $R:=16\left(\|f'\|^2_\infty+\left\|g'\right\|^2_\infty\right)$. Let us define $M_N$ as the smallest integer larger than $\sqrt{\epsilon_N} N$ and $N^{\frac{2}{3}}$. The exponent $\frac{2}{3}$ is somewhat arbitrary and we could use any sequence $\ell_N$ with $N^\frac{1}{2}\ll \ell_N\ll N$ instead. Using the estimate $1-f_{M}(x)^2\leq \frac{2}{M}x$ yields
\begin{align*}
\rho_N:=\braket{1\!-\!f_{M_N}(\mathcal{N})^2}_{U_N \Psi_N}\leq \frac{2}{M_N}\braket{\mathcal{N}}_{U_N\Psi_N}=\frac{2N}{M_N}\braket{b_{\geq 1}^\dagger\cdot b_{\geq 1}}_{U_N \Psi_N}\leq \frac{2}{\sqrt{\epsilon_N}}\epsilon_N\! \underset{N\rightarrow \infty}{\longrightarrow }\!0.
\end{align*}
Let us define $\Phi_{N}:=(1-\rho_N)^{-\frac{1}{2}}\ f_{M_N}(\mathcal{N}) U_N \Psi_N$. Using Lemma \ref{Lemma: Decomposition} and the inequality $\widetilde{H}_N\geq E_N$ yields
\begin{align}
\label{Inequality: Main Proof 1}
E_N\!+\!\alpha_N\!\geq\! \braket{\widetilde{H}_N}_{U_N \Psi_N}\!\geq\! (1\!-\!\rho_N)\braket{\widetilde{H}_N}_{\Phi_N}\!+\!\rho_N E_N\!-\!\frac{R}{M_N^2}\braket{k\widetilde{H}_N\!+\!k^2 N\!-\!E_N}_{U_N\Psi_N}.
\end{align}
Since $\lim_N N^{-1}E_N=e_\mathrm{H}$, we obtain that $\beta_N:=\frac{R}{M_N^2} \braket{k\widetilde{H}_N+k^2N-E_N}_{U_N \Psi_N}$ satisfies
\begin{align*}
\beta_N\leq \frac{R}{N^{\frac{4}{3}}}\left((k-1)E_N+k\alpha_N+k^2N\right)\underset{N\rightarrow \infty}{\longrightarrow }0.
\end{align*}
We can now rewrite Inequality (\ref{Inequality: Main Proof 1}) as
\begin{align*}
E_N\geq \braket{\widetilde{H}_N}_{\Phi_N}-\frac{\alpha_N+\beta_N}{1-\rho_N}.
\end{align*}

Let $r>0$ be as in the assumption of Theorem \ref{Theorem: Bogoliubov Ground State Energy} and recall the definition of $\mathbb{A}$ in Eq.~(\ref{Equation: NT}). Note that $N\mathbb{T}_N=\mathbb{A}+1$. By Theorem \ref{Theorem: Decomposition} and Remark \ref{Remark: Projection Formalism}, there exists a function $\epsilon$ with $\lim_{\delta\rightarrow 0}\epsilon(\delta)$, such that
\begin{align*}
&\braket{\widetilde{H}_N}_{\Phi_N}\geq N\ e_\mathrm{H}+\braket{\mathbb{H}}_{\mathcal{W}_N\Phi_N}-\epsilon\left(\frac{M_N}{N}\right)\ \braket{\mathbb{A}+1}_{\mathcal{W}_N \Phi_N}\\
&\ \ =N\ e_\mathrm{H}+\left(1-\frac{1}{r}\epsilon\left(\frac{M_N}{N}\right)\right)\braket{\mathbb{H}}_{\mathcal{W}_N\Phi_N}+\frac{1}{r}\epsilon\left(\frac{M_N}{N}\right)\braket{\mathbb{H}-r\mathbb{A}}_{\mathcal{W}_N\Phi_N}-\epsilon\left(\frac{M_N}{N}\right)\\
&\ \ \geq N\ e_\mathrm{H}+\inf \sigma\left(\mathbb{H}\right)+\frac{1}{r}\epsilon\left(\frac{M_N}{N}\right)\big(\inf \sigma\left(\mathbb{H}-r\mathbb{A}\right)-\inf \sigma\left(\mathbb{H}\right)\big)-\epsilon\left(\frac{M_N}{N}\right)
\end{align*}
for all $N$ large enough such that $1-\frac{1}{r}\epsilon\left(\frac{M_N}{N}\right)\geq 0$. This concludes the proof, since $\inf \sigma\left(\mathbb{H}-r\mathbb{A}\right)>-\infty$ by Theorem \ref{Theorem: Bogoliubov Ground State Energy}.
\end{proof}

\section{Taylor Expansion of $\left(\mathcal{W}_N U_N\right) H_N \left(\mathcal{W}_N U_N\right)^{-1}$}
\label{Section: Results in the transformed picture}

This section is devoted to the verification of the main technical Theorem \ref{Theorem: Decomposition}, which is the rigorous version of inequality Eq.~(\ref{Inequality: Non-Commutative}). Before we explain the proof, recall the definition of $\iota$ in Eq.~(\ref{Equation: Embedding}) and $F$ in Eq.~(\ref{Equation: Definition F}), and let us verify the classical counterpart Eq.~(\ref{Inequality: Commutative}). For this purpose we define the functional
\begin{align}
\label{Equation: Parameterized Hartree energy}
\mathcal{E}'(z):=\mathcal{E}_\mathrm{H}\left[\iota\left(F(z)\right)\right],
\end{align}
which satisfies according to the definition of $F$ that $\mathcal{E}'\left(\vec{t}\right)=e_\mathrm{H}$ for all $t\in \mathbb{R}^d$ with $\vec{t}:=\sum_{j=1}^d t_j u_j$, i.e. $F$ flattens the manifold of minimizers of $\mathcal{E}_\mathrm{H}\circ \iota$. We will verify Eq.~(\ref{Inequality: Commutative}) by sorting the functional $\mathcal{E}'$ in terms of powers in the variables $s$ and $z_{>d}$ for any $z=\sum_{j=1}^d \left(t_j+is_j\right)u_j+z_{>d}\in \mathcal{H}_0$ with $z_{>d} \in \{u_1,\dots,u_d\}^\perp$. In the following, let $\pi(z):=\sum_{j=1}^d is_ju_j+z_{>d}$ be the projection onto $\mathcal{V}:=\pi\left(\mathcal{H}_0\right)$. We can now sort $\mathcal{E}'(z)$ in terms of powers in $s$ and $z_{>d}$, i.e. in terms of powers in $\pi(z)$, using a Taylor approximation with expansion point $\vec{t}$
\begin{align}
\nonumber \mathcal{E}'(z)&=\mathcal{E}'\left(\vec{t}+\pi(z)\right)=\mathcal{E}'\left(\vec{t}\ \right)\!+D|_{\vec{t}}\ \mathcal{E}'\big(\pi(z)\big)+\frac{1}{2}D^2|_{\vec{t}}\ \mathcal{E}'\big(\pi(z)\big)+\{\ \mathrm{Higher Orders}\ \}\\
\label{Equation: Second version of Taylor}&=\mathcal{E}'\left(\vec{t}\ \right)\!+D_{\mathcal{V}}|_{\vec{t}} \ \mathcal{E}'(z)+\frac{1}{2}D_{\mathcal{V}}^2|_{\vec{t}}\ \mathcal{E}'\big(z\big)+\{\ \mathrm{Higher Orders}\ \},
\end{align}
where $D|_{z_0}\mathcal{E}'(v)$ is the first derivative of $\mathcal{E}'$ in the direction $v$ at $z_0$, $D^2|_{z_0} \mathcal{E}'(v)$ is the second derivative in the direction $v$, and $D_{\mathcal{V}}|_{z_0} \mathcal{E}'(v):=D|_{z_0} \mathcal{E}'\left(\pi(v)\right)$ and $D_{\mathcal{V}}^2|_{z_0} \mathcal{E}'\big(v\big):=D^2|_{z_0} \mathcal{E}'\big(\pi(v)\big)$ are the derivatives only with respect to directions in $\mathcal{V}$. Using $\mathcal{E}'\left(\vec{t}\ \right)=e_\mathrm{H}$, $D_{\mathcal{V}}|_{\vec{t}}\ \mathcal{E}'=0$ and the fact that $D_{\mathcal{V}}^2|_{\vec{t}}\, \mathcal{E}'(v)\geq \left(1-\frac{\epsilon}{2}\right) D_{\mathcal{V}}^2|_{0}\ \mathcal{E}'(v)$ for $t$ small enough by continuity, we formally arrive at Eq.~(\ref{Inequality: Commutative}), which is claimed to hold only for small $\|z\|^2=|t|^2+\|\pi(z)\|^2$ anyway.\\

By sorting the expression $\left(\mathcal{W}_N U_N\right) H_N \left(\mathcal{W}_N U_N\right)^{-1}$ in terms of powers in the operators $p$ and $b_{>d}$, we will verify that we end up with the same Taylor approximation we obtained by  sorting $\mathcal{E}'(z)$ in terms of powers in the variables $s$ and $z_{>d}$. More precisely, our goal is to  verify that
\begin{align}
\label{Equation: Formal main formula - first line}  \left(\mathcal{W}_N U_N\right)\! N^{-1}\!H_N \left(\mathcal{W}_N U_N\right)^{-1}\!&=\!\mathcal{E}\!\left(q\right)\!+\!D_{\mathcal{V}}\big|_{q} \mathcal{E}\Big(b_{\geq 1}\Big)\!+\!\frac{1}{2}D_{\mathcal{V}}^2\big|_{0} \mathcal{E}\Big(b_{\geq 1}\Big)\!+\!o_*\left(\mathbb{T}_N\right)\\
\nonumber &=e_\mathrm{H}+N^{-1}\mathbb{H}+o_*\left(\mathbb{T}_N\right),
\end{align}
where $\mathcal{E}:\mathrm{dom}\left[T\right]\rightarrow \mathbb{R}$ is a differentiable extension to all of $\mathrm{dom}\left[T\right]$ of the functional $\mathcal{E}'\big|_{B_r}$, restricted to the ball $B_r:=\{z\in \mathcal{H}_0\cap \mathrm{dom}\left[T\right]:\|z\|<r\}$ for a sufficiently small $r>0$. Note that the spectrum of the operators $q_1,\dots,q_d$ is the whole real axis $\mathbb{R}$. In order to even define $\mathcal{E}\left(q\right)$ and $D_{\mathcal{V}}\big|_{q} \mathcal{E}\Big(b_{\geq 1}\Big)$ with the help of functional calculus, it is therefore necessary that $\mathcal{E}$, in contrast to $\mathcal{E}'$, is an everywhere defined and differentiable functional. For such a function $\mathcal{E}$ we can define $\mathcal{E}\left(q\right)$ via functional calculus starting from the function $ t\mapsto \mathcal{E}\left(\sum_{j=1}^d t_j u_j\right)$ for $t\in \mathbb{R}^d$. 
%
%
The so far formal objects $D_{\mathcal{V}}\big|_{q} \mathcal{E}\Big(b_{\geq 1}\Big)$ and $\frac{1}{2}D_{\mathcal{V}}^2\big|_{0} \mathcal{E}\Big(b_{\geq 1}\Big)$ are later defined in Definition \ref{Definition: Quantization}. Note that it is a necessity to restrict $\mathcal{E}'$ to a sufficiently small ball $B_r$ first, to be precise we require that $\|F(z)\|\leq 1-\delta$ for all $z\in B_r$ where $0<\delta<1$, since $\mathcal{E}'$ itself does not have a differentiable extension due to the square root appearing in the definition of $\iota$, see Eq.~(\ref{Equation: Embedding}). 

In order to reduce the technical efforts of proving Eq.~(\ref{Equation: Formal main formula - first line}), we will make use of the fact that
\begin{align*}
\left(\mathcal{W}_N U_N\right) N^{-1}H_N \left(\mathcal{W}_N U_N\right)^{-1}=\mathcal{W}_N \widetilde{A}_N\mathcal{W}_N^{-1}+\mathcal{W}_N \widetilde{B}_N\mathcal{W}_N^{-1} +o_*\left(\mathbb{T}_N\right),
\end{align*}
which, as we will see in the proof of Theorem \ref{Theorem: Decomposition}, is a consequence of Eq.~(\ref{Equation: Corollary: Estimates Lower Bound - first line}). We can then prove Eq.~(\ref{Equation: Formal main formula - first line}) separately for the operators $\mathcal{W}_N\, \widetilde{A}_N\, \mathcal{W}_N^{-1}$ and $\mathcal{W}_N\, \widetilde{B}_N\, \mathcal{W}_N^{-1}$. In fact,  we are going to verify that
\begin{align}
\label{Equation: Quantum A} \mathcal{W}_N\, \widetilde{A}_N\, \mathcal{W}_N^{-1}&=\mathcal{E}_A\left(q\right)\!+D_{\mathcal{V}}\big|_{q} \mathcal{E}_A\Big(b_{\geq 1}\Big)+\!\frac{1}{2}D_{\mathcal{V}}^2\big|_{0} \mathcal{E}_A\Big(b_{\geq 1}\Big)+\frac{c}{N}+o_*\left(\mathbb{T}_N\right),\\
\label{Equation: Quantum B} \mathcal{W}_N\, \widetilde{B}_N  \, \mathcal{W}_N^{-1}&=\mathcal{E}_B\left(q\right)\!+D_{\mathcal{V}}\big|_{q} \mathcal{E}_B\Big(b_{\geq 1}\Big)+\!\frac{1}{2}D_{\mathcal{V}}^2\big|_{0} \mathcal{E}_B\Big(b_{\geq 1}\Big)-\frac{c}{N}+o_*\left(\mathbb{T}_N\right),
\end{align}
where the constant $c$ arises due to the non-commutative nature of the operators $q$ and $p$, and $\mathcal{E}_A$ and $\mathcal{E}_B$ are differentiable extensions of $\mathcal{E}'_A,\mathcal{E}'_B:B_r\longrightarrow \mathbb{C}$ 
\begin{align}
\label{Equation: Classic} \mathcal{E}'_A(z):&=u_z^\dagger\cdot T\cdot u_z, &\mathcal{E}'_B(z):=\frac{1}{2}\left(u_z\otimes u_z\right)^\dagger \cdot \hat{v}\cdot u_z\otimes u_z,
\end{align}
where $u_z:=\iota\left(F(z)\right)$. The proofs of Eqs.~(\ref{Equation: Quantum A}) and (\ref{Equation: Quantum B}) will be carried out in Subsections \ref{Subsection: Results in the transformed picture-A} and \ref{Subsection: Results in the transformed picture-B}, respectively. We have to perform a variety of operator estimates, and since $\mathcal{W}_N \widetilde{A}_N \mathcal{W}_N^{-1}$ and $\mathcal{W}_N\widetilde{B}_N \mathcal{W}_N^{-1}$ involve factors of the form $\sqrt{1-\mathbb{L}'}$ with $\mathbb{L}'$ defined in Definition \ref{Definition: Unitary Transformation}, we need in particular to estimate the Taylor residuum corresponding to approximations of such terms. The operator estimates can be found in Appendix \ref{Appendix: B}, respectively Appendix \ref{Appendix: C} for the operator square root specifically.

\subsection{Taylor Expansion of $\mathcal{W}_N \widetilde{A}_N \mathcal{W}_N^{-1}$}
\label{Subsection: Results in the transformed picture-A}
In order to structure the analysis, we split the operator $\mathcal{W}_N \widetilde{A}_N \mathcal{W}_N^{-1}$ into simpler operators $H_J$, introduced in Definition \ref{Definition: Hamilton Components A}, and we split the classical counterpart $\mathcal{E}_A$ defined in Eq.~(\ref{Equation: Classic}) into atoms $\mathcal{E}_J$, defined in Definition \ref{Definition: Functional Components A}. In Lemma \ref{Lemma: A Decomposition}, we then explain how $\mathcal{W}_N \widetilde{A}_N \mathcal{W}_N^{-1}$ and $\mathcal{E}_A$ can be written in terms of $H_J$ and $\mathcal{E}_J$, respectively.
\begin{defi}
\label{Definition: Hamilton Components A}
Recall the function $t\mapsto f(t)$ from Definition \ref{Definition: F}. For $i\in \{0,\dots,4\}$, we define operators $ h_i:\mathrm{dom}[\mathcal{N}]\longrightarrow \mathcal{F}_0\otimes \mathcal{H}$ by $ h_0:= 1_{\mathcal{F}_0}\otimes u_0$ and
\begin{align*}
 h_1:&= q=\sum_{j=1}^d q_j\otimes u_j,\   \ \ \ \  h_3:=i p'=i\sum_{j=1}^d \left(p_j-\mathfrak{Im}\left[\partial_j f(q)^\dagger\cdot b_{>d}\right]\right)\otimes u_j,\\
 h_2:&= f(q),\ \ \ \ \ \ \ \ \ \ \ \ \ \ \ \ \ \ \  h_4:=b_{>d},
\end{align*}
where $f(q)$ and $\partial_j f(q)$ are defined according to Definition \ref{Definition: Unitary Transformation}. Furthermore, for a multi-index $J=(i,j)$ with $i,j\in \{0,\dots,4\}$ we define an operator $H_J$ on $\mathcal{W}_N\mathcal{F}_{\leq N}$ as
\begin{align*}
H_J:&= h_i^\dagger \cdot  T\cdot  h_j\, \Big(1-\mathbb{L}'\Big)^{\frac{m_J}{2}},
\end{align*}
where $m_J$ counts how many of the indices in $J=(i,j)$ are zero.
\end{defi}

\begin{defi}
\label{Definition: Functional Components A}
Let us decompose an arbitrary  $z\in \mathcal{H}_0$ as $z=\sum_{j=1}^d(t_j+is_j)\ u_j+z_{>d}$, with $t,s\in \mathbb{R}^d$ and $z_{>d}\in \{u_1,\dots,u_d\}^\perp$. For $i\in \{0,\dots,4\}$, we define in analogy to Definition \ref{Definition: Hamilton Components A} the functions $e_i:\mathcal{H}_0 \longrightarrow \mathcal{H}$ by $e_0(z):=u_0$ and
\begin{align*}
e_1(z):&=\sum_{j=1}^d t_j\ u_j, \ \ \ \ \ \ \ \ e_3(z):=i\sum_{j=1}^d\left(s_j-\mathfrak{Im}\left[\partial_jf(t_1,\dots,t_d)^\dagger\cdot z\right]\right)u_j,\\
e_2(z):&=f\left(t\right),\ \ \  \ \ \ \ \ \ \ \ \ \ e_4(z):=z_{>d}.
\end{align*}
With this at hand, we can write the transformation $F:\mathcal{H}_0\longrightarrow \mathcal{H}_0$ from Eq.~(\ref{Equation: Definition F}) as
\begin{align*}
F(z)=e_1(z)+e_2(z)+e_3(z)+e_4(z).
\end{align*}
Furthermore, consider for $m\in \{0,\dots,4\}$ the functions
\begin{align}
\label{Equation: eta_m}
\eta_m\left(z\right):=
\begin{cases}
\Big(1-\|F(z)\|^2\Big)^{\frac{m}{2}}\hbox{ for even $m$},\\
\chi\big(\|F(z)\|^2\big)\Big(1-\|F(z)\|^2\Big)^{\frac{m}{2}}\hbox{ for odd $m$},
\end{cases}
\end{align}
where $\chi$ is a smooth function with $0\leq \chi(x)\leq 1$, $\mathrm{supp}\left(\chi\right)\subset [0,1)$ and $\chi(x)=1$ for $|x|<\frac{1}{2}$. Then we can define for a multi-index $J=(i,j)$ with $i,j\in \{0,\dots,4\}$ the function $\mathcal{E}_J:\mathcal{H}_0\cap \mathrm{dom}[T] \longrightarrow \mathbb{C}$ as
\begin{align*}
\mathcal{E}_J(z):=e_{i}(z)^\dagger\cdot T\cdot e_{j}(z)\ \eta_{m_J}\left(z\right),
\end{align*}
where $m_J$ counts how many of the two indices $i,j$ are zero.\\
\end{defi}

\begin{lem}
\label{Lemma: A Decomposition}
Let us define for all $i,j\in \{1,\dots,4\}$ the coefficients $\lambda_{(0,0)}:=1$, $\lambda_{(i,0)}:=2$, $\lambda_{(i,j)}:=1$ and $\lambda_{(0,j)}:=0$. Then
\begin{align}
\label{Equation: A Decomposition quantum} \mathcal{W}_N\, \widetilde{A}_N\, \mathcal{W}_N^{-1}&=\sum_{J\in \{0,\dots,4\}^2}\lambda_J\ \mathfrak{Re}\left[H_J\right],
\end{align}
where $\widetilde{A}_N$ is defined in Eq.~(\ref{Equation: tilde A}). Furthermore, the functional $\mathcal{E}_A$ defined as
\begin{align}
\label{Equation: A Decomposition classic}\mathcal{E}_A(z):&=\sum_{J\in \{0,\dots,4\}^2}\lambda_J\ \mathfrak{Re}\left[\mathcal{E}_J(z)\right],
\end{align}
is an extension of $\mathcal{E}'_A\big|_{B_r}$ defined in Eq.~(\ref{Equation: Classic}), where $B_r:=\{z\in \mathcal{H}_0\cap \mathrm{dom}\left[T\right]:\|z\|<r\}$ and $r>0$ is a constant such that $\|F(z)\|<\frac{1}{2}$ for all $z\in \mathcal{H}_0$ with $\|z\|<r$.
\end{lem}
Note that the operator $\mathcal{W}_N\, \widetilde{A}_N\, \mathcal{W}_N^{-1}$ involves terms with $\sqrt{1-\mathbb{L}'}$ on the right side as well as on the left side. In order to reduce the technical effort later, it will be convenient to have all of them on one side, say the right side. This can be achieved by using the real part, e.g. we can write for $j\in \{1,\dots,4\}$
\begin{align*}
h_j^\dagger\cdot T\cdot u_0 \sqrt{1-\mathbb{L}'}+\sqrt{1-\mathbb{L}'}u_0^\dagger\cdot T\cdot h_j=\mathfrak{Re}\left[h_j^\dagger\cdot T\cdot u_0 \sqrt{1-\mathbb{L}'}\right]=2\mathfrak{Re}\left[H_{(j,0)}\right].
\end{align*}
Therefore we set all the coefficients $\lambda_{(0,j)}$ in Lemma \ref{Lemma: A Decomposition} to zero, since the $(0,j)$-contribution is already included in the real part of the $(j,0)$-contribution.
\begin{proof}
Eq.~(\ref{Equation: A Decomposition quantum}) follows from the transformation law $\mathcal{W}_N\, b_{\geq 1}\, \mathcal{W}_N^{-1}=h_1+h_2+h_3+h_4$, where $h_i$ is defined in Definition \ref{Definition: Hamilton Components A}, and the definition $\mathbb{L}'=\mathcal{W}_N\, \mathbb{L}\, \mathcal{W}_N^{-1}$. Similarly we obtain $\mathcal{E}_A(z)=\mathcal{E}_A'(z)$ for all $z$ with $\|z\|<r$ for $r$ as above and the fact that
\begin{align*}
\iota\big(F(z)\big)=\iota\left(e_1(z)\!+\!e_2(z)\!+\!e_3(z)\!+\!e_4(z)\right)=\sqrt{1\!-\!\|F(z)\|^2}e_0\!+\!e_1(z)\!+\!e_2(z)\!+\!e_3(z)\!+\!e_4(z).
\end{align*}
\end{proof}

In order to prove the Taylor approximation in Eq.~(\ref{Equation: Quantum A}), we will verify that each of the atoms $H_J$ can be approximated using the quantized Taylor coefficients of $\mathcal{E}_J$. The quantized Taylor coefficients $D_{\mathcal{V}}\big|_{q} \mathcal{E}_J\Big(b_{\geq 1}\Big)$ and $D_{\mathcal{V}}^2\big|_{0}\ \mathcal{E}_J\Big(b_{\geq 1}\Big)$ are rigorously defined by the following Definition.
\begin{defi}
\label{Definition: Quantization}
Let $L_t:\mathcal{H}_0\longrightarrow \mathbb{C}$ be a bounded $\mathbb{R}$-linear map for all $t\in \mathbb{R}^d$, and let $w(t),\widetilde{w}(t)$ be the unique elements in $\mathcal{H}_0$ such that $L_t(z)=w(t)^\dagger\cdot z+z^\dagger\cdot \widetilde{w}(t)$. Then we define
\begin{align}
\label{Equation: Linear Quantization}
L_q\left(b_{\geq 1}\right):=w(q)^\dagger\cdot b_{\geq 1}+b_{\geq 1}^\dagger\cdot \widetilde{w}(q).
\end{align}
Let furthermore $\Lambda$ be an $\mathbb{R}$-quadratic form on $\mathcal{H}_0$ with a unique decomposition $\Lambda(z)=z^\dagger\cdot Q\cdot z+G^\dagger\cdot z\otimes z+\left(z\otimes z\right)^\dagger\cdot \widetilde{G}$ where $Q$ is an operator on $\mathcal{H}_0$ and $G,\widetilde{G}\in \mathcal{H}\otimes_\mathrm{s} \mathcal{H}_0$ (or, more generally, in $ \overline{\mathcal{H}_0\otimes_\mathrm{s} \mathcal{H}_0}^{\|.\|_*}$ as introduced in Lemma~\ref{Lemma: Hessian}). Then we define $\Lambda\left(b_{\geq 1}\right)$ as
\begin{align*}
\Lambda\left(b_{\geq 1}\right):=b_{\geq 1}^\dagger\cdot Q\cdot b_{\geq 1}+G^\dagger\cdot b_{\geq 1}\ \underline{\otimes}\ b_{\geq 1}+\left(b_{\geq 1}\ \underline{\otimes}\ b_{\geq 1}\right)^\dagger\cdot \widetilde{G}.
\end{align*} 
\end{defi}

In the following we want to verify that the residuum $R_J$ defined as 
\begin{align}
\label{Equation: Taylor residuum}
R_J:&=H_J-\mathcal{E}_J(q)-D_\mathcal{V}\big|_{q} \mathcal{E}_J\left(b_{\geq 1}\right)-\frac{1}{2}D^2_\mathcal{V}\big|_{0} \mathcal{E}_J\left(b_{\geq 1}\right)-\frac{c_J}{N}
\end{align}
is small, where the constant $c_J$ are given by
\begin{align}
\nonumber
c_{(0,0)}:&=\frac{d}{4}\ u_0^\dagger\cdot T\cdot u_0=-\frac{1}{8}\ \sum_{j=1}^d \partial_{t_j}^2\big|_{t=0}\mathcal{E}_{(0,0)}\left(\vec{t}\ \right),\\
\label{Equation: c_J for A} c_{(3,3)}:&=\frac{1}{4}\ \sum_{j=1}^d u_j^\dagger\cdot T\cdot u_j=\frac{1}{8}\ \sum_{j=1}^d \partial_{t_j}^2\big|_{t=0}\mathcal{E}_{(1,1)}\left(\vec{t}\ \right),
\end{align}
$c_{(1,3)}:=c_{(3,1)}:=-c_{(3,3)}$ and $c_J:=0$ for all other $J\in \{0,\dots,4\}^2$, where $\vec{t\,}:=\sum_{j=1}^d t_j u_j$. The proof will be spit into two parts. In Lemma \ref{Lemma: Taylor A} we derive an explicit representation of the residuum $R_J$ by sorting the operator $H_J$ in terms of powers in $p$ and $b_{>d}$, and in Theorem \ref{Theorem: Main A} we will make sure that this residuum is indeed small compared to the operator $\mathbb{T}_N$ defined in Eq. (\ref{Equation: Definition T}), which is quadratic in the operators $p$ and $b_{>d}$.

In order to illustrate the emergence of the additional constants $c_J$ in the residuum $R_J$ in Eq.~(\ref{Equation: Taylor residuum}), let us first investigate the following toy problem.\\

\textbf{Example}. Consider the toy Hamiltonian $H_{\mathrm{toy}}:=b_1^\dagger b_1$ and the corresponding Hartree functional $\mathcal{E}_{\mathrm{toy}}:\mathbb{C}\longrightarrow \mathbb{C}$ given by $\mathcal{E}_{\mathrm{toy}}[z]:=|z|^2$. Using $b_1=q_1+ip_1$ and the commutation relation $[ip_1,q_1]=\frac{1}{2N}$, we obtain
\begin{align}
\label{Equation: Toy}
H_{\mathrm{toy}}=q_1^2\!+\!p_1^2\!-\!\frac{1}{2N}=q_1^2\!-\!\frac{1}{4}\left(b_1\!-\!b_1^\dagger\right)^2\!-\!\frac{1}{2N}=q_1^2\!+\!\frac{1}{2}b_1^\dagger b_1\!-\!\frac{1}{4}b_1^2\!-\!\frac{1}{4}\left(b_1^\dagger\right)^2\!-\!\frac{1}{4N}.
\end{align}
Let $D_\mathcal{V}$ be the derivative with respect to the imaginary part and $z=t+is\in \mathbb{C}$, then
\begin{align*}
\frac{1}{2}D^2_\mathcal{V}|_0\, \mathcal{E}_{\mathrm{toy}}(z)=\frac{1}{2}D^2|_0\, \mathcal{E}_{\mathrm{toy}}(is)=s^2=\frac{1}{2}|z|^2-\frac{1}{4}z^2-\frac{1}{4}\bar{z}^2.
\end{align*}
With the definition $c_{\mathrm{toy}}:=-\frac{1}{8}\partial_t^2\big|_{t=0} \mathcal{E}_{\mathrm{toy}}[t]=-\frac{1}{4}$ we can therefore rewrite Eq.~(\ref{Equation: Toy}) as
\begin{align*}
H_{\mathrm{toy}}=\mathcal{E}_{\mathrm{toy}}[q_1]+\frac{1}{2}D^2_\mathcal{V}|_0 \mathcal{E}_{\mathrm{toy}}(b_1)+\frac{c_{\mathrm{toy}}}{N}.
\end{align*}

\begin{defi}[Taylor approximation of the square root]
\label{Definition: Taylor of the square root}
Let $\eta_m$ be the function defined in Eq.~(\ref{Equation: eta_m}) and let us define the constant $c_m:=\frac{m}{8}d$. We then define the residuum corresponding to the operator Taylor approximation of $\Big(1-\mathbb{L}'\Big)^{\frac{m}{2}}$, for different degrees of accuracy, as
\begin{align*}
E_m^0:&=\Big(1-\mathbb{L}'\Big)^{\frac{m}{2}}-\eta_m\left(q\right),\\
E_m^1:&=\Big(1-\mathbb{L}'\Big)^{\frac{m}{2}}-\eta_m\left(q\right)- D_{\mathcal{V}}\big|_q \eta_m\big(b_{\geq 1}\big),\\
E_m^2:&=\Big(1-\mathbb{L}'\Big)^{\frac{m}{2}}-\eta_m\left(q\right)- D_{\mathcal{V}}\big|_q \eta_m\big(b_{\geq 1}\big)-\frac{1}{2}D^2_{\mathcal{V}}\big|_0 \eta_m\big(b_{\geq 1}\big)-\frac{c_m}{N}.
\end{align*}
\end{defi}

\begin{lem}
\label{Lemma: Taylor A}
Let $J=(i,j)\in \{0,\dots,4\}^2$ be such that $\lambda_J\neq 0$, where $\lambda_J$ is defined in Lemma \ref{Lemma: A Decomposition}, and let $R_J$ be the residuum defined in Eq.~(\ref{Equation: Taylor residuum}). By distinguishing different cases with the help of the index $e_J:=\left|\{\ell\in J:\ell\in \{3,4\}\}\right|$ and the index $m_J:=\left|\{\ell\in J:\ell=0\}\right|$, we can explicitly express $R_J$ as
\begin{itemize}
\item In the case $m_J=2$, i.e. $J=(0,0)$: $R_{(0,0)}=\left( u_0^\dagger\cdot  T\cdot  u_0\right) E_{2}^2$.
\item In the case $e_J=0$ and $m_J<2$: $R_J=\left( h_i^\dagger\cdot T\cdot  h_j\right) E_{m_J}^1$.
\item In the case $e_J=1$, there exists a constant $C$ and functions $F_J:\mathbb{R}^d\longrightarrow \mathbb{R}$ with $|F_J(t)|\leq C|t|$, such that 
\begin{align}
\label{Equation: Proof with F_J}R_J=\left( h_i^\dagger\cdot  T\cdot  h_j\right) E_{m_J}^0+\frac{F_J(q)}{N}.
\end{align}
\item For $e_J=2$ we distinguish further between the individual cases and obtain
\begin{align*}
\ \ \ \ R_{(3,3)}&=\left(i p'-i p\right)^\dagger\cdot  T\cdot  ip'+ (ip)^\dagger \cdot  T\cdot \left(i p'-i p\right),\\
R_{(3,4)}&=\left(i p'-i p\right)^\dagger\cdot  T\cdot b_{>d}=R_{(4,3)}^\dagger,\\
R_{(4,4)}&=0.
\end{align*}
\end{itemize}
\end{lem}

\begin{proof}
The Lemma can be verified by straightforward computations for the different individual cases. For the purpose of illustration, we will explicitly carry out the computations for the case $J=(3,j)$ with $j\in \{0,1,2\}$, i.e. we are going to verify Eq.~(\ref{Equation: Proof with F_J}) for this special case. Using the definition of $E^0_m$ in Definition \ref{Definition: Taylor of the square root}, the observation $h_j=e_j(q)$ and the fact that $\left(ip'_\ell\right)^\dagger= b_{\geq 1}^\dagger \cdot  \big(u_\ell-\partial_{u_\ell} f(q)\big)-\big(u_\ell-\partial_{u_\ell} f(q)\big)^\dagger\cdot b_{\geq 1}$, we obtain
\begin{align}
&\nonumber H_J=\left(ip'\right)^\dagger\cdot T\cdot e_j(q)\left(1-\mathbb{L}'\right)^{\frac{m}{2}}=\left(ip'\right)^\dagger\cdot T\cdot e_j(q)\eta_m(q)+\left(ip'\right)^\dagger\cdot T\cdot e_j(q)E^0_{m}\\
&\nonumber \ \ =\frac{1}{2}\!\sum_{\ell=1}^d b_{\geq 1}^\dagger\!\cdot \! \big(u_\ell\!-\!\partial_{u_\ell} f(q)\!\big) u_\ell^\dagger \!\cdot \!T\!\cdot\! e_j(q)\, \eta_m(q)-\frac{1}{2}\!\sum_{\ell=1}^d\big(u_\ell\!-\!\partial_{u_\ell} f(q)\!\big)^\dagger\!\cdot\! b_{\geq 1} u_\ell^\dagger \!\cdot \!T\!\cdot\! e_j(q)\, \eta_m(q)\label{Equation: Expression for R_J}\\
& \ \ \ \ \ \ \ \ \ +\left(ip'\right)^\dagger\cdot T\cdot e_j(q)E^0_{m},
\end{align}
where $m:=m_J$. Our goal is to commute $b_{\geq 1}$ in $\big(u_\ell-\partial_{u_\ell} f(q)\big)^\dagger\cdot b_{\geq 1} u_\ell^\dagger \cdot T\cdot e_j(q)\, \eta_m(q)$ to the right side, in order to obtain an expression which is of the same form as~(\ref{Equation: Linear Quantization}). We define the corresponding functions $w$ and $\widetilde{w}$ as
\begin{align*}
w:=-\frac{1}{2}\sum_{\ell=1}^d \left(u_\ell^\dagger \cdot T \cdot e_j\left(\vec{t}\ \right)\! \eta_{m}\left(\vec{t}\ \right)\right) \big(u_\ell-\partial_{t_\ell} f(t)\big)
\end{align*} 
and $\widetilde{w}(t):=-w(t)$. The commutation law $\left[g(q),\big(u_\ell\!-\!\partial_{u_\ell} f(q)\!\big)^\dagger\!\cdot\! b_{\geq 1}\right]=\left[g(q),ip_\ell\right]=-\frac{1}{2N}\partial_{\ell}g(q)$, for $C^1$ functions $g:\mathbb{R}^d\longrightarrow \mathbb{R}$ then yields
\begin{align*}
-\frac{1}{2}\sum_{\ell=1}^d\big(u_\ell-\partial_{u_\ell} f(q)\big)^\dagger\! \cdot\! b_{\geq 1}\, u_\ell^\dagger \cdot T\cdot e_j(q)\, \eta_m(q)=w(q)^\dagger\cdot b_{\geq 1}+\frac{1}{N}y(q),
\end{align*}
where $y:\mathbb{R}^d\longrightarrow \mathbb{R}$ is defined as $y(t):=-\frac{1}{4}\sum_{\ell }\partial_{\ell} \left(u_\ell\cdot T\cdot e_j(\vec{t}\ )\eta_{m}(\vec{t}\ )\right)$. Furthermore
\begin{align*}
D_\mathcal{V}\big|_{\vec{t}}\ \mathcal{E}_J(z)\! & =\!e_3(z)^\dagger\!\cdot\! T\!\cdot\! e_j\left(\vec{t}\ \right)\!\eta_{m}\!\left(\vec{t}\ \right)\!=\!\sum_{\ell=1}^d\!\left(i\mathfrak{Im}\left[\!\big(u_\ell\!-\!\partial_{t_\ell} f(t)\big)^\dagger\!\cdot\! z\!\right]u_\ell\right)^\dagger\!\cdot\! T\!\cdot\! e_j\left(\vec{t}\ \right)\!\eta_{m}\!\left(\vec{t}\ \right)\\
&=w(t)^\dagger\cdot z+z^\dagger\cdot \widetilde{w}(t).
\end{align*}
Consequently we can rewrite Eq.~(\ref{Equation: Expression for R_J}) as
\begin{align*}
\left(ip'\right)^\dagger\cdot T\cdot e_j(q)\left(1-\mathbb{L}'\right)^{\frac{m}{2}}&=D_\mathcal{V}\big|_{q}\ \mathcal{E}_J\left(b_{\geq 1}\right)+\frac{1}{N}y(q)+\left(ip'\right)^\dagger\cdot T\cdot e_j(q)E^0_{m}.
\end{align*}
Note that $\mathcal{E}_J\left(\vec{t}\, \right)=0$ and $D^2_\mathcal{V}\big|_{0} \mathcal{E}_J=0$. Therefore Eq.~(\ref{Equation: Proof with F_J}) follows from the fact that $F(t):=y(t)-c_J$ is Lipschitz and $F(0)=0$, which implies that there exists a constant $C$ such that $|F(t)|\leq C|t|$.
\end{proof}

For the proof of the following Theorem, we will use various operator estimates derived in Appendices \ref{Appendix: B} and \ref{Appendix: C}.
\begin{theorem}
\label{Theorem: Main A}
Let $J\in \{0,\dots,4\}^2$ be such that $\lambda_J\neq 0$ and let $R_J$ be the residuum defined in Eq.~(\ref{Equation: Taylor residuum}). Then,
\begin{align*}
R_J=o_*\left(\mathbb{T}_N\right),
\end{align*}
with $\mathbb{T}_N$ defined in Eq.~(\ref{Equation: Definition T}) and the $o_*(\cdot)$ notation from Definition \ref{Definition: Smallness}.\end{theorem}

\begin{proof}
Recall the definitions in Lemma \ref{Lemma: Taylor A} of $e_J:=\left|\{l\in J:l\in \{3,4\}\}\right|$, which counts how many of the indices in $J=(i,j)$ are equal to $3$ or $4$, $m_J:=\left|\{l\in J:l=0\}\right|$, which counts how many of the indices are zero, and the residuum $R_J$ defined in Eq.~(\ref{Equation: Taylor residuum}). In order to prove the statement of the Theorem, we are going to verify $R_J=o_*\left(\mathbb{T}_N\right)$ for all $J$ with $\lambda_J\neq 0$.

\textit{\textcolor{blue}{The case $J=(0,0)$}}: In this case we have the identity $R_{(0,0)}=\left(u_0^\dagger\cdot T\cdot u_0\right) E_2^2$, hence we have to verify $E_2^2=o_*\left(\mathbb{T}_N\right)$. In order to do this, recall the function $\eta_2(x)=1-\|F(x)\|^2$ from Eq.~(\ref{Equation: eta_m}) and let us compute using Lemma \ref{Lemma: Transformation Laws}
\begin{align*}
\nonumber 1-\mathbb{L}'&=1-\left(q+f(q)+ip'+b_{>d}\right)^\dagger\cdot \left(q+f(q)+ip'+b_{>d}\right)\\
&=1- q^\dagger\cdot  q- f(q)^\dagger\cdot  f(q)- f(q)^\dagger\cdot b_{>d}-b_{>d}^\dagger\cdot  f(q)\\
&\ \ \ \ -b_{>d}^\dagger\cdot b_{>d}-\left( p^\dagger\cdot  p\right)+\frac{d}{2N}- p^\dagger\cdot \left( p'- p\right)-\left( p'- p\right)^\dagger\cdot  p'\\
&=\eta_2(q)+D_\mathcal{V}\big|_q \eta_2\big(b_{\geq 1}\big)+D^2_\mathcal{V}\big|_0 \eta_2\big(b_{\geq 1}\big)+\frac{d}{4N}- p^\dagger\cdot \left( p'- p\right)-\left( p'- p\right)^\dagger\cdot  p',
\end{align*}
where we used $\eta_2(q)=1-q^\dagger\cdot q-f(q)^\dagger\cdot f(q)$ and $ p^\dagger\cdot  p=\frac{1}{4}\sum_{j=1}^d\left( 2b_j^\dagger b_j-b_j^2-\left(b_j^\dagger\right)^2\right)+\frac{d}{4N}$. Note that $\frac{c_2}{N}=\frac{d}{4N}$, where $c_2$ is the constant from Definition \ref{Definition: Taylor of the square root}. Since $ p^2\leq \mathbb{T}_N$, it is clear that $ p^2=O_*\left(\mathbb{T}_N\right)$. In Lemmata \ref{Lemma: p'} and \ref{Lemma: b}, we will verify that $\left( p'\right)^2=O_*\left(\mathbb{T}_N\right)$ and $\left( p'- p\right)^2=o_*\left(\mathbb{T}_N\right)$. Therefore we obtain by the operator Cauchy--Schwarz inequality in the auxiliary Lemma \ref{Lemma: O_* results} that $ p^\dagger\cdot \left( p'- p\right)$ as well as $\left( p'- p\right)^\dagger\cdot  p'$ are of order $o_*\left(\mathbb{T}_N\right)$. We conclude $E_2^2=o_*\left(\mathbb{T}_N\right)$.\\

\textit{\textcolor{blue}{The case $e_J=0$, with $J\neq (0,0)$}}: In this case $m_J\in \{0,1\}$ and the error is given by
\begin{align*}
R_J= h_i^\dagger\cdot  T\cdot  h_j\, E^1_{m_J}=e_i\left(q\right)^\dagger\cdot T\cdot e_j\left(q \right)\, E^1_{m_J}.
\end{align*}
We clearly have $E_0^1=0$. For $m_J=1$, let us define the function $V(t):=e_i\left(\vec{t}\ \right)^\dagger\cdot T\cdot e_j\left(\vec{t}\ \right)$, which satisfies $V(t)\leq C |t|$ for a constant $C$. In Lemma \ref{Lemma: Taylor of the square root} we will then verify that $V(q) E^1_{1}=o_*\left(\mathbb{T}_N\right)$.

\textit{\textcolor{blue}{The case $e_J=1$}}: In this case the error reads $R_J=\left( h^\dagger_i\cdot  T\cdot  h_j\right)\, E^0_{m_J}+\frac{F_J(q)}{N}$, where $F_J(t)\leq C|t|$ for some constant $C$. Using Lemma \ref{Lemma: Taylor of the square root} and Lemma \ref{Lemma: Function of q} from the Appendix, we obtain that $\left(E^0_1\right)^\dagger E^0_1=o_*\left(T_N\right)$ and $\frac{F_J(q)}{N}=o_*\left(\frac{1}{N}\right)$. Regarding the first term, note that $E^0_0=0$. Hence, we assume w.l.o.g. $m_J=1$. We are done once we can verify 
\begin{align}
\label{Equation: Part of main proof A}
\left( h^\dagger_i\cdot  T\cdot  h_j\right)\cdot \left( h^\dagger_i\cdot  T\cdot  h_j\right)^\dagger=O_*\left(\mathbb{T}_N\right)
\end{align}
in case one of the indices $i,j$ is in $\{3,4\}$ and the other is zero. Let us first assume $i\in \{3,4\}$. Then $ h^\dagger_i\cdot  T\cdot  h_j= h_i^\dagger\cdot w$, with $w:=T\cdot u_0\in \mathcal{H}$, and therefore Eq.~(\ref{Equation: Part of main proof A}) follows from Lemma \ref{Lemma: p'} in the case $i=3$ and from Lemma \ref{Lemma: b} in the case $i=4$. The proof of the case $j\in \{3,4\}$ follows analogously.

\textit{\textcolor{blue}{The case $e_J=2$}}: In this case, the error is a linear combination of $\left(i p'-i p\right)^\dagger\cdot  T\cdot  h_j$ and $ h_i^\dagger\cdot  T\cdot \left(i p'-i p\right)$ with $ h_i, h_j\in \{ p',b_{>d}\}$. Note that $A:=\sqrt{T} \left(1_{\mathcal{H}}-\pi_{>d}\right)$ is bounded, and therefore
\begin{align*}
&\left(i p'-i p\right)^\dagger\cdot  T\cdot \left(i p'-i p\right)=\left(i p'-i p\right)^\dagger\cdot  A^\dagger  A\cdot \left(i p'-i p\right)\\
&\ \ \ \ \leq \|A\|^2\ \left(i p'-i p\right)^\dagger\cdot \left(i p'-i p\right)=o_*\left(\mathbb{T}_N\right)
\end{align*}
by Lemma \ref{Lemma: b}. Similarly, we have $\left( p'\right)^\dagger\cdot  T\cdot  p'\leq \|A\|^2 (p')^\dagger\cdot p'=O_*\left(\mathbb{T}_N\right)$ by Lemma \ref{Lemma: p'}. Hence Lemma \ref{Lemma: O_* results} tells us that $\left(i p'- i p\right)^\dagger\cdot  T\cdot  h_j$ and $ h_i^\dagger\cdot  T\cdot \left(i p'- i p\right)$ are of order $o_*\left(\mathbb{T}_N\right)$.
\end{proof}

\begin{cor}
\label{Corollary: Main A}
Recall the functional $\mathcal{E}_A$ defined in Eq.~(\ref{Equation: A Decomposition classic}) and let us define the constant $c:=\sum_{J\in \{0,\dots,4\}^2}\lambda_J c_J$. Then 
\begin{align*}
\mathcal{W}_N \widetilde{A}_N \mathcal{W}_N^{-1}&=\mathcal{E}_A(q)+D_\mathcal{V}\big|_{q} \mathcal{E}_A\left(b_{\geq 1}\right)+\frac{1}{2}D^2_\mathcal{V}\big|_{0} \mathcal{E}_A\left(b_{\geq 1}\right)+\frac{c}{N}+o_*\left(\mathbb{T}_N\right).
\end{align*}
\end{cor}
\begin{proof}
The statement follows from combining Lemma \ref{Lemma: A Decomposition} and Theorem \ref{Theorem: Main A}.
\end{proof}

\subsection{Taylor Expansion of $\mathcal{W}_N \widetilde{B}_N \mathcal{W}_N^{-1}$}
\label{Subsection: Results in the transformed picture-B}
Similar to the previous subsection, we introduce atoms $H_J$ in Definition \ref{Definition: Hamilton Components B} as well as their classical counterparts $\mathcal{E}_J$ in Definition \ref{Definition: Functional Components B}. In Lemma \ref{Lemma: B Decomposition} we explain how $\mathcal{W}_N \widetilde{B}_N \mathcal{W}_N^{-1}$ and $\mathcal{E}_B$ can be written in terms of $H_J$ and $\mathcal{E}_J$, respectively.
\begin{defi}
\label{Definition: Hamilton Components B}
Recall the definition of $ h_i:\mathrm{dom}[\mathcal{N}]\longrightarrow \mathcal{F}_0\otimes \mathcal{H}$ from Definition \ref{Definition: Hamilton Components A}. For a multi-index $J=(i,j,k,\ell)$ with $i,j,k,\ell\in \{0,\dots,4\}$, we define an operator $H_J$ on $\mathcal{W}_N\mathcal{F}_{\leq N}$ as
\begin{align*}
H_J:&=\left( h_i\ \underline{\otimes}\  h_j\right)^\dagger \cdot  \hat{v}\cdot h_k\ \underline{\otimes}\  h_\ell  \left(1-\mathbb{L}'\right)^{\frac{m_J}{2}},
\end{align*}
where $m_J$ counts how many of the indices $i,j,k,\ell$ are zero.\\
\end{defi}

\begin{defi}
\label{Definition: Functional Components B}
Recall the definition of $e_i:\mathcal{H}_0 \longrightarrow \mathcal{H}$ and $\eta_m$ from Definition \ref{Definition: Functional Components A}. For a multi-index $J=(i,j,k,\ell)$ with $i,j,k,\ell\in \{0,\dots,4\}$, we define $\mathcal{E}_J:\mathcal{H}_0\cap \mathrm{dom}[T] \longrightarrow \mathbb{C}$
\begin{align*}
\mathcal{E}_J(z):&=\Big[e_i(z)\otimes e_j(z)\Big]^\dagger\cdot \hat{v}\cdot e_k(z)\otimes e_\ell(z)\ \eta_{m_J}(z),
\end{align*}
where $m_J$ counts how many of the indices $i,j,k,\ell$ are zero and $\eta_m$ are the functions defined in Eq.~(\ref{Equation: eta_m}).\\
\end{defi}

\begin{lem}
\label{Lemma: B Decomposition}
Let us define for all $i,j,k,\ell\in \{1,\dots,4\}$ the coefficients $\lambda_{(0,0,0,0)}:=\frac{1}{2}$, $\lambda_{(i,0,0,0)}:=2$, $\lambda_{(i,j,0,0)}:=\lambda_{(i,0,k,0)}:=\lambda_{(0,j,k,0)}:=1$, $\lambda_{(i,j,k,0)}:=2$, $\lambda_{(i,j,k,\ell)}:=\frac{1}{2}$ and all other coefficients are defined as $\lambda_J:=0$. Then
\begin{align*}
\mathcal{W}_N\, \widetilde{B}_N\, \mathcal{W}_N^{-1}&=\sum_{J\in \{0,\dots,4\}^4}\lambda_J\ \mathfrak{Re}\left[H_J\right].
\end{align*}
 Furthermore, the functional $\mathcal{E}_B$ defined as
\begin{align}
\label{Equation: B Decomposition classic}\mathcal{E}_B(z)&:=\sum_{J\in \{0,\dots,4\}^4}\lambda_J\ \mathfrak{Re}\left[\mathcal{E}_J(z)\right],
\end{align}
is an extension of $\mathcal{E}'_B\big|_{B_r}$ defined in Eq.~(\ref{Equation: Classic}), where $B_r:=\{z\in \mathcal{H}_0\cap \mathrm{dom}\left[T\right]:\|z\|<r\}$ and $r>0$ is a constant such that $\|F(z)\|<\frac{1}{2}$ for all $z\in \mathcal{H}_0$ with $\|z\|<r$.
\end{lem}

The proof of Lemma \ref{Lemma: B Decomposition} works analogously to the proof of Lemma \ref{Lemma: A Decomposition}. Following the strategy from Subsection \ref{Subsection: Results in the transformed picture-A} we are going to verify that the residuum $R_J$ 
\begin{align}
\label{Equation: Taylor residuum B}
R_J:&=H_J-\mathcal{E}_J(q)-D_\mathcal{V}\big|_{q} \mathcal{E}_J\left(b_{\geq 1}\right)-\frac{1}{2}D^2_\mathcal{V}\big|_{0} \mathcal{E}_J\left(b_{\geq 1}\right)-\frac{c_J}{N}
\end{align}
is small, where the constant $c_J$ are given by $c_{(0,0,0,0)}:=-\frac{1}{8}\sum_{j=1}^d \partial^2_{t_j}\big|_{t=0}\mathcal{E}_{(0,0,0,0)}\left(\vec{t}\ \right)$ and
\begin{align}
\nonumber c_{(3,3,0,0)}:&=-\frac{1}{8}\sum_{j=1}^d \partial_{t_j}^2\big|_{t=0}\mathcal{E}_{(1,1,0,0)}\left(\vec{t}\ \right),\ \ c_{(3,1,0,0)}:=-c_{(3,3,0,0)},\ \ \ c_{(1,3,0,0)}:=c_{(3,3,0,0)},\\
\nonumber c_{(3,0,3,0)}:&=\frac{1}{8}\sum_{j=1}^d \partial_{t_j}^2\big|_{t=0}\mathcal{E}_{(1,0,1,0)}\left(\vec{t}\ \right),\ \ \ \ c_{(1,0,3,0)}:=-c_{(3,0,3,0)},\ \ \ c_{(3,0,1,0)}:=-c_{(3,0,3,0)},\\
\label{Equation: c_J for B}c_{(0,3,3,0)}:&=\frac{1}{8}\sum_{j=1}^d \partial_{t_j}^2\big|_{t=0}\mathcal{E}_{(0,1,1,0)}\left(\vec{t}\ \right),\ \ \ \ c_{(0,1,3,0)}:=-c_{(0,3,3,0)},\ \ \ c_{(0,3,1,0)}:=-c_{(0,3,3,0)},
\end{align}
and all other constants are defined as $c_J:=0$. The proof will be split into two parts. In Lemma \ref{Lemma: Taylor B} we derive an explicit representation of the residuum $R_J$ by sorting the operator $H_J$ in terms of powers in $p$ and $b_{>d}$, and in Theorem \ref{Theorem: Main B} we will make sure that this residuum is indeed small compared to the operator $\mathbb{T}_N$.\\

\begin{lem}
\label{Lemma: Taylor B}
Let $J=(i,j,k,\ell)\in \{0,\dots,4\}^4$ be such that $\lambda_J\neq 0$, where $\lambda_J$ is defined in Lemma \ref{Lemma: B Decomposition}, and let $R_J$ be the residuum defined in Eq.~(\ref{Equation: Taylor residuum B}). By distinguishing different cases with the help of the indices $e_J:=\left|\{\ell\in J:\ell\in \{3,4\}\}\right|$ and $m_J:=\left|\{\ell\in J:\ell=0\}\right|$, we can explicitly express $R_J$ as:
\begin{itemize}
\item In the case $m_J=4$, i.e. $J=(0,0,0,0)$: $R_J=\left( h_i\ \underline{\otimes}\  h_j\right)^\dagger\cdot  \hat{v}\cdot  h_k\ \underline{\otimes}\  h_\ell\, E_{4}^2$.
\item In the case $e_J=0$ and $m_J<4$: $R_J=\left( h_i\ \underline{\otimes}\  h_j\right)^\dagger\cdot  \hat{v}\cdot  h_k\ \underline{\otimes}\  h_\ell\, E_{m_J}^1$.
\item In the case $e_J=1$, there exists a constant $C$ and functions $F_J:\mathbb{R}^d\longrightarrow \mathbb{R}$ with $|F_J(t)|\leq C|t|$, such that $R_J=\left( h_i\ \underline{\otimes}\  h_j\right)^\dagger\cdot  \hat{v}\cdot  h_k\ \underline{\otimes}\  h_\ell\, E_{m_J}^0+\frac{F_J(q)}{N}$.
\item In the case $e_J=2$ and $m_J=2$ when two of the indices are $4$:
\begin{align*}
R_J=-\left( h_i\ \underline{\otimes}\  h_j\right)^\dagger\cdot  \hat{v}\cdot  h_k\ \underline{\otimes}\  h_\ell\, \mathbb{L}'. \ \ \ \ \ \ \ \ \ \ \ \ \ \ \ \ \ \ \ \ \ \ \ \ \ \ \ 
\end{align*}
\item In the case $e_J=2$ and $m_J=2$ when one of the indices is $3$ and another one is $4$, let us define $\widetilde{ h}_3:= p'- p$ and $\widetilde{ h}_r:= h_r$ for $r\in \{0,1,2,4\}$. Then,
\begin{align*}
 \ \ \ \ \ \ R_J=-\left( h_i\ \underline{\otimes}\  h_j\right)^\dagger\cdot   \hat{v}\cdot  h_k\ \underline{\otimes}\  h_\ell \, \mathbb{L}'+\left(\widetilde{ h}_i\ \underline{\otimes}\ \widetilde{ h}_j\right)^\dagger\cdot   \hat{v}\cdot \widetilde{ h}_k\ \underline{\otimes}\ \widetilde{ h}_\ell.
\end{align*}
\item In the case $e_J=2$ and $m_J=2$ when two of the indices are $3$, let us define the coefficients $\Lambda_{(3,3,0,0)}^{r,r'}:=-\left(u_r\otimes u_{r'}\right)^\dagger \cdot \hat{v}\cdot u_0\otimes u_0$, $\Lambda_{(3,0,3,0)}^{r,r'}:=\left(u_r\otimes u_0\right)^\dagger \cdot \hat{v}\cdot u_{r'}\otimes u_0$ and $\Lambda_{(0,3,3,0)}^{r,r'}:=\left(u_0\otimes u_{r}\right)^\dagger \cdot \hat{v}\cdot u_{r'}\otimes u_0$. Then,
\begin{align}
\label{Equation: Proof of 03}
R_J\!=&\!-\!\left( h_i\, \underline{\otimes}\,  h_j\right)^\dagger\cdot   \hat{v}\cdot  h_k\, \underline{\otimes}\,  h_\ell \, \mathbb{L}'\!+\!\sum_{r,r'=1}^d\! \Lambda_J^{r,r'}\left[(p'_r\!-\!p_r)\cdot p'_{r'}\!\!+p_r\cdot (p'_{r'}\!-\!p_{r'})\right].
\end{align}
\item In the cases $e_J=2$ and $m_J<2$, respectively $e_J>2$: $R_J=H_J$.
\end{itemize}
\end{lem}
\begin{proof}
Similar to the proof of Lemma \ref{Lemma: Taylor A}, the proof of Lemma \ref{Lemma: Taylor B} follows from a straightforward computation for the individual cases. For the purpose of illustration, we will explicitly carry out the computations for the case $J=(3,0,3,0)$, i.e. we are going to verify Eq.~(\ref{Equation: Proof of 03}). Since $\mathcal{E}_{(3,0,3,0)}(z)$ is quadratic in $\pi(z)$, we immediately obtain $\mathcal{E}_{(3,0,3,0)}\left(\vec{t}\ \right)=0$ and $D_\mathcal{V}\big|_{\vec{t}}\ \mathcal{E}_{(3,0,3,0)}=0$. Let us define the coefficients $\lambda_{\alpha, \gamma}:=\left(u_\alpha\otimes u_{0}\right)^\dagger\cdot \hat{v}\cdot u_{\gamma}\otimes u_0$, the operator $Q=\frac{1}{2}\sum_{\alpha,\gamma=1}^d \hat{v}_{\alpha 0,\gamma 0}\ u_{\alpha}\cdot u_{\gamma}^\dagger$ and $G\in \mathcal{H}_0\otimes \mathcal{H}_0$ by $G=-\frac{1}{4}\sum_{\alpha,\gamma=1}^d \hat{v}_{\alpha 0,\gamma 0}\ u_\alpha\otimes u_\gamma$. Then 
\begin{align*}
D^2_\mathcal{V}\big|_{0}\ \mathcal{E}_{(3,0,3,0)}(z)=z^\dagger\cdot Q\cdot z+G^\dagger\cdot z\otimes z+\left(z\otimes z\right)\cdot G
\end{align*}
and therefore $D^2_\mathcal{V}\big|_{0}\ \mathcal{E}_{(3,0,3,0)}\left(b_{\geq 1}\right)=b_{\geq 1}^\dagger\cdot Q\cdot b_{\geq 1}+G^\dagger\cdot b_{\geq 1}\otimes b_{\geq 1}+\left(b_{\geq 1}\otimes b_{\geq 1}\right)\cdot G$. This concludes the proof of Eq.~(\ref{Equation: Proof of 03}), since
\begin{align*}
&H_J-\left(i p'\ \underline{\otimes}\  u_0\right)^\dagger\cdot   \hat{v}\cdot i p'\ \underline{\otimes}\  u_0 \left(-\mathbb{L}'\right)-\sum_{r,r'=1}^d\Lambda_J^{r,r'}\left[(p'_r-p_r)\cdot p'_{r'}+p_r\cdot (p'_{r'}-p_{r'})\right]\\
&\ =\left(i p\ \underline{\otimes}\  u_0\right)^\dagger\cdot   \hat{v}\cdot i p\ \underline{\otimes}\  u_0\!=\!\sum_{\alpha,\gamma=1}^d\! \left(u_\alpha\otimes u_0\right)^\dagger\cdot \hat{v}\cdot u_\gamma\otimes u_0\ \frac{1}{2}\left(b_\alpha-b_\alpha^\dagger\right)^\dagger\cdot \frac{1}{2}\left(b_\gamma-b_{\gamma}^\dagger\right)\\
&\ =b_{\geq 1}^\dagger\cdot   Q\cdot b_{\geq 1}+\left(b_{\geq 1}\ \underline{\otimes}\ b_{\geq 1}\right)^\dagger\cdot   G+  G^\dagger\cdot b_{\geq 1}\ \underline{\otimes}\ b_{\geq 1}+\frac{c_{(3,0,3,0)}}{N}.
\end{align*}

\end{proof}

In the remainder of this subsection, we are going to verify that the residuum $R_J$ is small compared to the quadratic operator $\mathbb{T}_N$. Note that the error term in the last case of Lemma \ref{Lemma: Taylor B} is quite different from the other cases, since it simply corresponds to the whole operator $H_J$. This is not surprising, however, since the second order Taylor approximation of an object that is already of an higher order than two is zero, i.e. the residuum coincides with the object itself. With the help of the following three results in Lemma \ref{Lemma: Smallness of quadratic terms}, Lemma \ref{Lemma: Control of G-terms} and Theorem \ref{Theorem: High Order Estimates}, we will systematically verify that $H_J$ is small compared to the quadratic operator $\mathbb{T}_N$ in the cases $e_J=2$ and $m_J<2$, respectively $e_J>2$. Regarding all other cases, we will verify the smallness of the residuum in Theorem \ref{Theorem: Main B}. In order to do this, we will repeatedly use results derived in Appendices \ref{Appendix: B} and \ref{Appendix: C}.\\

\begin{lem}
\label{Lemma: Smallness of quadratic terms}
For indices $i,j\in \{0,\dots,4\}$, we have the following estimates:
\begin{itemize}
\item In case one of the indices is contained in $\{3,4\}$, we have
\begin{align*}
\left( h_i\ \underline{\otimes}\  h_j\right)^\dagger\cdot   |\hat{v}|\cdot  h_i\ \underline{\otimes}\  h_j = O_*\left(\mathbb{T}_N\right).
\end{align*}
\item In case one of the indices is contained in $\{3,4\}$ and the other one is contained in $\{1,\dots,4\}$, we have
\begin{align*}
\left( h_i\ \underline{\otimes}\  h_j\right)^\dagger\cdot   |\hat{v}|\cdot  h_i\ \underline{\otimes}\  h_j=o_*\left(\mathbb{T}_N\right).
\end{align*}
\end{itemize}
\end{lem}
\begin{proof}
We will repeatedly use the inequality $|v|\leq \Lambda\left(T+1\right)=:S$ from Assumption \ref{Assumption: Part I}, which implies together with the translation-invariance of $T$ the inequalities $|\hat{v}|\leq S\otimes 1_{\mathcal{H}}$ and $|\hat{v}|\leq\ 1_{\mathcal{H}}\otimes S$.\\

\textit{\textcolor{blue}{The case $i\in \{1,2\}$ and $j\in \{3,4\}$}}: Recall that $h_k=e_k(q)$ for $k\in \{0,1,2\}$ and let us define the function $\phi(t):=e_i\left(\vec{t}\ \right)^\dagger\cdot S\cdot e_i\left(\vec{t}\ \right)$. Using the inequality $|\hat{v}|\leq S\otimes 1_{\mathcal{H}}$ we obtain 
\begin{align*}
\left( e_i(q)\ \underline{\otimes}\  h_j\right)^\dagger \cdot   |\hat{v}|\cdot  e_i(q)\ \underline{\otimes}\  h_j\leq  h_j^\dagger\cdot \phi(q)\cdot  h_j.
\end{align*}
Since $|\phi(t)|\leq C \left(|t|+|t|^2\right)$ for a constant $C$, we obtain $h_3^\dagger\cdot \phi(q)\cdot h_3=o_*\left(\mathbb{T}_N\right)$ and $h_4^\dagger\cdot \phi(q)\cdot h_4=o_*\left(\mathbb{T}_N\right)$ by Lemmata \ref{Lemma: Function of q} and \ref{Lemma: p'}.\\

\textit{\textcolor{blue}{The case $i\in \{3,4\}$ and $j\in \{1,2\}$}}: Making use of the commutation laws $[b_\alpha,q_\beta]=0$ and $[ip_\alpha,q_\beta]=\frac{1}{2N}\delta_{\alpha,\beta}$, this case follows from the previous one.\\

\textit{\textcolor{blue}{The case $i=3,j=3$}}: Let $\pi_{\leq d}:=\sum_{r=1}^d u_r\cdot u_r^\dagger$. Since $\left(( p')^\dagger\cdot  p'\right)^2=o_*\left(\mathbb{T}_N\right)$, we obtain
\begin{align*}
\left( p'\ \underline{\otimes}\  p'\right)^\dagger\cdot   |\hat{v}|\cdot  p'\ \underline{\otimes}\  p'\leq ( p')^\dagger\cdot \left(( p')^\dagger\cdot  p'\right)\otimes S\cdot  p'\leq \big\|\pi_{\leq d}\, S\, \pi_{\leq d}\big\| \left(( p')^\dagger\cdot  p'\right)^2=o_*\left(\mathbb{T}_N\right).
\end{align*}

\textit{\textcolor{blue}{The case $i=4$ and $j\in \{3,4\}$}}: Note that
\begin{align*}
\left(b_{>d}\ \underline{\otimes}\  h_j\right)^\dagger\cdot   |\hat{v}| &\cdot b_{>d}\ \underline{\otimes}\  h_j \leq 2\left( f(q)\ \underline{\otimes}\  h_j\right)^\dagger\cdot   |\hat{v}|\cdot  f(q)\ \underline{\otimes}\  h_j\\
& +2\left((b_{>d}+ f(q))\ \underline{\otimes}\  h_j\right)^\dagger\cdot   |\hat{v}|\cdot (b_{>d}+ f(q))\ \underline{\otimes}\  h_j
\end{align*}
By the previous case $i\in \{1,2\}$ and $j\in \{3,4\}$, we know that $\left( f(q)\ \underline{\otimes}\  h_j\right)^\dagger\cdot   |\hat{v}|\cdot  f(q)\ \underline{\otimes}\  h_j=o_*\left(\mathbb{T}_N\right)$. For the second contribution, recall the definition of $\pi_{M,N}$ from Remark \ref{Remark: Projection Formalism} and let $ \hat{\pi}_{M,N}$ be the orthogonal projection onto the subspace $\mathcal{W}_N\mathcal{F}^+_{\leq M}\subset \mathcal{F}_{0}$, where
\begin{align}
\label{Equation: Definition of hat Pi}
\mathcal{F}^+_{\leq M}:=\mathds{1}_{[0,M]}\left(\sum_{j>d}^\infty a_j^\dagger\cdot a_j\right).
\end{align}
Since we have $b_{k}\, \hat{\pi}_{M,N}=\hat{\pi}_{M,N}\, b_{k}\, \hat{\pi}_{M,N}$ for $k>d$ and $[p'_j,\hat{\pi}_{M,N}]=0$ by Lemma \ref{Lemma: Auxiliary}, we obtain using $\pi_{M,N}=\hat{\pi}_{M,N}\, \pi_{M,N} $ (which follows from $\mathcal{W}_N\mathcal{F}_{\leq M}\subset \mathcal{W}_N\mathcal{F}^+_{\leq M}$)
\begin{align*}
&\pi_{M,N} \left((b_{>d}+ f(q))\ \underline{\otimes}\  h_j\right)^\dagger\cdot   |\hat{v}|\cdot (b_{>d}+ f(q))\ \underline{\otimes}\  h_j\, \pi_{M,N}\\
&\ \ \ \leq \pi_{M,N}\, h_j^\dagger\cdot \left(\hat{\pi}_{M,N}\, (b_{>d}+ f(q))^\dagger\cdot (b_{>d}+ f(q))\, \hat{\pi}_{M,N}\right)\otimes S\cdot  h_j\, \pi_{M,N}\\ 
&\ \ \ \leq \frac{M}{N}\ \pi_{M,N}\,  h_j^\dagger\cdot   S \cdot h_j\, \pi_{M,N}\leq C\ \frac{M}{N}\ \pi_{M,N}\, \mathbb{T}_N\, \pi_{M,N}
\end{align*}
for a constant $0<C<\infty$, where we used $ h_j^\dagger\cdot   S \cdot h_j=O_*\left(\mathbb{T}_N\right)$ and and the characterization of the $O_*(\cdot)$ notation in Remark \ref{Remark: Projection Formalism} for the last inequality. Using this characterization for the inequality above yields $ \left((b_{>d}+ f(q))\ \underline{\otimes}\  h_j\right)^\dagger\cdot   |\hat{v}|\cdot (b_{>d}+ f(q))\ \underline{\otimes}\  h_j=o_*\left(\mathbb{T}_N\right)$.\\

\textit{\textcolor{blue}{The case $i=3$ and $j=4$}}: Making use of the commutation laws $[ip'_\alpha,b]=-\frac{1}{2N}\partial_\alpha f(q)$, this case follows from the previous one.\\

\textit{\textcolor{blue}{The case $i=0$ and $j\in \{3,4\}$, respectively $i\in \{3,4\}$ and $j=0$}}: Since $ h_0= 1_{\mathcal{F}_0}\otimes u_0$ commutes with $h_3=i p'$ and $h_4=b_{>d}$, we assume w.l.o.g. $i=0$ and $j\in \{3,4\}$. With $\lambda:=u_0^\dagger\cdot S\cdot u_0$, we obtain
\begin{align*}
\left( u_0\ \underline{\otimes}\  h_j\right)^\dagger \cdot   |\hat{v}|\cdot  u_0\ \underline{\otimes}\  h_j\leq \lambda\  h_j^\dagger\cdot  h_j=O_*\left(\mathbb{T}_N\right).
\end{align*}
\end{proof}

\begin{lem}
\label{Lemma: Control of G-terms}
In the following, let $G:\mathbb{R}^d\longrightarrow \mathcal{H}\otimes \mathcal{H}$ be a differentiable function and let us define the operators $X,Y:\mathrm{dom}[\mathcal{N}]\longrightarrow {\mathcal{F}_0}\otimes \mathcal{H}$ as 
\begin{align*}
X&:=\left(i p'\right)^\dagger\otimes 1_{\mathcal{H}}\cdot G(q)=-\sum_{k=1}^\infty \left(\sum_{j=1}^d iG_{j,k}p'_j \right)\otimes u_k,\\
Y&:= b_{>d}^\dagger\otimes 1_{\mathcal{H}}\cdot G(q)=\sum_{k=1}^\infty \left(\sum_{j>d} iG_{j,k}b^\dagger_j \right)\otimes u_k.
\end{align*}
Then we have the estimates
\begin{align*}
X^\dagger\cdot X &\leq 2\left[\left( p'\right)^\dagger\cdot \|G\|^2(q)\cdot p'+\frac{d}{N^2}\sum_{\alpha=1}^d\|\partial_\alpha G\|^2(q)\right],\\
Y^\dagger\cdot Y &\leq  b_{>d}^\dagger\cdot \|G\|^2(q)\cdot b_{>d}+\frac{1}{N}\|G\|^2(q).
\end{align*}
\end{lem}
The proof of Lemma \ref{Lemma: Control of G-terms} is based on the commutation relations $[b_\alpha,b_\beta^\dagger]=\frac{1}{N}\delta_{\alpha,\beta}$ and $[p_\alpha,q_\beta]=\frac{1}{2iN}\delta_{\alpha,\beta}$, and is left to the reader.

\begin{theorem}
\label{Theorem: High Order Estimates}
Let $\mathbb{L}'$ be the operator from Definition \ref{Definition: Unitary Transformation}. Then we have the following estimates:
\begin{itemize}
\item In case at least two of the indices $i,j,k,\ell\in \{1,\dots,4\}$ are contained in $\{3,4\}$, we have
\begin{align*}
\left( h_i\ \underline{\otimes}\  h_j\right)^\dagger\cdot   \hat{v}\cdot  h_k\ \underline{\otimes}\  h_\ell=o_*\left(\mathbb{T}_N\right),
\end{align*}
\item In case at least two of the indices $i,j,k\in \{0,\dots,4\}$ are contained in $\{3,4\}$, we have
\begin{align*}
\left( h_i\ \underline{\otimes}\  h_j\right)^\dagger\cdot   \hat{v}\cdot  h_k\ \underline{\otimes}\  u_0\, \mathbb{L}'=o_*\left(\mathbb{T}_N\right),
\end{align*}
\item In case at least two of the indices $i,j,k\in \{1,\dots,4\}$ are contained in $\{3,4\}$, we have
\begin{align*}
\left( h_i\ \underline{\otimes}\  h_j\right)^\dagger\cdot   \hat{v}\cdot  h_k\ \underline{\otimes}\  u_0\, \sqrt{1-\mathbb{L}'}=o_*\left(\mathbb{T}_N\right).
\end{align*}
\end{itemize}
\end{theorem}

\begin{proof}
Let us denote with $e_{(a,b)}$ the number of indices in $(a,b)$ that are elements of $\{3,4\}$. In the following, we will verify the theorem separately for the case $e_{(i,j)}\geq 1$ and $e_{(k,\ell)}\geq 1$, and the case $e_{(k,\ell)}=0$. Note that the case $e_{(i,j)}=0$ is only possible for the first bullet point, and the proof of the statement follows from the case $e_{(k,\ell)}=0$, since
\begin{align*}
\left[\left( h_i\ \underline{\otimes}\  h_j\right)^\dagger\cdot   \hat{v}\cdot  h_k\ \underline{\otimes}\  h_\ell\right]^\dagger=\left( h_k\ \underline{\otimes}\  h_\ell\right)^\dagger\cdot   \hat{v}\cdot  h_i\ \underline{\otimes}\  h_j.
\end{align*}

\textit{\textcolor{blue}{The case $e_{(i,j)}\geq 1$ and $e_{(k,\ell)}\geq 1$}}:
Let us define the operators $A:= h_i\ \underline{\otimes}\  h_j$ and $Q:=  \hat{v}$, and depending on the concrete bullet point let us define $B$ as $ h_k\ \underline{\otimes}\  h_\ell$, $ h_k\ \underline{\otimes}\  u_0\, \mathbb{L}'$ or $ h_k\ \underline{\otimes}\  u_0\, \sqrt{1-\mathbb{L}'}$. In any case we have to verify
\begin{align*}
A^\dagger\cdot Q\cdot B=o_*\left(\mathbb{T}_N\right).
\end{align*}
By Lemma \ref{Lemma: O_* results}, it is enough to verify that one of the operators $A^\dagger\cdot |Q|\cdot A$ and $B^\dagger\cdot |Q|\cdot B$ is of order $o_*\left(\mathbb{T}_N\right)$, and the other one is of order $O_*\left(\mathbb{T}_N\right)$, which follows from Lemma \ref{Lemma: Smallness of quadratic terms} and the auxiliary Corollary \ref{Corollary: Lifting}.\\

\textit{\textcolor{blue}{The case $e_{(k,\ell)}=0$}}: In this case we have $i,j\in \{3,4\}$ for any of the bullet points. Let us define the function $G:\mathbb{R}^d\longrightarrow \mathcal{H}\otimes \mathcal{H}$ by $G(t):=1_{\mathcal{H}}\otimes \left(T+1\right)^{-\frac{1}{2}}\cdot \hat{v}\cdot e_k\left(\vec{t}\ \right)\otimes e_{\ell}\left(\vec{t}\ \right)$. Note that $G(t)\in \mathcal{H}\otimes \mathcal{H}$ follows from Assumption \ref{Assumption: Part I}. We define the operator $X:=   \left(T+1\right)^{\frac{1}{2}}\cdot  h_i$ and depending on the concrete bullet point let us define $Y:= h^\dagger_j\otimes 1_{\mathcal{H}}\cdot G(q)\, Z$ with $Z:=1_{\mathcal{F}_0}$, $Z:=\mathbb{L}'$ or $Z:=\sqrt{1-\mathbb{L}'}$. In the following, we have to verify $X^\dagger\cdot Y=o_*\left(\mathbb{T}_N\right)$. Since $i\in \{3,4\}$, we know that $X^\dagger\cdot X=O_*\left(\mathbb{T}_N\right)$. By the Cauchy--Schwarz like result in Lemma \ref{Lemma: O_* results}, it is therefore enough to verify $Y^\dagger\cdot Y=o_*\left(\mathbb{T}_N\right)$. Applying Lemma \ref{Lemma: Control of G-terms} yields in any case
\begin{align*}
Y^\dagger\cdot Y&=Z^\dagger \left( h^\dagger_j\otimes 1_{\mathcal{H}}\cdot G(q)\right)^\dagger\cdot \left( h^\dagger_j\otimes 1_{\mathcal{H}}\cdot G(q)\right) Z\\
&\leq 2\ Z^\dagger \left[ h_j^\dagger \cdot \|G(q)\|^2\cdot  h_j+\frac{1}{N}\|G(q)\|^2+\frac{d}{N^2}\sum_{r=1}^d \|\partial_r G(q)\|^2\right] Z,
\end{align*}
and Corollary \ref{Corollary: Lifting} then yields that $Z^\dagger\, \|G(q)\|^2\, Z$ and $Z^\dagger\, \frac{1}{N}\left(\sum_{r=1}^d \|\partial_r G(q)\|^2\right) Z$ are of order $o_*(1)$. Therefore, $Z^\dagger\,\frac{1}{N}\|G(q)\|^2 Z$ and $Z^\dagger\,\frac{d}{N^2}\left(\sum_{r=1}^d \|\partial_r G(q)\|^2\right) Z$ are both of order $o_*\left(\mathbb{T}_N\right)$. Finally, $Z^\dagger\,  h_j^\dagger \cdot \|G(q)\|^2\cdot  h_j\, Z=o_*\left(\mathbb{T}_N\right)$ follows from the auxiliary Lemmata \ref{Lemma: Function of q} and \ref{Lemma: p'}, and the auxiliary Corollary \ref{Corollary: Lifting}.
\end{proof}

\begin{theorem}
\label{Theorem: Main B}
Let $J\in \{0,\dots,4\}^4$ be such that $\lambda_J\neq 0$ and let $R_J$ be the residuum defined in Eq.~(\ref{Equation: Taylor residuum B}). Then,
\begin{align*}
R_J&=o_*\left(\mathbb{T}_N\right).
\end{align*}
\end{theorem}
\begin{proof}
Let $J=(i,j,k,\ell)$ be a multi index with $\lambda_J\neq 0$, and recall the index $e_J:=\left|\{l\in J:l\in \{3,4\}\}\right|$ and the index $m_J:=\left|\{l\in J:l=0\}\right|$ from Lemma \ref{Lemma: Taylor B} as well as the residuum defined in Eq.~(\ref{Equation: Taylor residuum B}). In order to prove the statement of the Theorem, we have to verify $R_J=o_*\left(\mathbb{T}_N\right)$ for all $J\in \{0,\dots,4\}^4$.

\textit{\textcolor{blue}{The case $e_J=0$ and $m_J=0$}}: In this case we have a trivial residuum $R_J=0$.\\

\textit{\textcolor{blue}{The case $e_J=0$ and $m_J=1$}}: In this case, $R_J=V(q) E_1^1$, with $V(t):=\big(e_i\left(\vec{t}\ \right)\otimes e_j\left(\vec{t}\ \right)\big)^\dagger\cdot \hat{v}\cdot e_k\left(\vec{t}\ \right)\otimes e_{\ell}\left(\vec{t}\ \right)$. Since the $C^1$ function $V$ satisfies $F(0)=0$, we obtain $V(q) E_1^1=o_*\left(\mathbb{T}_N\right)$ by Lemma \ref{Lemma: Taylor of the square root}.\\

\textit{\textcolor{blue}{The case $e_J=0$ and $m_J=2$}}:
 In this case $R_J=V(q) E_2^1$, with $V(t):=\big(e_i\left(\vec{t}\ \right)\otimes  e_j\left(\vec{t}\ \right)\big)^\dagger\cdot \hat{v}\cdot e_k\left(\vec{t}\ \right)\otimes e_{\ell}\left(\vec{t}\ \right)$. We compute
\begin{align*}
E_2^1&=\left(1-\mathbb{L}'\right)-\eta_2(q)- D_{\mathcal{V}}\big|_q \eta_2\big(b_{\geq 1}\big)=-\left(b_{>d}^\dagger\cdot b_{>d}+( p')^\dagger\cdot  p'-\frac{d}{2N}\right).
\end{align*}
By Lemmata \ref{Lemma: Function of q} and \ref{Lemma: p'} we know that $V(q)\, b_{>d}^\dagger\cdot b_{>d}=V(q)\, b_{>d}^\dagger\cdot b_{>d}$ and $V(q) \left( p'\right)^\dagger\cdot  p'$ are of order $o_*\left(\mathbb{T}_N\right)$, and consequently $V(q) E_2^1=o_*\left(\mathbb{T}_N\right)$.\\

\textit{\textcolor{blue}{The case $e_J=0$ and $m_J=3$}}: In this case $R_J=V(q) E_3^1$, with $V(t):=\big(e_i\left(\vec{t}\ \right)\otimes  e_j\left(\vec{t}\ \right)\big)^\dagger\cdot \hat{v}\cdot e_k\left(\vec{t}\ \right)\otimes e_{\ell}\left(\vec{t}\ \right)$. We compute
\begin{align*}
E_3^1:&=\left(1-\mathbb{L}'\right) \sqrt{1-\mathbb{L}'}-\eta_3(q)- D_\mathcal{V}\big|_q \eta_3\big(b_{\geq q}\big)\\
&=(1\!-\!\eta_2(q)) E_1^1\!-\!\left( f(q)^\dagger\!\cdot\! b_{>d}\!+\!b_{>d}\!\cdot\!  f(q)\right) E_1^0\!-\!\left(\!b_{>d}^\dagger\!\cdot\! b_{>d}\!+\!( p')^\dagger\!\cdot\!  p'\!-\!\frac{d}{2N}\!\right)\sqrt{1\!-\!\mathbb{L}'}.
\end{align*}
By Lemma \ref{Lemma: Taylor of the square root}, we know that $V(q) (1-\eta_2(q)) E_1^1=o_*\left(\mathbb{T}_N\right)$ and $\left(E_1^0\right)^2=o_*\left(\mathbb{T}_N\right)$. Note that we further have $\left[V(q) \left( f(q)^\dagger\cdot b_{>d}+b_{>d}\cdot  f(q)\right)\right]^2=o_*\left(\mathbb{T}_N\right)$, and therefore the product $V(q) \left( f(q)^\dagger\cdot b_{>d}+b_{>d}\cdot  f(q)\right) E_1^0$ is of order $o_*\left(\mathbb{T}_N\right)$ as well. By making use of Lemmata \ref{Lemma: Function of q} and \ref{Lemma: p'}, and Corollary \ref{Corollary: Lifting}, we obtain
\begin{align*}
V(q)\left(\!b_{>d}^\dagger\!\cdot\! b_{>d}\!+\!( p')^\dagger\!\cdot\!  p'\!-\!\frac{d}{2N}\!\right)\sqrt{1\!-\!\mathbb{L}'}=o_*\left(\mathbb{T}_N\right).
\end{align*}

\textit{\textcolor{blue}{The case $e_J=0$ and $m_J=4$}}: In this case $R_J=\left(u_0\otimes u_0\right)\cdot \hat{v}\cdot u_0\otimes u_0\ E_4^2$. We compute
\begin{align*}
E_4^2:&=\left(1-\mathbb{L}'\right)^2-\eta_4(q)- D_\mathcal{V}\big|_q \eta_4\big(b_{\geq 1}\big)-D^2_\mathcal{V}\big|_0 \eta_4\big(b_{\geq 1}\big)-\frac{c_4}{N}\\
&=\left( f(q)^\dagger\cdot b_{>d}+b_{>d}\cdot  f(q)\right)^2+\Big\{ f(q)^\dagger\cdot b_{>d}+b_{>d}\cdot  f(q),\ \left( p'\right)^\dagger\cdot  p'+b_{>d}^\dagger\cdot b_{>d}\Big\}\\
&\ \ +\left(\left( p'\right)^\dagger\cdot  p'+b_{>d}^\dagger\cdot b_{>d}\right)^2+\Big\{\eta_2(q),\ \left( p'\right)^\dagger\cdot  p'+b_{>d}^\dagger\cdot b_{>d}\Big\}\\
&\ \ +2 p^\dagger\cdot \left( p'- p\right)+2\left( p'- p\right)^\dagger\cdot  p'-\frac{d}{N}\mathbb{L}'+\frac{d^2}{4N^2},
\end{align*}
with the notation $\{A,B\}:=A B+B A$. Clearly $\frac{d}{N}\mathbb{L}'=o_*\left(\mathbb{T}_N\right)$. From Lemmata \ref{Lemma: Function of q}, \ref{Lemma: b} and \ref{Lemma: p'}, we know that all the operators $ p^\dagger\cdot \left( p'- p\right)$, $\left( p'- p\right)^\dagger\cdot  p'$, $\left(( p')^\dagger\cdot  p'\right)^2$, $\left(b_{>d}^\dagger\cdot b_{>d}\right)^2$, $\left(b_{>d}^\dagger\cdot  f(q)+ f(q)^\dagger\cdot b_{>d}\right)^2$, $\eta_2(q)\, b_{>d}^\dagger\cdot b_{>d}$ and $\eta_2(q)\, \left( p'\right)^\dagger\cdot  p'\, \eta_2(q)$ are of order $o_*\left(\mathbb{T}_N\right)$. Consequently, $\{\eta_2(q),\left( p'\right)^\dagger\cdot  p'+b_{>d}^\dagger\cdot b_{>d}\}$ and $\Big\{ f(q)^\dagger\cdot b_{>d}+b_{>d}\cdot  f(q),\ \left( p'\right)^\dagger\cdot  p'+b_{>d}^\dagger\cdot b_{>d}\Big\}$ are of order $o_*\left(\mathbb{T}_N\right)$ as well.\\

\textit{\textcolor{blue}{The case $e_J=1$}}: In this case, we have $R_J=\left( h_i\ \underline{\otimes}\  h_j\right)^\dagger\cdot   \hat{v}\cdot  h_k\ \underline{\otimes}\  h_\ell\, E_{m_J}^0+\frac{F_J(q)}{N}$. By Lemma Lemma \ref{Lemma: Function of q}, we know that $\frac{F_J(q)}{N}=o_*\left(\mathbb{T}_N\right)$. Since we know that $\left(E_{m_J}^0\right)^2=o_*\left(\mathbb{T}_N\right)$ by Corollary \ref{Corollary: Taylor of the square root}, we are done once we can verify that $X_J\cdot X_J^\dagger=O_*\left(\mathbb{T}_N\right)$, where $X_J:=\left( h_i\ \underline{\otimes}\  h_j\right)^\dagger\cdot   \hat{v}\cdot  h_k\ \underline{\otimes}\  h_\ell$.\\

\textit{With $3\in \{i,j,k,\ell\}$}: Let us first assume $j=3$, and define $w(t):=e_i\left(\vec{t}\ \right)^\dagger\otimes 1_{\mathcal{H}}\cdot \hat{v}\cdot e_k\left(\vec{t}\ \right)\otimes e_{\ell}\left(\vec{t}\ \right)$. Clearly, $X_J=\left(i p'\right)^\dagger\cdot  w(q)$ and therefore $X_J X_J^\dagger=O_*\left(\mathbb{T}_N\right)$ follows from \ref{Lemma: p'}. The other cases $i=3,k=3$ and $\ell=3$ follows from the commutation relation $[ip'_\alpha,q_\beta]=\frac{1}{2N}\delta_{\alpha,\beta}$.\\

\textit{With $4\in \{i,j,k,\ell\}$}: In any case, $X_J$ is either equal to $ w(q)^\dagger\cdot b_{>d}$ or $b^\dagger_{>d}\cdot  w(q)$, where $w:\mathbb{R}^d\longrightarrow \mathcal{H}$ with $\|w(t)\|\leq c|t|^j$ and $j\geq 0$. Note that we use the commutativity of $q_j$ and $b_{>d}$ here. Therefore, Lemma \ref{Lemma: b} implies $X_J\cdot X_J^\dagger=O_*\left(\mathbb{T}_N\right)$.\\

\textit{\textcolor{blue}{The case $e_J=2$ and $m_J=2$}}: In any case, we know by the second bullet point of Theorem \ref{Theorem: High Order Estimates}, that $\left( h_i\ \underline{\otimes}\  h_j\right)^\dagger\cdot   \hat{v}\cdot  h_k\ \underline{\otimes}\  h_\ell \left(-\mathbb{L}'\right)=o_*\left(\mathbb{T}_N\right)$. In case $\{i,j,k,\ell\}=\{0,4\}$, this is the whole residuum $R_J$. In case $\{i,j,k,\ell\}=\{0,3\}$, the residuum reads
\begin{align*}
R_J=\left( h_i\ \underline{\otimes}\  h_j\right)^\dagger\cdot   \hat{v}\cdot  h_k\ \underline{\otimes}\  h_\ell \left(-\mathbb{L}'\right)+\sum_{r,r'=1}^d\Lambda_J^{r,r'}\left[(p'_r-p_r)\cdot p'_{r'}+p_r\cdot (p'_{r'}-p_{r'})\right].
\end{align*}
Since any of the products $(p'_r-p_r)\cdot p'_{r'}$ and $p_r\cdot (p'_{r'}-p_{r'})$ are of order $o_*\left(\mathbb{T}_N\right)$, we conclude $R_J=o_*\left(\mathbb{T}_N\right)$. The case $\{i,j,k,\ell\}=\{0,3,4\}$ works similarly, and is left to the reader.\\

\textit{\textcolor{blue}{The cases $e_J=2$ and $m_J<2$, respectively $e_J>2$}}: We obtain for $m_J=0$ by the first bullet point of Theorem \ref{Theorem: High Order Estimates}, and for $m_J=1$ by the third bullet point, that
\begin{align*}
R_J=H_J=o_*\left(\mathbb{T}_N\right).
\end{align*}
\end{proof}

\begin{cor}
\label{Corollary: Main B}
Recall the functional $\mathcal{E}_B$ defined in Eq.~(\ref{Equation: B Decomposition classic}) and the constant $c$ from Corollary \ref{Corollary: Main A}. Then,
\begin{align*}
\mathcal{W}_N \widetilde{B}_N \mathcal{W}_N^{-1}&=\mathcal{E}_B(q)+D_\mathcal{V}\big|_{q} \mathcal{E}_B\left(b_{\geq 1}\right)+\frac{1}{2}D^2_\mathcal{V}\big|_{0}\ \mathcal{E}_B\left(b_{\geq 1}\right)-\frac{c}{N}+o_*\left(\mathbb{T}_N\right).
\end{align*}
\end{cor}
\begin{proof}
Let us define $\tilde{c}:=\sum_{J\in \{0,\dots,4\}^4}\lambda_J c_J$. Combining Lemma \ref{Lemma: B Decomposition} and Theorem \ref{Theorem: Main B} immediately yields
\begin{align*}
\mathcal{W}_N \widetilde{B}_N \mathcal{W}_N^{-1}&=\mathcal{E}_B(q)+D_\mathcal{V}\big|_{q} \mathcal{E}_B\left(b_{\geq 1}\right)+\frac{1}{2}D^2_\mathcal{V}\big|_{0}\ \mathcal{E}_B\left(b_{\geq 1}\right)+\frac{\tilde{c}}{N}+o_*\left(\mathbb{T}_N\right).
\end{align*}
Recall the definition of $c_J$ in Eq.~(\ref{Equation: c_J for A}) for $J\in \{0,\dots,4\}^2$, respectively Eq.~(\ref{Equation: c_J for B}) for $J\in \{0,\dots,4\}^4$. Making use of the observation that most of the $c_J$ are zero, we obtain
\begin{align*}
c+\tilde{c}&=\sum_{J\in \{0,\dots,4\}^2}\lambda_J c_J+\sum_{J\in \{0,\dots,4\}^4}\lambda_J c_J=\lambda_{(0,0)}c_{(0,0)}+\lambda_{(1,1)}\Big[c_{(1,3)}+c_{(3,1)}+c_{(3,3)}\Big]\\
&\ +\lambda_{(0,0,0,0)}c_{(0,0,0,0)}+\lambda_{(1,1,0,0)}\Big[c_{(3,3,0,0)}+c_{(3,1,0,0)}+c_{(1,3,0,0)}\Big]\\
&\  +\lambda_{(1,0,1,0)}\Big[c_{(3,0,3,0)}+c_{(3,0,1,0)}+c_{(1,0,3,0)}\Big]+\lambda_{(0,1,1,0)}\Big[c_{(0,3,3,0)}+c_{(0,1,3,0)}+c_{(0,3,1,0)}\Big]\\
&=-\frac{1}{8}\partial^2_{t_j}\big|_{t=0}\sum_{j=1}^d\Big(\lambda_{(0,0)} \mathcal{E}_{(0,0)}\left(\vec{t}\ \right)+\lambda_{(1,1)} \mathcal{E}_{(1,1)}\left(\vec{t}\ \right)+\lambda_{(0,0,0,0)} \mathcal{E}_{(0,0,0,0)}\left(\vec{t}\ \right)\\
&\ +\lambda_{(1,1,0,0)} \mathcal{E}_{(1,1,0,0)}\left(\vec{t}\ \right)+\lambda_{(1,0,1,0)} \mathcal{E}_{(1,0,1,0)}\left(\vec{t}\ \right)+\lambda_{(0,1,1,0)} \mathcal{E}_{(0,1,1,0)}\left(\vec{t}\ \right)\Big)\\
&=-\frac{1}{8}\sum_{j=1}^d \partial^2_{t_j}\big|_{t=0}\Big(\mathcal{E}_A\left(\vec{t}\, \right)+\mathcal{E}_B\left(\vec{t}\, \right)\Big)=0,
\end{align*}
where we have used in the first equality of the last line that $\partial^2_{t_j}\big|_{t=0}\lambda_J\mathcal{E}_J\left(\vec{t}\, \right)=0$ for 
\begin{align*}
J\notin \{(0,0),(1,1),(1,1,0,0),(1,0,1,0),(0,1,1,0)\}
\end{align*}
and in the second equality of the last line that $\mathcal{E}_A\left(\vec{t}\, \right)+\mathcal{E}_B\left(\vec{t}\, \right)=\mathcal{E}'\left(\vec{t}\, \right)=e_\mathrm{H}$ for $t$ small enough, where $\mathcal{E}'$ is defined in Eq.~(\ref{Equation: Parameterized Hartree energy}).
\end{proof}

\begin{proof}[\textbf{Proof of Theorem \ref{Theorem: Decomposition}}.]
Making use of Eq.~(\ref{Equation: Corollary: Estimates Lower Bound - first line}), we obtain
\begin{align}
\label{Equation: Split of the operator}
\left(\mathcal{W}_N U_N\right) N^{-1}H_N \left(\mathcal{W}_N U_N\right)^{-1}=\mathcal{W}_N \widetilde{A}_N\mathcal{W}_N^{-1}+\mathcal{W}_N \widetilde{B}_N\mathcal{W}_N^{-1} +o_*\left(\mathbb{T}_N\right),
\end{align}
where we have used that $\mathcal{W}_N\, b_{\geq 1}^\dagger\cdot T\cdot b_{\geq 1}\, \mathcal{W}_N^{-1}\leq 2\left(X_1+X_2\right)$ with
\begin{align*}
X_1:&=\left(q+f(q)\right)^\dagger\cdot T\cdot \left(q+f(q)\right)=o_*(1),\\
X_2:&=\left(ip'+b_{>d}\right)^\dagger\cdot T\cdot \left(ip'+b_{>d}\right)=O_*\left(\mathbb{T}_N\right),
\end{align*}
see Lemmata \ref{Lemma: Function of q} and \ref{Lemma: p'}. Combining Corollaries \ref{Corollary: Main A} and \ref{Corollary: Main B} yields
\begin{align*}
\mathcal{W}_N \widetilde{A}_N\mathcal{W}_N^{-1}+\mathcal{W}_N \widetilde{B}_N\mathcal{W}_N^{-1}&=\mathcal{E}\left(q\right)\!+D_{\mathcal{V}}\big|_{q} \mathcal{E}\Big(b_{\geq 1}\Big)+\!\frac{1}{2}D_{\mathcal{V}}^2\big|_{0} \mathcal{E}\Big(b_{\geq 1}\Big)\!+o_*\left(\mathbb{T}_N\right),
\end{align*}
with $\mathcal{E}:=\mathcal{E}_A+\mathcal{E}_B$. Furthermore, note that $\mathcal{E}(z)=\mathcal{E}'(z)$ for $\|z\|<r$ where $\mathcal{E}'$ is defined in Eq.~(\ref{Equation: Parameterized Hartree energy}), see Lemmata \ref{Lemma: A Decomposition} and \ref{Lemma: B Decomposition}. Therefore, $\mathcal{E}\left(\vec{t}\, \right)=e_\mathrm{H}$ and $D_{\mathcal{V}}\big|_{t} \mathcal{E}=0$ for $t$ small enough. As we will show in Lemma \ref{Lemma: q function}, this implies $\mathcal{E}\left(q\right)=e_\mathrm{H}+o_*\left(\mathbb{T}_N\right)$ and $D_{\mathcal{V}}\big|_{q} \mathcal{E}\Big(b_{\geq 1}\Big)=o_*\left(\mathbb{T}_N\right)$. Furthermore, we have $\frac{1}{2}D_{\mathcal{V}}^2\big|_{0}\ \mathcal{E}=\mathrm{Hess}|_{u_0}\mathcal{E}_\mathrm{H}$ and therefore 
\begin{align*}
\frac{1}{2}D_{\mathcal{V}}^2\big|_{0}\ \mathcal{E}\Big(b_{\geq 1}\Big)=N^{-1}\mathbb{H}.
\end{align*}
In combination with Eq.~(\ref{Equation: Split of the operator}) this concludes the proof.
\end{proof}

\appendix


\section{Coercivity of the Hessian in Example (II)}
\label{Appendix: Non-degenerate Hessian}

In the following we are going to verify that the Hartree energy of a system of  pseudo-relativistic bosons in $\mathbb{R}^3$ interacting  via a Newtonian potential, given by 
\begin{align*}
\mathcal{E}_g\left[u\right]:=\big\langle\sqrt{m^2-\Delta}-m\big\rangle_u-(u\otimes u)^\dagger\cdot \frac{g}{2|x-y|}\cdot u\otimes u ,
\end{align*}
satisfies the coercivity assumption in Eq.~(\ref{Equation: Non-degenerate Hessian}) for $g$ small enough, see Example (II) in the introduction. Note that we are using the notation introduced in Section \ref{Secction: Fock Space Formalism}. Let us denote with $u_{g,\beta}$ the unique radial minimizer of the functional $\mathcal{E}_g$ subject to the rescaled condition $\|u\|=1+\beta$, i.e. $u_{g,\beta}$ is radial, and satisfies $\|u_{g,\beta}\|=1+\beta$ and $\mathcal{E}_g[u_{g,\beta}]=\underset{\|u\|=1+\beta}{\inf}\mathcal{E}_g[u]$. Let us further denote the normed minimizers by $u_g:=u_{g,0}$. By a scaling argument it is easy to see that $u_{g,\beta}=(1+\beta)u_{g(1+\beta)^2}$. For real-valued functions $f$ and $h$ in $\{u_{g,\beta}\}^\perp$ we can express the Hessian as $\frac{1}{2}\mathrm{Hess}|_{u_{g,\beta}}\mathcal{E}_g[f+ih]=\braket{L^+_{g,\beta}}_f+\braket{L^-_{g,\beta}}_h$, where $L^+_{g,\beta}$ and $L^-_{g,\beta}$ are selfadjoint operators given by
\begin{align*}
L^-_{g,\beta}&:=\sqrt{m^2-\Delta}-m-\mu_{g,\beta}- \left(1\otimes u_{g,\beta}\right)^\dagger\cdot \frac{g}{|x-y|}\cdot 1\otimes u_{g,\beta},\\
L^+_{g,\beta}&:=L^-_{g,\beta}-\left(1\otimes u_{g,\beta}\right)^\dagger \cdot \frac{2g}{|x-y|}\cdot u_{g,\beta}\otimes 1,
\end{align*}
with $\mu_{g,\beta}:=\big\langle\sqrt{m^2-\Delta}-m\big\rangle_{u_{g,\beta}}-(u_{g,\beta}\otimes u_{g,\beta})^\dagger\cdot \frac{g}{|x-y|}\cdot u_{g,\beta}\otimes u_{g,\beta}$. Furthermore we denote the operators associated to the normed minimizers $u_g$ by $L^\pm_{g}:=L^\pm_{g,0}$. Note that 
\begin{align*}
\braket{L^-_{g} - L^+_{g}}_f=\left(f \otimes u_{g}\right)^\dagger\cdot \frac{2g}{|x-y|}\cdot u_{g}\otimes f>0
\end{align*}
for all $f\neq 0$, and consequently it is enough to verify the following Theorem \ref{Theorem: Concavity for pseudo-relativistic systems} in order to prove Eq.~(\ref{Equation: Non-degenerate Hessian}).

\begin{theorem}
\label{Theorem: Concavity for pseudo-relativistic systems}
There exist constants $g_0$ and $\eta>0$ such that for all $0<g<g_0$ and $f\in L^2\left(\mathbb{R}^d\right)$ with $f\perp \{u_g,\partial_{x_1} u_g,\partial_{x_2} u_g,\partial_{x_3} u_g\}$
\begin{align*}
\braket{L^+_g}_f\geq \eta\|f\|^2.
\end{align*}
\end{theorem}
In order to prove Theorem \ref{Theorem: Concavity for pseudo-relativistic systems}, we first need some auxiliary results regarding the minimizers $u_{g,\beta}$ subject to the rescaled condition $\|u_{g,\beta}\|=1+\beta$. 

\begin{lem}
\label{Lemma: R_beta}
Let us define $R_{g,\beta}:=u_{g,\beta}-u_g$ for $\beta\in [0, 1)$ (where $1$ can be replaced by any other positive number). Then there exist constants $g_0,C>0$ such that
\begin{align*}
L^+_g R_{g,\beta}=\delta_{g,\beta} u_{g}+\epsilon_{g,\beta},
\end{align*}
with $|\delta_{g,\beta}|\leq C\beta$ and $\|\epsilon_{g,\beta}\|\leq C\|R_{g,\beta}\|^2$ for $g\in (0,g_0)$ and $\beta\in [0,1)$.
\end{lem}
\begin{proof}
Since the elements $u_{g,\beta}$ are minimizers of $\mathcal{E}_g$, they satisfy the corresponding Euler-Lagrange equations $L^-_{g,\beta} u_{g,\beta}=0$. A straightforward computation yields
\begin{align*}
0=&L^-_{g,\beta} u_{g,\beta}-L^-_{g} u_{g}=L^+_{g}R_{g,\beta}-\delta_{g,\beta}u_g-\epsilon_{g,\beta}
\end{align*}
with 
\begin{align}
\nonumber \delta_{g,\beta}&:=\mu_{g,\beta}-\mu_{g},\\
\nonumber \epsilon_{g,\beta}&:=\left(\mu_{g,\beta}-\mu_{g}\right)R_{g,\beta}+ \left(1\otimes u_{g,\beta}\right)^\dagger\cdot \frac{g}{|x-y|}\cdot R_{g,\beta}\otimes R_{g,\beta}\\
&\label{Equation: Definition of epsilon_g,beta} \ \ \ \ +\left(1\otimes R_{g,\beta}\right)^\dagger\cdot \frac{g}{|x-y|}\cdot R_{g,\beta}\otimes u_{g}+\left(1\otimes R_{g,\beta}\right)^\dagger\cdot \frac{g}{|x-y|}\cdot u_{g}\otimes R_{g,\beta}.
\end{align}
Let us first investigate the contributions involving $\frac{g}{|x-y|}$. From \cite[Proposition 1]{Le} it is clear that there exists a constant $C$ such that $\|u_{g,\beta}\|_{H^1\left(\mathbb{R}^3\right)}\leq C<\infty$ for all $g$ small enough and $\beta\in [0,1)$. With the notation $S:=\sqrt{1-\Delta}$ we obtain
\begin{align*}
&\left\|\left(1\otimes u_{g,\beta}\right)^\dagger\cdot \frac{g}{|x-y|}\cdot R_{g,\beta}\otimes R_{g,\beta}\right\|=\left\|\left(1\otimes S u_{g,\beta}\right)^\dagger\cdot 1\otimes S^{-1}\frac{g}{|x-y|}\cdot R_{g,\beta}\otimes R_{g,\beta}\right\|\\
&\ \ \leq g\|S u_{g,\beta}\|\, \left\|1\otimes S^{-1}\frac{1}{|x-y|}\right\|\, \|R_{g,\beta}\|^2\leq Cg \left\|S^{-1}\frac{1}{|x|}\right\|\, \|R_{g,\beta}\|^2,
\end{align*}
where $\left\|S^{-1}\frac{1}{|x|}\right\|$ is the operator norm of the bounded one-particle operator $S^{-1}\frac{1}{|x|}$. Similarly, the other contributions involving $\frac{g}{|x-y|}$ in Eq.~(\ref{Equation: Definition of epsilon_g,beta}) can be estimated by $Cg \left\|S^{-1}\frac{1}{|x|}\right\|\, \|R_{g,\beta}\|^2$ as well. The uniform control of the norm $\|u_{g,\beta}\|_{H^1\left(\mathbb{R}^3\right)}\leq C<\infty$ furthermore implies $|\delta_{g,\beta}|=|\mu_{g,\beta}-\mu_g|\leq \tilde{C}\beta$ for some constant $\tilde{C}$. Note that $\|R_{g,\beta}\|\geq \|u_{g,\beta}\|-\|u_g\|=\beta$, and consequently $\|\left(\mu_{g,\beta}-\mu_{g}\right)R_{g,\beta}\|\leq \tilde{C}\beta\|R_{g,\beta}\|\leq \tilde{C}\|R_{g,\beta}\|^2$. We conclude that
\begin{align*}
\|\epsilon_{g,\beta}\|\leq \left(\tilde{C}+3Cg \left\|S^{-1}\frac{1}{|x|}\right\|\right)\|R_{g,\beta}\|^2.
\end{align*}
\end{proof}

\begin{lem}
\label{Lemma: Non-degenerate growth of R_g,beta}
Let $R_{g,\beta}$ and $\epsilon_{g,\beta}$ be as in Lemma \ref{Lemma: R_beta}. Then there exists a constant $g_0>0$ such that $\underset{\beta\rightarrow 0}{\lim}\, \frac{\|\epsilon_{g,\beta}\|}{\braket{u_g,R_{g,\beta}}}=0$ and $\underset{\beta\rightarrow 0}{\limsup}\, \frac{|\delta_{g,\beta}|}{\braket{u_g,R_{g,\beta}}}\leq C$ for a suitable constant $C>0$ and $g\in (0,g_0)$.
\end{lem}
\begin{proof}
By the results in \cite{Le} we know that $0$ is an isolated eigenvalue of $L^+_{g}$, i.e. there exists a constant $\delta>0$ such that $\sigma\left(L^+_{g}\right)\cap (-\delta,\delta)=\{0\}$, with corresponding eigenvectors $\partial_{x_1} u_g,\partial_{x_2} u_g,\partial_{x_3} u_g$. Since $u_{g,\beta}$ is radial, we know that $R_{g,\beta}\perp \partial_{x_j}u_g$, and therefore we obtain by Lemma \ref{Lemma: R_beta}
\begin{align}
\label{Equation: Differentiability}
\delta \|R_{g,\beta}\|\leq \|L^+_g R_{g,\beta}\|\leq C\left(\beta+\|R_{g,\beta}\|^2\right).
\end{align}
Using \cite[Proposition 1]{Le} again, it is clear that $\underset{g\rightarrow 0, \beta\rightarrow 0}{\lim}\,\|R_{g,\beta}\|=0$ and therefore there exists a constant $g_0$ such that $\|R_{g,\beta}\|\leq \frac{\delta}{2C}$ for all $g\in (0,g_0)$ and $\beta$ small enough. Consequently Eq.~(\ref{Equation: Differentiability}) yields $\|R_{g,\beta}\|\leq \frac{2C}{\delta}\beta$. Using the fact that $\|u_{g,\beta}\|=1+\beta$, we further obtain
\begin{align*}
1+2\beta\leq (1+\beta)^2=1+2\braket{u_g,R_{g,\beta}}+\|R_{g,\beta}\|^2\leq 1+ 2\braket{u_g,R_{g,\beta}}+\left(\frac{2C}{\delta}\right)^2 \beta^2,
\end{align*}
and therefore
\begin{align*}
\frac{\|\epsilon_{g,\beta}\|}{\braket{u_g,R_{g,\beta}}}\leq \frac{C\left(\frac{2C}{\delta}\right)^2\beta^2}{\beta-2\left(\frac{C}{\delta}\right)^2 \beta^2}\underset{\beta\rightarrow 0}{\longrightarrow}0,\\
\frac{|\delta_{g,\beta}|}{\braket{u_g,R_{g,\beta}}}\leq \frac{C\beta}{\beta-2\left(\frac{C}{\delta}\right)^2 \beta^2}\underset{\beta\rightarrow 0}{\longrightarrow}C.
\end{align*}
\end{proof}

\begin{proof}[Proof of Theorem \ref{Theorem: Concavity for pseudo-relativistic systems}]
Let $Q$ denote the projection onto the space $\{u_g\}^\perp$. Clearly there exists a $w\in L^2(\mathbb{R}^3)$ such that
\begin{align*}
L^+_g f=QL^+_g f+\braket{w,f}u_g
\end{align*}
for all $f\in L^2(\mathbb{R}^3)$. With $R_{g,\beta},\delta_{g,\beta}$ and $\epsilon_{g,\beta}$ from Lemma \ref{Lemma: R_beta} at hand, we obtain
\begin{align*}
L^+_g\left(\delta_{g,\beta} f-\braket{w,f}R_{g,\beta}\right)=\delta_{g,\beta} QL^+_g f-\braket{w,f}\epsilon_{g,\beta},
\end{align*}
and therefore $\|L^+_g\left(\delta_{g,\beta} f-\braket{w,f}R_{g,\beta}\right)\|\leq |\delta_{g,\beta}|\, \|QL^+_g f\|+\|w\|\, \|\epsilon_{g,\beta}\|\, \|f\|$. Using again that there exists a constant $\delta>0$ such that $\sigma\left(L^+_{g}\right)\cap (-\delta,\delta)=\{0\}$ with corresponding eigenvectors $\partial_{x_1} u_g,\partial_{x_2} u_g,\partial_{x_3} u_g$, see \cite{Le}, and that $R_{g,\beta}$ as a radial function is orthogonal to them, we obtain for all $f\in \{u_g,\partial_{x_1} u_g,\partial_{x_2} u_g,\partial_{x_3} u_g\}^\perp$
\begin{align*}
\|L^+_g\left(\delta_{g,\beta} f-\braket{w,f}R_{g,\beta}\right)\|\geq \delta\|\delta_{g,\beta} f-\braket{w,f}R_{g,\beta}\|\geq \delta |\braket{u_g,R_{g,\beta}}|\, |\braket{w,f}|.
\end{align*}
Combining the estimates we have so far yields
\begin{align*}
|\braket{w,f}|\leq \frac{|\delta_{g,\beta}|}{\delta |\braket{u_g,R_{g,\beta}}|} \|QL^+_g f\|+\frac{\|w\|\, \|\epsilon_{g,\beta}\|}{\delta |\braket{u_g,R_{g,\beta}}|} \|f\|=:x_\beta \|QL^+_g f\|+y_\beta \|f\|.
\end{align*}
By Lemma \ref{Lemma: Non-degenerate growth of R_g,beta} we know that $\underset{\beta \rightarrow 0}{\limsup}|x_\beta|\leq C$ for some constant $C>0$ and $y_\beta\underset{\beta \rightarrow 0}{\longrightarrow}0$. Using again that $\sigma\left(L^+_{g}\right)\cap (-\delta,\delta)=\{0\}$, we obtain
\begin{align*}
\delta \|f\|\leq \|L^+_g f\|\leq \|QL^+_g f\|+|\braket{w,f}|\leq \left(1+x_\beta\right)\|QL^+_g f\|+y_\beta \|f\|
\end{align*}
and consequently $\|QL^+_g f\|\geq \frac{\delta-y_\beta}{1+x_\beta}\|f\|$. This holds for all (small) $\beta$, hence $\beta\to 0$ gives $ \|QL^+_g f\|\geq  \frac{\delta}{1+C}\|f\|$. Finally note that $QL^+_gQ\geq 0$ since $\braket{L^+_g}_f=\frac{1}{2}\mathrm{Hess}|_{u_{g}}\mathcal{E}_g[f]\geq 0$ for real-valued $f\perp u_g$, which concludes the proof.
\end{proof}

\section{The Bogoliubov Operator}
\label{Appendix A}
In the following we will prove Theorem \ref{Theorem: Bogoliubov Ground State Energy}, i.e. we are going to verify that the Bogoliubov operator $\mathbb{H}$ constructed in Definition \ref{Definition: Bogoliubov Operator} is bounded from below and that its ground state energy can be approximated by $\Psi\in \bigcup_{M\in \mathbb{N}}\mathrm{dom}[a_{\geq 1}^\dagger\cdot (T+1)\cdot a_{\geq 1}]\cap \mathcal{F}_{\leq M}$ with $\|\Psi\|=1$. Our strategy is to decouple the degenerate modes from the non-degenerate ones and to apply the general framework for non-degenerate Bogoliubov operators in \cite{NNS}. 

\begin{defi}
\label{Definition: Perpendicular Bogoliubov}
Let $Q_\mathrm{H}=\sum_{i,j\geq 1}Q_{i,j}\ u_i\cdot u_j^\dagger$ and $G_\mathrm{H}=\sum_{i,j\geq 1}G_{ij}\ u_i\otimes u_j$ be as in Lemma \ref{Lemma: Hessian}, and let us denote the operator $Q_\perp:=\sum_{i,j>d}Q_{i,j}\ u_i\cdot u_j^\dagger$ on $\mathcal{H}_\perp:=\braket{u_0,u_1,\dots,u_d}^\perp$ as well as $G_\perp:=\sum_{i,j>d}G_{ij} u_i\otimes u_j$. Then we define the operator $\mathbb{H}_\perp$ as
\begin{align*}
\mathbb{H}_\perp:= a_{>d}^\dagger\cdot   Q_\perp\cdot  a_{>d}+2\mathfrak{Re}\left[  G_\perp^\dagger\cdot  a_{>d}\ \underline{\otimes}\ a_{>d}\right].
\end{align*}
\end{defi}

\begin{lem}
The operator $\mathbb{H}_\perp$ is semi-bounded from below, i.e. $\inf \sigma\left(\mathbb{H}_\perp\right)>-\infty$. Furthermore, there exists a constant $R>0$ such that
\begin{align}
\label{Equality: Bound from above}
\mathbb{H}_\perp\leq R\left( a_{>d}^\dagger\cdot   Q_\perp\cdot  a_{>d}+1\right).
\end{align}
\end{lem}
\begin{proof}
Let us define the operator $G_{\mathrm{op}}$ on $\mathcal{H}_\perp$ by the condition $z^\dagger\cdot G_{\mathrm{op}}\cdot z=2 G_\perp^\dagger\cdot \overline{z}\otimes z$, with $\overline{z}$ being the usual complex conjugation in $L^2\left(\mathbb{R}^d\right)$. Then, $z^\dagger\cdot Q_\perp\cdot z+2\mathfrak{Re}\left[\overline{z}^\dagger\cdot G_{\mathrm{op}}\cdot z\right]=\mathrm{Hess}|_{u_0}\mathcal{E}_\mathrm{H}[z]\geq \eta \|z\|^2$ for all $z\in \mathcal{H}_\perp$ with $\eta>0$ by Assumption \ref{Assumption: Part II}. As pointed out in Section 2.1 in \cite{LNSS}, this implies $Q_\perp\geq r>0$ as well as
\begin{align*}
\left(\begin{array}{rr}
Q_\perp & G_{\mathrm{op}}^\dagger  \\ G_{\mathrm{op}}& Q_\perp
\end{array}\right)\geq 0,
\end{align*}
where we have used that $Q_\perp$ is a real operator. Since $Q_\perp>0$, this is further equivalent to $G_{\mathrm{op}}\, Q_\perp^{-1}\, G_{\mathrm{op}}^\dagger\leq Q_\perp$. Since $G_\mathrm{H}\in \overline{\mathcal{H}_0\otimes \mathcal{H}_0}^{\|.\|_*}$, where the $\|.\|_*$ norm is defined in Lemma \ref{Lemma: Hessian}, and since $z^\dagger\cdot  Q_\perp^{-\frac{1}{2}}\cdot z\leq c\, z^\dagger\cdot \left(T+1\right)^{-\frac{1}{2}}\cdot z$ for a suitable constant $c$ and $z\in \mathcal{H}_\perp$, which is an easy consequence of the operator inequality in Lemma \ref{Lemma: Hessian} and the fact that $Q_\perp\geq r>0$, we obtain
\begin{align}
\label{Equation: constant c}
\mathrm{Tr}\left[G_{\mathrm{op}}\, Q_\perp^{-1}\, G_{\mathrm{op}}^\dagger\right]=\|1_{\mathcal{H}_\perp}\otimes Q_\perp^{-\frac{1}{2}}\cdot G_\perp\|^2_{\mathcal{H}\otimes \mathcal{H}}\leq c^2\, \|G_\mathrm{H}\|^2_*<\infty,
\end{align}
i.e. $G_{\mathrm{op}}\, Q_\perp^{-\frac{1}{2}}$ is a Hilbert-Schmidt operator. By the general results in \cite{NNS}, this implies that $\mathbb{H}_\perp$ is semi-bounded as well as the existence of a constant $R>0$ such that Eq.~(\ref{Equality: Bound from above}) holds.
\end{proof}

\begin{lem}
\label{Lemma: constant c_0}
Let us define $P_j:=\frac{1}{2i}\left(a_j-a_j^\dagger\right)=\sqrt{N}p_j$ for $j\in \{1,\dots,d\}$, the constant $c_0:=\sum_{j=1}^d G_{j,j}$, the quadratic function $\nu(y):=-4\sum_{j,k=1}^d G_{j,k}y_j y_k$ for $y\in \mathbb{R}^d$ and the linear $\mathcal{H}_\perp$ valued function
\begin{align*}
u(y_1,\dots,y_d):=4i\sum_{j=1}^d y_j\sum_{k>d}G_{j,k}u_k\in \mathcal{H}_\perp.
\end{align*}
Then we can rewrite the Bogoliubov operator $\mathbb{H}$ from Definition (\ref{Definition: Bogoliubov Operator}) as
\begin{align}
\label{Equation: P-linear terms}
\mathbb{H}=c_0+\nu\left(P_1,\dots,P_d\right)+u\left(P_1,\dots,P_d\right)^\dagger\cdot  a_{>d}+ a_{>d}^\dagger\cdot u\left(P_1,\dots,P_d\right)+\mathbb{H}_\perp.
\end{align}
\end{lem}
\begin{proof}
Since $\mathrm{Hess}|_{u_0}\mathcal{E}_\mathrm{H}[z]=z^\dagger\cdot Q_\mathrm{H}\cdot z+2\mathfrak{Re}\left[G_\mathrm{H}^\dagger\cdot z\otimes z\right]$ is degenerate in the directions $u_1,\dots,u_d$, we obtain $Q_{j,k}=-2G_{j,k}$ in case $j$ or $k$ is in $\{1,\dots,d\}$. Writing $\mathbb{H}$ in coordinates therefore yields
\begin{align*}
\mathbb{H}\! &=\!\! \!  \sum_{j,k=1}^d\!  G_{j,k}\left(\!-\!2 a_j^\dagger a_k\!+\!a_j a_k \!+\!a_j^\dagger a_k^\dagger\right)\!+\!2\sum_{j=1}^d \sum_{k>d} \! G_{j,k}\left(a_j\!-\!a_j^\dagger\right)a_k\!-\!2\sum_{j=1}^d \sum_{k>d} G_{j,k}\left(a_j\!-\!a_j^\dagger\right)a_k^\dagger\\
&\ \ \ \ +\sum_{j,k>d}\! \left(Q_{j,k}a_j^\dagger a_k+G_{j,k}a_j a_k+G_{j,k}a_j^\dagger a_k^\dagger\right)\\
&=\big\{c_0+\nu\left(P_1,\dots,P_d\right)\big \}+u\left(P_1,\dots,P_d\right)^\dagger\cdot  a_{>d}+ a_{>d}^\dagger\cdot u\left(P_1,\dots,P_d\right)+\mathbb{H}_\perp,
\end{align*}
where we have used $-2 a_j^\dagger a_k+a_j a_k +a_j^\dagger a_k^\dagger=-4P_j P_k+\delta_{j,k}$ in the second identity.
\end{proof}

\begin{note}
\label{Remark: Completement of the square}
In the subsequent Lemma \ref{Lemma: Transformed Bogoliubov}, we want to get rid of the term $u\left(P_1,\dots,P_d\right)^\dagger\cdot  a_{>d}+ a_{>d}^\dagger\cdot u\left(P_1,\dots,P_d\right)$ in Eq.~(\ref{Equation: P-linear terms}) by completing the square, i.e. by applying a shift $a_{>d}\mapsto a_{>d}+w\left(P_1,\dots,P_d\right)$ where $w(y_1,\dots,y_d)\in \mathcal{H}_\perp$ is a suitable vector. In the following we are going to construct such a $w(y)$. Let us first define the $\mathbb{R}$-linear map $L:\mathcal{H}_\perp\longrightarrow \mathcal{H}_\perp$
\begin{align*}
z^\dagger\cdot L(w):=z^\dagger\cdot Q_\perp\cdot w+2\left(w \otimes z\right)^\dagger \cdot G_\perp,
\end{align*}
for all $z\in \mathcal{H}_\perp$. Furthermore, let us define the real inner product $\braket{z,w}_\mathbb{R}:=\mathfrak{Re}\left[z^\dagger\cdot w\right]$ on $\mathcal{H}_\perp$. Clearly, $L$ is symmetric with respect to this inner product. By Assumption \ref{Assumption: Part II} we have for all $w\in \mathcal{H}_\perp$
\begin{align}
\braket{w,L(w)}_\mathbb{R}=\mathrm{Hess}|_{u_0}\mathcal{E}_\mathrm{H}[w]\geq \eta\, \|w\|^2,
\end{align}
and consequently we can define $w(y)\in \mathcal{H}_\perp$ for all $y\in \mathbb{R}^d$ as the solution of the equation 
\begin{align}
\label{Equation: Definition w vector}
L\cdot w(y)=-u(y).
\end{align}
We note that $w(y)\in \mathrm{dom}[Q_\perp]$ due to the improved coercivity
\begin{align}
\label{Equation: Improved Coercivity}
\braket{w,L(w)}_\mathbb{R}\geq \tilde{c}\, w^\dagger\cdot Q_\perp\cdot w
\end{align}
where $\tilde{c}$ is a suitable constant, which follows from the fact that
\begin{align*}
2\left|\left(w \otimes w\right)^\dagger \!\cdot\! G_\perp\right|\leq w^\dagger\!\cdot\! \left(\epsilon\, Q_\perp\!+\!\epsilon^{-1}c^2\|G_\mathrm{H}\|_*^2\right)\!\cdot\! w\leq \epsilon\, w^\dagger\!\cdot\! Q_\perp\!\cdot \!w\!+\!\frac{c^2\|G_\mathrm{H}\|_*^2}{\epsilon\eta}\braket{w,L(w)}_\mathbb{R}
\end{align*}
for all $\epsilon>0$, where $c$ is the constant in Eq.~(\ref{Equation: constant c}).
\end{note}

\begin{lem}
\label{Lemma: Transformed Bogoliubov}
Let $w:\mathbb{R}^d\rightarrow \mathcal{H}_\perp$ be the function defined by Eq.~(\ref{Equation: Definition w vector}) and let us define the unitary transformation $\mathcal{R}:\mathcal{F}_0\longrightarrow \mathcal{F}_0$
\begin{align*}
\mathcal{R}:&=\exp\left[w\left(P_1,\dots,P_d\right)^\dagger\cdot  a_{>d}- a_{>d}^\dagger\cdot w\left(P_1,\dots,P_d\right)\right].
\end{align*}
Then there exists a non-negative quadratic function $\eta:\mathbb{R}^d\longrightarrow \mathbb{R}$, s.t.
\begin{align}
\label{Equation: R transformation}
\mathcal{R}\, \mathbb{H}\, \mathcal{R}^{-1}=c_0+\eta\left(P_1,\dots,P_d\right)+\mathbb{H}_\perp,
\end{align}
where $c_0$ and $\mathbb{H}_\perp$ are as in Lemma \ref{Lemma: constant c_0} and Definition \ref{Definition: Perpendicular Bogoliubov}.
\end{lem}
\begin{proof}
Let us define $\eta(y_1,\dots,y_d):=\nu(y_1,\dots,y_d)+\braket{w(y_1,\dots,y_d),u(y_1,\dots,y_d)}_\mathbb{R}$. With $\eta$ and the vector valued function $w$ at hand, we can rewrite Eq.~(\ref{Equation: P-linear terms}) as
\begin{align}
\mathbb{H}&=c_0+\eta\left(P_1,\dots,P_d\right)+\Big( a_{>d}-w\left(P_1,\dots,P_d\right)\Big)^\dagger\cdot   Q_\perp\cdot \Big( a_{>d}-w\left(P_1,\dots,P_d\right)\Big) \nonumber \\
&\ \ \ +2\mathfrak{Re}\left[  G_\perp^\dagger\cdot \Big( a_{>d}-w\left(P_1,\dots,P_d\right)\Big)\ \underline{\otimes}\ \Big( a_{>d}-w\left(P_1,\dots,P_d\right)\Big)\right].\label{Equation: Completed Square}
\end{align}
Eq.~(\ref{Equation: R transformation}) follows now from the representation of $\mathbb{H}$ in Eq.~(\ref{Equation: Completed Square}) and the fact that
\begin{align*}
\mathcal{R}\,  a_{>d}\, \mathcal{R}^{-1}= a_{>d}+w\left(P_1,\dots,P_d\right).
\end{align*}
In order to see that $\eta$ is indeed non-negative, note that we can use $\eta$ and $w$ to complete the square in $\mathrm{Hess}|_{u_0}\mathcal{E}_\mathrm{H}[z]$ as well, i.e. for $z=\sum_{j=1}^d \left(t_j+is_j\right)u_j+z_{>d}$ with $t,s\in \mathbb{R}^d$ and $z_{>d}\in \mathcal{H}_\perp$ we can write $\mathrm{Hess}|_{u_0}\mathcal{E}_\mathrm{H}[z]$ as
\begin{align*}
\eta(s)\!+\!\left(z_{>d}\!-\!w(s)\right)^\dagger\!\cdot\! Q_\perp\!\cdot\! \left(z_{>d}\!-\!w(s)\right)\!+\!2\mathfrak{Re}\left[G_\perp^\dagger\cdot \left(z_{>d}\!-\!w(s)\right)\!\otimes\! \left(z_{>d}\!-\!w(s)\right)\right].
\end{align*}
Therefore, $\mathrm{Hess}|_{u_0}\mathcal{E}_\mathrm{H}[z]\geq 0$ for all $z$ implies $\eta(s)\geq 0$ for all $s\in \mathbb{R}^d$.
\end{proof}

\begin{proof}[Proof of Theorem \ref{Theorem: Bogoliubov Ground State Energy}]
Since the function $\eta$ in Lemma \ref{Lemma: Transformed Bogoliubov} is non-negative, we immediately obtain the lower bound
\begin{align*}
\inf \sigma\left(\mathbb{H}\right)\geq c_0+\inf \sigma\left(\mathbb{H}_\perp\right)>-\infty.
\end{align*}
In order to verify the bound from below for the operator $\mathbb{H}-r\mathbb{A}$, where $\mathbb{A}$ is defined in Eq.~(\ref{Equation: NT}), we will make use of the improved coercivity
\begin{align}
\label{Equation: Improved Coercivity Hessian}
\mathrm{Hess}|_{u_0}\mathcal{E}_\mathrm{H}[z]\geq r_*\left(\sum_{j=1}^d s_j^2+z_{>d}^\dagger\cdot (T+1)\cdot z_{>d}\right),
\end{align}
where $r_*$ is a suitable constant and $z=\sum_{j=1}^d (t_j+is_j)u_j+z_{>d}$ with $z_{>d}\in \mathcal{H}_\perp$, which can be verified analogously to Eq.~(\ref{Equation: Improved Coercivity}) in Remark \ref{Remark: Completement of the square}. With the definition $\eta_r:=r_*-r$ for $r<r_*$ we obtain, in analogy to Assumption \ref{Assumption: Part II},
\begin{align*}
\mathrm{Hess}|_{u_0}\mathcal{E}_\mathrm{H}[z]-r\left(\sum_{j=1}^d s_j^2+z_{>d}^\dagger\cdot (T+1)\cdot z_{>d}\right)\geq \eta_r \|z\|^2
\end{align*}
for all $z$ of the form $z=i\sum_{j=1}^d s_j u_j+z_{>d}$ with $s_j\in \mathbb{R}$ and $z_{>d}\in \mathcal{H}_\perp$. Therefore we can repeat the proof of the lower bound for the operator $\mathbb{H}-r\mathbb{A}$, which yields
\begin{align}
\label{Equation: Semi-bound by self-adjoint operator}
\mathbb{H}-r\mathbb{A}\geq \inf \sigma\left(\mathbb{H}-r\mathbb{A}\right)>-\infty.
\end{align}
Note that this further implies that the Friedrichs extension of the quadratic form $\mathbb{H}$ is well-defined, i.e. $\mathbb{H}$ is semi-bounded and closeable, since $\mathbb{H}$ is comparable to the non-negative selfadjoint operator $\mathbb{A}$, i.e. there exist constants $\alpha_1,\alpha_2,\beta_1,\beta_2>0$ with
\begin{align*}
 \alpha_1 \mathbb{A}-\beta_1\leq \mathbb{H}\leq \alpha_2\mathbb{
 A}+\beta_2.
\end{align*}
 In order to verify that there exists an approximate sequence of ground states $\Psi_M$ with $\Psi_M\in \mathcal{F}_{\leq M}$ and $\Psi_M\in \mathrm{dom}\left[a_{\geq 1}^\dagger\cdot (T+1)\cdot a_{\geq 1}\right]$, it is enough to prove that such states are dense in $\mathrm{dom}\left[a_{\geq 1}^\dagger\cdot   \left(T+1\right)\cdot  a_{\geq 1}\right]$, the domain of the quadratic form which defines the Bogoliubov operator $\mathbb{H}$ by Friedrichs extension, with respect to the norm $\|\Psi\|^2_{\mathbb{H}}:=\braket{\mathbb{H}+C}_\Psi$ where $C>-\inf \sigma\left(\mathbb{H}\right)$. The lower bound follows from Eq.~(\ref{Equation: Semi-bound by self-adjoint operator}), while the upper bound follows from Eq.~(\ref{Equation: P-linear terms}) and Inequality (\ref{Equality: Bound from above}). Furthermore, we have
\begin{align*}
\|\Psi\|^2_\mathbb{H}\leq \alpha_2\braket{ \mathbb{A}}_\Psi+(\beta_2+C)\|\Psi\|^2\leq \|\Psi\|^2_\diamond,
\end{align*}
for all $\Psi\in \mathcal{F}_0$, where $\|\Psi\|_\diamond^2:=\alpha_2\braket{ a_{\geq 1}^\dagger\cdot   (T+1)\cdot  a_{\geq 1}}_\Psi+(\beta_2+C+\frac{d}{4})\|\Psi\|^2$. Clearly, $\bigcup_M \mathcal{F}_{\leq M}\cap \mathrm{dom}\left[ a_{\geq 1}^\dagger\cdot   (T+1)\cdot  a_{\geq 1}\right]$ is dense in the domain $\mathrm{dom}\left[ a_{\geq 1}^\dagger\cdot   \left(T+1\right)\cdot  a_{\geq 1}\right]$ with respect to the norm $\|.\|_\diamond$ and therefore it is also dense with respect to $\|.\|_{\mathbb{H}}$. 
\end{proof}

\section{Auxiliary Lemmata}
\label{Appendix: B}
In the following section we will derive various operator estimates involving powers of the operators $p,p',b_{>d}$ and functions of $q$, with an emphasis on asymptotic results of the form $A_N=o_*(B_N)$, where the $o_*(\cdot )$ notation is introduced in Definition \ref{Definition: Smallness}. It is a crucial observation that all of our basic variables $q_i,p_j$ and $b_k$ are of order $o_*(1)$, and therefore the product of a basic variable with an operator $A_N$ should be of order $o_*(A_N)$, which we will verify for specific examples $A_N$. Let us first discuss an important tool, which we will repeatedly use, given by the following Cauchy--Schwarz inequality for operators.

\begin{lem}
\label{Lemma: O_* results}
For any $\lambda\in \mathbb{C}$ with $|\lambda|=1$, $t>0$, linear operators $A:\mathcal{H}_1\longrightarrow \mathcal{H}_2$ and $B:\mathcal{H}_1\longrightarrow \mathcal{H}_2$, and selfadjoint operator $Q:\mathcal{H}_2\longrightarrow \mathcal{H}_2$, we have the operator inequality
\begin{align}
\label{Inequality: Schwarz}
\mathfrak{Re}\left[\lambda\ A^\dagger\cdot Q\cdot B\right]\leq t\ A^\dagger\cdot \left|Q\right|\cdot A+t^{-1}\ B^\dagger\cdot \left|Q\right|\cdot B.
\end{align}
Furthermore, let $A_N,B_N$ be sequences of linear operators $\mathcal{H}_1\longrightarrow \mathcal{H}_2$, $Q$ a selfadjoint operator on $\mathcal{H}_2$ and $C_N:\mathcal{H}_1\longrightarrow \mathcal{H}_1$ a sequence of non-negative operators, which satisfy $A_N^\dagger\cdot \left|Q\right|\cdot A_N=O_*\left(C_N\right)$ and $B_N^\dagger\cdot \left|Q\right|\cdot B_N=o_*\left(C_N\right)$. Then,
\begin{align*}
A_N^\dagger\cdot Q\cdot B_N=o_*\left(C_N\right),\\
B_N^\dagger\cdot Q\cdot A_N=o_*\left(C_N\right).
\end{align*}
\end{lem}
\begin{proof}
Let $Q=U |Q|$ be the polar decomposition of $Q$. Inequality (\ref{Inequality: Schwarz}) immediately follows from the inequality
\begin{align*}
0\leq \left(\sqrt{t}\ A-\sqrt{\frac{1}{t}}\lambda\ U B\right)^\dagger\cdot |Q|\cdot \left(\sqrt{t}\ A-\sqrt{\frac{1}{t}}\lambda\ U B\right).
\end{align*}

By our assumption $A_N^\dagger\cdot \left|Q\right|\cdot A_N=O_*\left(C_N\right)$ we know that there exist constants $c,\delta_0>0$, such that $\pi_{M,N}\, A_N^\dagger\cdot \left|Q\right|\cdot A_N\, \pi_{M,N}\leq c\braket{C_N}_\Psi$ for all $\frac{M}{N}\leq \delta_0$. Furthermore, by our assumption $B_N^\dagger\cdot \left|Q\right|\cdot B_N=o_*\left(C_N\right)$, there exists a function $\epsilon:\mathbb{R}^+\longrightarrow \mathbb{R}^+$ with $\underset{\delta\rightarrow 0}{\lim}\epsilon(\delta)$, such that $\pi_{M,N}\, B_N^\dagger\cdot \left|Q\right|\cdot B_N\, \pi_{M,N}\leq \epsilon\left(\frac{M}{N}\right)C_N$. Applying Inequality (\ref{Inequality: Schwarz}) with $t:=\sqrt{\epsilon\left(\frac{M}{N}\right)}$ yields for all $\lambda\in \mathbb{C}$ with $|\lambda|=1$ and $\frac{M}{N}\leq \delta_0$
\begin{align*}
\pi_{M,N}\, \mathfrak{Re}\left[\lambda\ A_N^\dagger\cdot Q\cdot B_N\right]\, \pi_{M,N}\leq \sqrt{\epsilon\left(\frac{M}{N}\right)}\ C_N.
\end{align*}
\end{proof}

Consider a function $g:\mathbb{R}^d\longrightarrow \mathbb{R}$. The following Lemma states that the operator $g(q)$ depends, up to an exponentially small error, only on the local data of $g$ in an arbitrary small neighborhood $[-\epsilon,\epsilon]^d$ of the origin, i.e. $g(q)=\widetilde{g}(q)+O_*\left(e^{-\delta N}\right)$ in case $g|_{[-\epsilon,\epsilon]^d}=\widetilde{g}|_{[-\epsilon,\epsilon]^d}$. This property plays a key role in the proof of the main technical Theorem \ref{Theorem: Decomposition}, since the involved functions are (somewhat arbitrary) extensions of locally constructed functions with specific properties, which the extensions no longer have, see for example the definition of $f:\mathbb{R}^d\longrightarrow \mathcal{H}_0$ in Definition \ref{Definition: F}.
\begin{lem}
\label{Lemma: q function}
Let $q_1,\dots,q_d$ be the operators defined in Eq.~(\ref{Equation: Definition q}) and let $g:\mathbb{R}^d\rightarrow \mathbb{R}$ be a function such that $g|_{[-\epsilon,\epsilon]^d}=0$ for some $\epsilon>0$. Furthermore, assume that $g$ satisfies the growth condition $|g(t)|\leq C|t|^{2j}$, with $C>0$ and $j\in \mathbb{N}$. Then
\begin{align*}
g\left(q\right)=O_*\left(e^{-\delta N}\right)
\end{align*}
for some $\delta>0$.
\end{lem}
\begin{proof}
Using the elementary estimate $|t|^{2j}=\left(\sum_{r=1}^d t_r^2\right)^j\leq d^j\max_{1\leq r\leq d} t_r^{2j}$ yields
\begin{align*}
|g(t)|\leq  d^j C \sum_{r=1}^d t_r^{2j}\ \mathds{1}_{(\epsilon,\infty)}\left(|t_r|\right).
\end{align*}
In the following we want to verify that there exist constants $C$, $\delta>0$ and $\delta_0>0$ such that $\braket{q_r^{2j}\mathds{1}_{(\epsilon,\infty)}(|q_r|)}_{\Psi}\leq Ce^{-\delta N}$ for all states $\Psi\in \mathcal{W}_N\mathcal{F}_{\leq M}$,  $\|\Psi\|=1$, with $\frac{M}{N}\leq \delta_0$ and $r\in \{1,\dots,d\}$. Since $\mathcal{W}_N\, q_r\, \mathcal{W}_N^{-1}=q_r$, it is equivalent to verify this for $\Psi\in \mathcal{F}_{\leq M}$ instead. Due to the reflection symmetry $q_r\mapsto -q_r$ of $q_r^{2j}$, it is furthermore enough to verify that $\braket{\mathds{1}_{(\epsilon,\infty)}(q_r)q_r^{2j}}_{\Psi}\leq Ce^{-\delta N}$ for all states $\Psi\in \mathcal{F}_{\leq M}$ with $\frac{M}{N}\leq \delta_0$. Note that the operators $q_r:=\frac{1}{\sqrt{2N}}\frac{a_r+a_r^\dagger}{\sqrt{2}}$ depend on $N$. In the following we will make use of the description of the Fock space $\mathcal{F}_{\leq M}$ in terms of Hermite polynomials $h_n$, i.e. for $r\in \{1,\dots,d\}$ and $\Psi\in \mathcal{F}_{\leq M}$ there exist states $\Psi_n\in \mathcal{F}_{\leq M-n}$ with $a_r\Psi_n=0$, such that $\Psi=\sum_{n= 0}^M h_n\!\left(\frac{a_r+a_r^\dagger}{\sqrt{2}}\right)\Psi_n$, see for example  Eq.~(1.26), respectively Exercise 1(ii), in \cite{MR}. Furthermore we define the density matrix $\gamma_r(x,y):=\sum_{n_1,n_2=0}^M \braket{\Psi_{n_1},\Psi_{n_2}}h_{n_2}(x)h_{n_1}(y)\frac{1}{\sqrt{\pi}}e^{-\frac{x^2+y^2}{2}}$ on $L^2\left(\mathbb{R}\right)$. With $\gamma_r$ at hand we have
\begin{align*}
\braket{\mathds{1}_{(\epsilon,\infty)}(q_r)q_r^{2j}}_{\Psi}=\int_{\sqrt{2N}\epsilon}^\infty \left(\frac{x}{\sqrt{2N}}\right)^{2j}\gamma_r(x,x)\mathrm{d}x.
\end{align*}
In order to estimate this quantity, let us define the harmonic oscillator Hamiltonian $H:=-\frac{\mathrm{d}^2}{\mathrm{d} x^2}+x^2$ on $L^2(\mathbb{R})$. Since $\gamma_r$ involves only eigenfunctions $h_n(x)e^{-\frac{x^2}{2}}$ of $H$ with $n\leq M$, we have the operator inequality $\gamma_r\leq e^{2M+1-H}$. Using the Mehler kernel for $e^{-H}$ therefore yields for $c:=\frac{1}{\sqrt{2\pi\, \mathrm{sinh}(2)}}$ and $\lambda:=\mathrm{coth}(2)-\mathrm{cosech(2)}>0$, and all $M\leq \epsilon^2\lambda  N/2$
\begin{align*}
\int_{\sqrt{2N}\epsilon}^\infty \left(\frac{x}{\sqrt{2N}}\right)^{2j}\gamma_r(x,x)\mathrm{d}x\leq c\,  e^{\epsilon^2\lambda N+1 }\int_{\sqrt{2N}\epsilon}^\infty \left(\frac{x}{\sqrt{2N}}\right)^{2j}e^{-\lambda x^2}\mathrm{d}x=O_{N\rightarrow \infty}\left(e^{-\epsilon^2\lambda  N}\right).
\end{align*}
\end{proof}

The following Lemma is an auxiliary result, which will be useful for the verification of various asymptotic results involving the operator $b_{>d}$.
\begin{lem}
\label{Lemma: Auxiliary}
Recall the operators $\mathcal{W}_N,\mathbb{L}', p_j'$ and $f(q)$ from Definition \ref{Definition: Unitary Transformation} and the definition of $\hat{\pi}_{M,N}$ above Eq.~(\ref{Equation: Definition of hat Pi}). Then, $q_j$ and $p'_j$ commute with $\hat{\pi}_{M,N}$ for $j\in \{1,\dots,d\}$, $\hat{\pi}_{M,N}\mathbb{L}'\Psi=\mathbb{L}'\hat{\pi}_{M,N}\Psi$ for all $\Psi\in \mathcal{W}_N\mathcal{F}_{\leq N}$, and $b_k\, \hat{\pi}_{M,N}=\hat{\pi}_{M,N}\, b_k\, \hat{\pi}_{M,N}$ for all $k>d$. Furthermore, we have for all $M\leq N$ the estimate
\begin{align}
\label{Equation: b+f}
&\hat{\pi}_{M,N} \left(b_{>d}+ f(q)\right)^\dagger\cdot \left(b_{>d}+ f(q)\right) \hat{\pi}_{M,N}\leq \frac{M}{N},\\
\label{Equation: b only}
&\hat{\pi}_{M,N}\, b_{>d}^\dagger\cdot b_{>d}\, \hat{\pi}_{M,N}\leq 4.
\end{align}
\end{lem}

\begin{proof}
Recall $\mathcal{N}:=\sum_{j=1}^\infty a_j^\dagger a_j$ and let us define $\mathcal{N}_+:=\sum_{j>d}a_j^\dagger\cdot a_j$. Since $\mathcal{N}_+$ commutes with $\mathbb{L}$ and $q_j,p_j$ for $j\in \{1,\dots,d\}$, we obtain that $\hat{\pi}_{M,N}=\mathcal{W}_N\, \mathds{1}_{[0,M]}\left(\mathcal{N}_+\right) \mathcal{W}_N^{-1}$ commutes with $q_j=\mathcal{W}_N\, q_j\, \mathcal{W}_N^{-1}$ and $p_j'=\mathcal{W}_N\, p_j\, \mathcal{W}_N^{-1}$. Similarly $\hat{\pi}_{M,N}\mathbb{L}'\Psi=\mathbb{L}'\hat{\pi}_{M,N}\Psi$ for all $\Psi\in \mathcal{W}_N\mathcal{F}_{\leq N}$. Making use of the fact that $b_j=\mathcal{W}_N\, \left(b_j-f_j(q)\right)\, \mathcal{W}_N^{-1}$ yields
\begin{align*}
b_j\, \hat{\pi}_{M,N}&=\mathcal{W}_N \left(b_j-f_j(q)\right) \mathds{1}_{[0,M]}\left(\mathcal{N}_+\right)\mathcal{W}_N^{-1}\\
&=\mathcal{W}_N\, \mathds{1}_{[0,M]}\left(\mathcal{N}_+\right) \left(b_j-f_j(q)\right) \mathds{1}_{[0,M]}\left(\mathcal{N}_+\right) \mathcal{W}_N^{-1}=\hat{\pi}_{M,N}\, b_j\, \hat{\pi}_{M,N}.
\end{align*}
Inequality (\ref{Equation: b+f}) follows from $\big(b_{>d}+ f(q)\big)^\dagger\cdot \big(b_{>d}+ f(q)\big)=\frac{1}{N}\mathcal{W}_N\, \mathcal{N}_+\, \mathcal{W}_N^{-1}$ and
\begin{align*}
\hat{\pi}_{M,N}\, \left(b_{>d}+ f(q)\right)^\dagger\cdot \left(b_{>d}+ f(q)\right)\, \hat{\pi}_{M,N}=\frac{1}{N}\mathcal{W}_N\, \mathcal{N}_+ \mathds{1}_{[0,M]}\left(\mathcal{N}_+\right) \, \mathcal{W}_N^{-1}\leq \frac{M}{N}.
\end{align*}
In order to verify Inequality (\ref{Equation: b only}), note that $f(t)^\dagger\cdot f(t)\leq 1$ for all $t$. Applying the Cauchy--Schwarz inequality as in \ref{Lemma: O_* results} yields
\begin{align*}
\hat{\pi}_{M,N}\, b_{>d}^\dagger\cdot b_{>d}\,\hat{\pi}_{M,N}\!&\leq\! 2\hat{\pi}_{M,N} \left(b_{>d}+ f(q)\right)^\dagger\cdot \left(b_{>d}+ f(q)\right)\hat{\pi}_{M,N}\!+\!2\hat{\pi}_{M,N}\,  f(q)^\dagger\!\cdot  f(q)\, \hat{\pi}_{M,N}\\
&\leq 2\frac{M}{N}+2\leq 4.
\end{align*}
\end{proof}
The proof of the main technical Theorem \ref{Theorem: Decomposition} consists of two steps: First one has to identify the residuum $R_J$, which is carried out in the Lemmata \ref{Lemma: Taylor A} and \ref{Lemma: Taylor B}, and in the second step one has to derive asymptotic results for these residua $R_J$, which is carried out in the Theorems \ref{Theorem: Main A} and \ref{Theorem: Main B}. The following three Lemmata provide asymptotic results for the  types of operators most frequently encountered during our analysis of $R_J$.

\begin{lem}
\label{Lemma: Function of q}
Let $\phi,\Phi:\mathbb{R}^d\longrightarrow \mathbb{R}$ be functions with $|\phi(t)|\leq C|t|^k$ and $|\Phi(t)|\leq C(1+|t|^k)$ for some $k\geq 1$. Then, $\phi(q)=o_*(1)$ and $\Phi(q)=O_*(1)$. Furthermore,
\begin{align}
\label{Equation: phi-b-b}
\phi(q)\, b_{>d}^\dagger\cdot b_{>d}&=o_*\left(b_{>d}^\dagger\cdot b_{>d}+\frac{1}{N}\right),\\
\label{Equation: Phi-b-b}
\Phi(q)\, b_{>d}^\dagger\cdot b_{>d}&=O_*\left(b_{>d}^\dagger\cdot b_{>d}+\frac{1}{N}\right).
\end{align}
\end{lem}
\begin{proof}
In the following, let $0\leq \tau\leq 1$ be a smooth function with $\mathrm{supp}\left(\tau\right)\subset B_{1}(0)$ and $\tau(t)=1$ for all $t\in B_{\frac{1}{2}}(0)$, and let $\tau_r(t):=\tau\left(\frac{t}{r}\right)$ for $r>0$. Clearly $\phi(q)=\tau_r(q) \phi(q)+(1-\tau_r(q)) \phi(q)$. By our assumptions we know that $|\tau_r \phi|\leq \epsilon_r$ with $\epsilon_r\underset{r\rightarrow 0}{\longrightarrow}0$ and $(1-\tau_r) \phi$ is zero in a neighborhood of zero, hence $(1-\tau_r(q)) \phi(q)=O_*\left(e^{-\delta N}\right)$ by Lemma \ref{Lemma: q function}. We conclude that $|\phi(q)|\leq \epsilon_r+O_*\left(e^{-\delta N}\right)$ for all $r>0$, and consequently $\phi(q)=o_*(1)$. The corresponding statement for $\Phi(q)$ follows from the fact that $\Phi(q)\leq C+\phi(q)$ with $\phi(t):=|t|^k$ and $\phi(q)=o_*(1)$.

Let us write similar to before $\phi(q)\, b_{>d}^\dagger\cdot b_{>d}=\tau_r(q)\phi(q)\, b_{>d}^\dagger\cdot  b_{>d}+(1-\tau_r(q))\phi(q)\, b_{>d}^\dagger\cdot b_{>d}$. In order to verify Eq.~(\ref{Equation: Phi-b-b}). First of all $\tau_r(q)  \phi(q) \,  b_{>d}^\dagger \cdot   b_{>d}\!\leq \epsilon_r  b_{>d}^\dagger \cdot  b_{>d}$, where we use that $q$ commutes with $b_{>d}$. For the treatment of the second term, recall Inequality (\ref{Equation: b only}) and $\pi_{M,N} (1-\tau_r(q))^2 |\phi_r(q)|^2 \pi_{M,N}\leq C^2e^{-2\delta N}$ for $\frac{M}{N}\leq \delta$ with $C,\delta>0$, which follows from Lemma \ref{Lemma: q function}. Hence,
\begin{align*}
&\pi_{M,N} (1\!-\!\tau_r(q)) |\phi(q)|\, b_{>d}^\dagger\cdot b_{>d}\, \pi_{M,N}\!=\!\pi_{M,N}(1\!-\!\tau_r(q)) |\phi(q)| \hat{\pi}_{M,N}\, b_{>d}^\dagger\cdot b_{>d}\, \pi_{M,N}\\
&\ \leq \Big\|\pi_{M,N} (1\!-\!\tau_r(q))|\phi(q)|\Big\| \ \Big\|\hat{\pi}_{M,N}\, b_{>d}^\dagger\cdot b_{>d}\,\hat{\pi}_{M,N}\Big\|\leq 4Ce^{-\delta N}.
\end{align*}
We conclude that $\pi_{M,N}\phi(q)\pi_{M,N}\leq 4Ce^{-\delta N}+\epsilon_r b_{>d}^\dagger \cdot  b_{>d}$ for $\frac{M}{N}\leq \delta$, and therefore $\phi(q)=o_*\left(b_{>d}^\dagger\cdot b_{>d}+\frac{1}{N}\right)$. The corresponding statement for $\Phi(q)\, b_{>d}^\dagger\cdot b_{>d}$ follows as above.
\end{proof}

\begin{lem}
\label{Lemma: b}
Given $w:\mathbb{R}^d\longrightarrow \mathcal{H}$ with $\|w(t)\|\leq c\ |t|^k$ and $W:\mathbb{R}^d\longrightarrow \mathcal{H}$ with $\|W(t)\|\leq c\left(1+|t|^k\right)$ for some $c>0$ and $k\geq 1$, we define $X:= W(q)^\dagger\cdot b_{>d}$ and $Y:= w(q)^\dagger\cdot b_{>d}$. Then, $X^\dagger X$ and $X X^\dagger$ are of order $O_*\left(b_{>d}^\dagger\cdot b_{>d}+\frac{1}{N}\right)$, and $Y^\dagger Y$ and $Y Y^\dagger$ are of order $o_*\left(b_{>d}^\dagger\cdot b_{>d}+\frac{1}{N}\right)$. Furthermore, for $\Phi:\mathbb{R}^d\longrightarrow \mathbb{R}$ with $|\Phi(t)|\leq c\left(1+|t|^j\right)$, we obtain
\begin{align*}
&\Phi(q) \left(b_{>d}^\dagger\cdot b_{>d}\right)^2=o_*\left(b_{>d}^\dagger\cdot b_{>d}+\frac{1}{N}\right).
\end{align*}
Recall the operator $ p'$ from Definition \ref{Definition: Unitary Transformation}. We have
\begin{align*}
\left( p'- p\right)^\dagger\cdot \left( p'- p\right)=o_*\left(b_{>d}^\dagger\cdot b_{>d}+\frac{1}{N}\right).
\end{align*}
\end{lem}
\begin{proof}
Let us define $G(t):=W(t)^\dagger\cdot W(t)$ and $g(t):=w(t)^\dagger\cdot w(t)$. Then we obtain by Lemma \ref{Lemma: Function of q} together with the inequality $W(t)\cdot  W(t)^\dagger\leq G(t)\ 1_\mathcal{H}$ the estimate
\begin{align*}
X^\dagger X&=b_{>d}^\dagger\cdot  W(q)\cdot  W(q)^\dagger\cdot b_{>d}\leq G(q)\, b_{>d}^\dagger\cdot b_{>d}=O_*\left(b_{>d}^\dagger\cdot b_{>d}+\frac{1}{N}\right).
\end{align*}
Similarly, $Y^\dagger Y\leq g(q)\, b_{>d}^\dagger\cdot b_{>d}=o_*\left(b_{>d}^\dagger\cdot b_{>d}+\frac{1}{N}\right)$. For the reversed order, we use the fact that $\|G(q)\|^2=O_*\left(1\right)$ and $\|g(q)\|^2=o_*\left(1\right)$
\begin{align*}
X X^\dagger&=X^\dagger X+\frac{1}{N}\|G(q)\|^2=O_*\left(b_{>d}^\dagger\cdot b_{>d}+\frac{1}{N}\right),\\
Y Y^\dagger&=Y^\dagger Y+\frac{1}{N}\|g(q)\|^2=o_*\left(b_{>d}^\dagger\cdot b_{>d}+\frac{1}{N}\right).
\end{align*}
For the next statement, note that we have $\left(b_{>d}^\dagger\cdot b_{>d}\right)^2=b_{>d}^\dagger\cdot \left(b_{>d}^\dagger\cdot b_{>d}+\frac{1}{N}\right)\cdot b_{>d}$ and $b_{>d}^\dagger\cdot b_{>d}\leq 2(b_{>d}\!+\! f(q))^\dagger\!\cdot \!(b_{>d}\!+\! f(q))+2f(q)^\dagger\!\cdot\! f(q)$, and consequently
\begin{align*}
&\Phi(q) \left(b_{>d}^\dagger\cdot b_{>d}\right)^2=\Phi(q)\, b_{>d}^\dagger\cdot \left(b_{>d}^\dagger\cdot b_{>d}+\frac{1}{N}\right)\cdot b_{>d}\\
&\leq 2 b_{>d}^\dagger\!\cdot\!  (b_{>d}\!+\! f(q))^\dagger\!\cdot\!\Phi(q)\!\cdot \!(b_{>d}\!+\! f(q))\!\cdot\! b_{>d}\!+\!2\Phi(q) f(q)^\dagger\!\cdot\! f(q)\, b_{>d}^\dagger\!\cdot\! b_{>d}\!+\!\frac{\Phi(q)}{N}\, b_{>d}^\dagger\!\cdot\! b_{>d}.
\end{align*}
Note that $2\Phi(q) f(q)^\dagger\!\cdot\! f(q)\, b_{>d}^\dagger\!\cdot\! b_{>d}$ and $\frac{\Phi(q)}{N}\, b_{>d}^\dagger\!\cdot\! b_{>d}$ are of order $o_*\left(b_{>d}^\dagger\cdot b_{>d}+\frac{1}{N}\right)$ by Lemma \ref{Lemma: Function of q}. For the other term in the inequality above, note that we have the estimate
\begin{align*}
&\ \ \ \ \pi_{M,N}\,  b_{>d}^\dagger\cdot \left[\Phi(q) ( b_{>d}+ f(q))^\dagger\cdot ( b_{>d}+ f(q))\right]\cdot  b_{>d}\, \pi_{M,N}\\
&=\pi_{M,N}\,  b_{>d}^\dagger\cdot \left[\Phi(q) \hat{\pi}_{M,N} ( b_{>d}+ f(q))^\dagger\cdot ( b_{>d}+ f(q)) \hat{\pi}_{M,N}\right]\cdot  b_{>d}\, \pi_{M,N}\\
&\leq \frac{M}{N}\ \pi_{M,N}\Phi(q)\,  b_{>d}^\dagger\cdot  b_{>d}\, \pi_{M,N}\leq C\ \frac{M}{N}\pi_{M,N} \left( b_{>d}^\dagger\cdot  b_{>d}+\frac{1}{N}\right)\, \pi_{M,N},
\end{align*}
where we have used that $\pi_{M,N}\Phi(q)\,  b_{>d}^\dagger\cdot  b_{>d}\, \pi_{M,N}\leq C\ \pi_{M,N}\left( b_{>d}^\dagger\cdot  b_{>d}+\frac{1}{N}\right) \pi_{M,N}$ for $\frac{M}{N}\leq \delta_0<1$, see Lemma \ref{Lemma: Function of q}. In order to verify the last part of the Lemma, let us define the operators $Y_\ell:=\partial_\ell f(q)^\dagger\cdot  b_{>d}$. From the previous part of this Lemma we know
\begin{align*}
\left( p'- p\right)^\dagger\cdot \left( p'- p\right)=\sum_{\ell=1}^d \mathfrak{Im}\left[Y_\ell\right]^2\leq \frac{1}{2}\sum_{\ell=1}^d\left(Y_\ell^\dagger Y_\ell+Y_\ell Y_\ell^\dagger\right)=o_*\left( b_{>d}^\dagger\cdot  b_{>d}+\frac{1}{N}\right).
\end{align*}
\end{proof}

For the following Lemma \ref{Lemma: p'} as well as for the results in Appendix \ref{Appendix: C}, it is convenient to define the operator
\begin{align}
\label{Equation: Definition of Q}
\mathbb{Q}_N:= p^\dagger\cdot  p+ b_{>d}^\dagger\cdot  b_{>d}+\frac{1}{N}.
\end{align}
Since $\mathbb{Q}_N\leq \mathbb{T}_N$, where $\mathbb{T}_N$ is defined in Eq.~(\ref{Equation: Definition T}), any sequence with $X_N=O_*\left(\mathbb{Q}_N\right)$, respectively $X_N=o_*\left(\mathbb{Q}_N\right)$, satisfies $X_N=O_*\left(\mathbb{T}_N\right)$, respectively $X_N=o_*\left(\mathbb{T}_N\right)$, as well.\\

\begin{lem}
\label{Lemma: p'}
Let $\phi:\mathbb{R}^d\longrightarrow \mathbb{R}$ be a function with $|\phi(t)|\leq c\ |t|^k$ and $\Phi:\mathbb{R}^d\longrightarrow \mathbb{R}$ with $|\Phi(t)|\leq c\left(1+|t|^k\right)$ for some constant $c$ and $k\geq 1$. Then
\begin{align*}
\left( p'\right)^\dagger\cdot\phi(q) \cdot  p'=o_*\left(\mathbb{Q}_N\right),\\
\left( p'\right)^\dagger\cdot\Phi(q) \cdot  p'=O_*\left(\mathbb{Q}_N\right).
\end{align*}
In case the partial derivatives $\partial_{i}\phi(t)$, $\partial_j\Phi(t)$ and $\partial_i \partial_j\Phi(t)$ are bounded by $c\left(1+|t|^j\right)$, we also have
\begin{align*}
\phi(q) \left( p'\right)^\dagger\cdot  p'\, \phi(q)=o_*\left(\mathbb{Q}_N\right),\\
\Phi(q) \left[\left( p'\right)^\dagger\cdot  p'\right]^2 \Phi(q)=o_*\left(\mathbb{Q}_N\right).\\
\end{align*}
\end{lem}
\begin{proof}
Since $ p^\dagger\cdot  p\leq \mathbb{Q}_N$ and $\left( p'- p\right)^\dagger\cdot \left( p'- p\right)=O_*\left(\mathbb{Q}_N\right)$ by Lemma \ref{Lemma: b}, we obtain $( p')^\dagger\cdot  p'=O_*\left(\mathbb{Q}_N\right)$ as well, i.e.
\begin{align*}
\pi_{M,N}( p')^\dagger\cdot  p' \pi_{M,N}\leq C_1\ \pi_{M,N}\mathbb{Q}_N \pi_{M,N}
\end{align*}
for all $M,N$ with $\frac{M}{N}\leq \delta_1<1$ where $\delta_1$ and $C_1$ are suitable constants. By Lemma \ref{Lemma: Function of q}, we know that $\pi_{M,N} \Phi(q) \pi_{M,N}\leq C_2$ for all $\frac{M}{N}\leq \delta_2<1$ where $\delta_2$ and $C_2$ are suitable constants, and $\pi_{M,N} \phi(q) \pi_{M,N}\leq \epsilon\left(\frac{M}{N}\right)$ with $\lim_{\delta\rightarrow 0}\ \epsilon(\delta)=0$. Based on the observation that $ p'\, \pi_{M,N}=\pi_{M+1,N}\,  p'\, \pi_{M,N}$, we obtain for all $M,N$ that satisfy $\frac{M}{N}\leq \delta:=2\mathrm{min}\{\delta_1,\delta_2\}$
\begin{align*}
&\pi_{M,N} \left( p'\right)^\dagger\cdot\Phi(q) \cdot  p'\, \pi_{M,N}\!=\!\pi_{M,N} \left( p'\right)^\dagger\cdot\pi_{M\!+\!1,N}\Phi(q) \pi_{M\!+\!1,N} \cdot  p'\, \pi_{M,N}\\
&\ \ \ \ \leq C_1C_2\ \pi_{M,N}\cdot\mathbb{Q}_N\cdot \pi_{M,N}.
\end{align*}
Similarly, we have $\pi_{M,N} \left( p'\right)^\dagger\cdot\phi(q)\cdot  p'\, \pi_{M,N}\leq C_1\epsilon\left(\frac{M}{N}\right)\ \pi_{M,N}\, \mathbb{Q}_N\, \pi_{M,N}$. Hence, $\left( p'\right)^\dagger\cdot\Phi(q) \cdot  p'=O_*\left(\mathbb{Q}_N\right)$ and $\left( p'\right)^\dagger\cdot\phi(q) \cdot  p'=o_*\left(\mathbb{Q}_N\right)$. In case we have a polynomial bound on the partial derivatives as well, let us define $w(t):=\frac{1}{2}\sum_{\ell=1}^d \partial_\ell \phi(t)\otimes u_\ell$ and $W(t):=\frac{1}{2}\sum_{\ell=1}^d \partial_\ell \Phi(t)\otimes u_\ell$. Using the commutation relation $[ip'_j,q_k]=\frac{\delta_{j,k}}{2N}$, we compute
\begin{align*}
\phi(q)\, \left( p'\right)^\dagger\cdot  p'\, \phi(q)=\left(\phi(q)\cdot i p'+\frac{1}{N} w(q)\right)^\dagger\cdot \left(\phi(q)\cdot i p'+\frac{1}{N} w(q)\right).
\end{align*}
From the previous part, we know that $\left(\phi(q)\cdot i p'\right)^\dagger\cdot \phi(q)\cdot i p'=o_*\left(\mathbb{Q}_N\right)$. Furthermore, Lemma \ref{Lemma: Function of q} tells us that $w(q)^\dagger\cdot w(q)=O_*\left(1\right)$, and therefore $\frac{1}{N} w(q)^\dagger\cdot \frac{1}{N} w(q)=o_*\left(\mathbb{Q}_N\right)$. Hence, $\phi(q) \left( p'\right)^\dagger\cdot  p'\, \phi(q)$ is of order $o_*\left(\mathbb{Q}_N\right)$ as well. The last estimate in the Lemma can be verified analogously.
\end{proof}

\section{Analysis of the Operator Square Root}
In the following section we derive asymptotic results for operators involving the square root $\sqrt{1-\mathbb{L}'}$, where $\mathbb{L}'$ is defined in Definition \ref{Definition: Unitary Transformation}, allowing us to prove a Taylor approximation for the operators $\left(1-\mathbb{L}'\right)^{\frac{m}{2}}$, see Definition \ref{Definition: Taylor of the square root}. The easiest case $m=2$ will be discussed in the following Lemma \ref{Lemma: Taylor of the square root - order two}, the case $m=1$ is the content of Lemma \ref{Lemma: Taylor of the square root} and the case $m=3$ is covered by Corollary \ref{Corollary: Taylor of the square root}.
\label{Appendix: C}
\begin{lem}
\label{Lemma: Taylor of the square root - order two}
Recall the operator $\mathbb{Q}_N$ from Eq.~(\ref{Equation: Definition of Q}) and the function $f$ from Definition \ref{Definition: F}, and let us define $g(t):=\sum_{j=1}^d t_j^2+f(t)^\dagger\cdot f(t)$. Then, 
\begin{align}
\label{Equation: Approximation of L'}
\left[\mathbb{L}'-g(q)\right]^2=o_*\left(\mathbb{Q}_N\right).
\end{align}
\end{lem}
\begin{proof}
Using the transformation laws in Lemma \ref{Lemma: Transformation Laws} we obtain
\begin{align*}
\mathbb{L}'-g(q)= f(q)^\dagger\cdot  b_{>d}+ b_{>d}^\dagger\cdot  f(q)+\left( p'\right)^\dagger\cdot  p'+ b_{>d}^\dagger\cdot  b_{>d}-\frac{d}{2N}.
\end{align*}
By Lemma \ref{Lemma: b}, we know that $\left[ f(q)^\dagger\cdot  b_{>d}+ b_{>d}^\dagger\cdot  f(q)\right]^2=o_*\left(\mathbb{Q}_N\right)$ and $\left[ b_{>d}^\dagger\cdot  b_{>d}\right]^2=o_*\left(\mathbb{Q}_N\right)$, and by Lemma \ref{Lemma: p'} we know that $\left[\left( p'\right)^\dagger\cdot  p'\right]^2=o_*\left(\mathbb{Q}_N\right)$. 
\end{proof}

\begin{lem}
\label{Lemma: Taylor of the square root}
Let Assumption \ref{Assumption: Part II} hold and recall the function $\eta_1$ from Eq.~(\ref{Equation: eta_m}). Then,
\begin{align}
\label{Equation: First order Taylor residuum of sqrt}
\left[\sqrt{1-\mathbb{L}'}-\eta_1(q)\right]^2=o_*\left(\mathbb{Q}_N\right).
\end{align}
Furthermore for any function $V:\mathbb{R}^d\longrightarrow \mathbb{R}$ with $|V(t)|\leq c\left(|t|+|t|^k\right)$ and bounded derivatives $|\partial_{t_i} V(t)|+|\partial_{t_i}\partial_{t_j} V(t)|\leq c\left(1+|t|^k\right)$ for some $k\geq 1$, we have
\begin{align}
\label{Equation: Second order Taylor residuum of sqrt}
V(q)  &\Big[\sqrt{1-\mathbb{L}'}-\eta_1(q)-D_{\mathcal{V}}\big|_q \eta_1\big(b_{\geq 1}\big)\Big]=o_*\left(\mathbb{Q}_N\right).
\end{align}
\end{lem}
\begin{proof}
Let us define $h(x):=\chi\left(x\right)\sqrt{1-x}$, where $\chi:[0,\infty)\longrightarrow [0,1]$ is the function from the definition of $\eta_1$ in Eq.~(\ref{Equation: eta_m}), as well as the operator $Q:= q^\dagger\cdot  q+ f(q)^\dagger \cdot  f(q)$. By the support properties of $\chi$ we have for all $\frac{M}{N}<\frac{1}{2}$ and $\Psi\in \mathcal{W}_N\mathcal{F}_{\leq M}$
\begin{align*}
\sqrt{1-\mathbb{L}'}\, \Psi=h\left(\mathbb{L}'\right) \Psi,
\end{align*}
and therefore it is enough to verify the statements of this Lemma for $h\left(\mathbb{L}'\right)$ instead of $\sqrt{1-\mathbb{L}'}$. With $h$ at hand, we have $\eta_1(q)=h\left(\|F(q)\|^2\right)=h\left(Q\right)$ and 
\begin{align*}
D_{\mathcal{V}}\big|_q \eta_1(v)=w(q)^\dagger\cdot v+v^\dagger\cdot w(q)
\end{align*}
with $w(t):=h'\left(\sum_{j=1}^d t_j^2+f(t)^\dagger\cdot f(t)\right) f(t)$, for all $v\in \mathcal{H}_0$. Hence $w(q)=h'\left(Q\right) f(q)$. In the following, let $\hat{h}$ be the Fourier transform of the smooth function $h$, normalized such that $h(x)=\int \hat{h}(z)\ e^{iz x}\ \mathrm{d}z$. Then,
\begin{align*}
h\left(\mathbb{L}'\right)-h\left(Q\right)=\int \hat{h}(z)\left(e^{iz\mathbb{L}'}-e^{izQ}\right)\ \mathrm{d}z.
\end{align*}
In order to investigate the integrand, we use the following integral representation
\begin{align*}
e^{iz\mathbb{L}'}-e^{izQ}&=i\int_0^z e^{iy\mathbb{L}'}\left(\mathbb{L}'-Q\right)e^{i(z-y)Q}\ \mathrm{d}y\\
&=i\int_0^z e^{iy\mathbb{L}'}e^{i(z-y)Q}\ \mathrm{d}y\left(\mathbb{L}'-Q\right)+i\int_0^z e^{iy\mathbb{L}'}\left[\mathbb{L}',e^{i(z-y)Q}\right]\ \mathrm{d}y.
\end{align*}
Let us define the operators $B_z:=i\int_0^z e^{iy\mathbb{L}'}e^{i(z-y)Q}\ \mathrm{d}y$ and $R_z:=i\int_0^z e^{iy\mathbb{L}'}\left[\mathbb{L}',e^{i(z\!-\!y)Q}\right]\mathrm{d}y$. Clearly, $\|B_z\|\leq |z|$. Regarding $R_z$, note that every term in the definition of $\mathbb{L}'$ commutes with $Q$, except $\left( p'\right)\cdot  p'$, which satisfies the relation $\left[p_j',\phi(q)\right]=\left[p_j,\phi(q)\right]=\frac{1}{i2N}\left(\partial_j \phi\right)(q)$. We define the family of functions
\begin{align}
\label{Equation: Definition of phi_x}
\phi_x(t):=e^{ix\left(\sum_{j=1}^d t_j^2+f(t)^\dagger\cdot f(t)\right)}
\end{align}
and compute
\begin{align*}
\left[\mathbb{L}',e^{ixQ}\right]&=\left[\left( p'\right)\cdot  p',\phi_{x}(q)\right]=\sum_{j=1}^d\left[\left(p_j'\right)^2,\phi_{x}(q)\right]\\
&=\frac{1}{iN}\sum_{j=1}^d \partial_j \phi_{x}(q) p_j'-\frac{1}{4N^2}\sum_{j=1}^d \partial_j^2 \phi_{x}(q).
\end{align*}
We have the estimates $\left|\partial_j \phi_{x}(t)\right|\leq c|x|\ |t|$ and $\left|\partial^2_j \phi_{x}(t)\right|\leq c\left(1+|x|^2\right) \left(1+|t|^2\right)$ for some $c>0$, where we use the fact that $t\mapsto f(t)$ is a $C^2\left(\mathbb{R}^d,\mathcal{H}_0\right)$ function, see Definition \ref{Definition: F}. As before, let $\pi_{M,N}$ be the orthogonal projection onto $\mathcal{W}_N\left(\mathcal{F}_{\leq M}\right)$. By Lemma \ref{Lemma: Function of q}
\begin{align*}
\|\partial_j \phi_{x}(q) \pi_{M,N}\|&\leq c|x| \ \|\ |q|\, \pi_{M,N} \|\leq \tilde{c}|x|,\\
\|\partial_j^2 \phi_{x}(q) \pi_{M,N}\|&\leq c\left(1+|x|^2\right)\|\left(1+|q|^2\right) \pi_{M,N}\|\leq \tilde{c}\left(1+|x|^2\right),
\end{align*}
for some constant $\tilde{c}$ and all $x\in \mathbb{R}$ and all $M\leq N$. Note that $p_j' \mathcal{W}_N\mathcal{F}_{\leq M}\subset \mathcal{W}_N\mathcal{F}_{\leq M+1}$ and $\|p_j'\, \pi_{M,N}\|\leq \sqrt{\frac{M+1}{N}}$, and consequently we have for all $M\leq N-1$
\begin{align*}
\|\!\left[\mathbb{L}',e^{ixQ}\right] \pi_{M,N}\|&\leq \frac{2}{N}\sum_{j=1}^d\! \|\partial_j \phi_{x}(q)\, \pi_{M\!+\!1,N}\|\ \|p_j'\, \pi_{M,N}\|\!+\!\frac{1}{N^2}\sum_{j=1}^d\!\|\partial_j^2 \phi_{x}(q)\, \pi_{M,N}\|\\
&\leq \frac{d\tilde{c}}{N}|x|+\frac{d\tilde{c}}{4N^2}\left(1+|x|^2\right)\leq \frac{2d\tilde{c}}{N}\left(1+|x|^2\right).
\end{align*}
Therefore, $\|R_z \pi_{M,N}\|\leq \frac{C}{N}\left(1+|z|^3\right)$ for some constant $C$. 

Let us define $B:=\int \hat{h}(z)B_z\ \mathrm{d}z$ and $R:=\int \hat{h}(z)R_z\ \mathrm{d}z$. From our estimates on $B_z,R_z$, we deduce $\|B\|\leq \int |\hat{h}(z)|\ |z|\ \mathrm{d}z:=C_1<\infty$ and $\|R\, \pi_{M,N}\|\leq \frac{C}{N}\int |\hat{h}(z)|\left(1+|z|^3\right)\ \mathrm{d}z:=\frac{C_2}{N}<\infty$. Hence, $R^\dagger R=o_*\left(\mathbb{Q}_N\right)$. Since $h\left(\mathbb{L}'\right)-h\left(Q\right)=B\, \left(\mathbb{L}'-Q\right)+R$, we obtain the estimate
\begin{align*}
&\left[h\left(\mathbb{L}'\right)-h\left(Q\right)\right]^2=[B \left(\mathbb{L}'-Q\right)+R]^\dagger [B \left(\mathbb{L}'-Q\right)+R]\\
&\ \ \ \leq 2\left(\mathbb{L}'-Q\right) B^\dagger B \left(\mathbb{L}'-Q\right)+2R^\dagger R\\
&\ \ \ \leq 2(C_1)^2\left(\mathbb{L}'-Q\right)^2+2R^\dagger R=o_*\left(\mathbb{Q}_N\right),
\end{align*}
where we have used that $\left(\mathbb{L}'-Q\right)^2$ is of order $o_*\left(\mathbb{Q}_N\right)$, see Lemma \ref{Lemma: Taylor of the square root - order two}. This proves Eq.~(\ref{Equation: First order Taylor residuum of sqrt}).

In order to verify Eq.~(\ref{Equation: Second order Taylor residuum of sqrt}) let us compute
\begin{align}
\nonumber &\sqrt{1-\mathbb{L}'}-\eta_1(q)-D_{\mathcal{V}}\big|_q \eta_1\big(b_{\geq 1}\big)=\!h\left(\mathbb{L}'\right)\!-\!h\left(Q\right)-h'\left(Q\right)\left( b_{>d}^\dagger\cdot  f(q)\!+\! f(q)^\dagger\cdot  b_{>d}\right)\\
\nonumber &\ \ =\int \hat{h}(z)\left[\!i\!\int_0^z e^{iy\mathbb{L}'} \left(\mathbb{L}'\!-\!Q\right) e^{i(z-y)Q}\ \mathrm{d}y-iz\ e^{izQ}\left( f(q)^\dagger\cdot  b_{>d}\!+\! b_{>d}^\dagger\cdot  f(q)\right)\right]\ \mathrm{d}z\\
\nonumber &\ \ =R\!+\!\int \hat{h}(z)\left[\!i\!\int_0^z e^{iy\mathbb{L}'}e^{i(z-y)Q} \mathrm{d}y \left(\mathbb{L}'\!-\!Q\right)\!-\!iz\ e^{izQ}\left( f(q)^\dagger\cdot  b_{>d}\!+\! b_{>d}^\dagger\cdot  f(q)\right)\right] \mathrm{d}z\\
\label{Equation: Square Root Second order expansion} &\ \ =R+i\int \hat{h}(z)\int_0^z \left(e^{iy\mathbb{L}'}-e^{iyQ} \right) e^{i(z-y)Q}\ \mathrm{d}y\ \mathrm{d}z \left(\mathbb{L}'-Q\right)  \\
\nonumber &\ \ \ \ \ +i\int \hat{h}(z)z\ e^{izQ}\ \mathrm{d}z \left(\mathbb{L}'-Q- f(q)^\dagger\cdot  b_{>d}- b_{>d}^\dagger\cdot  f(q)\right).
\end{align}
Let $V$ be a function that satisfies the assumptions of the Lemma. To complete the proof, we need to verify that $V(q) \left[h\left(\mathbb{L}'\right)-h\left(Q\right)-h'\left(Q\right) \left( f(q)^\dagger\cdot  b_{>d}+ b_{>d}^\dagger\cdot  f(q)\right)\right]$ is of order $o_*\left(\mathbb{Q}_N\right)$. By Lemma \ref{Lemma: Function of q}, we know that $|V|^2(q)=o_*(1)$ and from the previous part it is clear that $\pi_{M,N}\, R^\dagger R\, \pi_{M,N}=O_*(\frac{1}{N^2})$. Hence, $V(q) R=o_*(\frac{1}{N})$ and especially $V(q) R=o_*(\mathbb{Q}_N)$. Regarding the second term in Eq.~(\ref{Equation: Square Root Second order expansion}), recall that $e^{iy\mathbb{L}'}-e^{iyQ}=\left(\mathbb{L}'-Q\right)\, B_{-y}^\dagger+R_{-y}^\dagger$. Therefore,
\begin{align*}
&\int \hat{h}(z)i\int_0^z \left(e^{iy\mathbb{L}'}-e^{iyQ} \right) e^{i(z-y)Q}\ \mathrm{d}y\  \mathrm{d}z \left(\mathbb{L}'-Q\right)\\
&\ \ \ =\int \hat{h}(z)i\int_0^z \left[\left(\mathbb{L}'-Q\right) B_{-y}^\dagger+R^\dagger_{-y}\right] e^{i(z-y)Q}\ \mathrm{d}y\ \mathrm{d}z \left(\mathbb{L}'-Q\right)\\
&\ \ \ =\left[\left(\mathbb{L}'-Q\right) \widetilde{B}^\dagger+\widetilde{R}^\dagger\right] \left(\mathbb{L}'-Q\right),
\end{align*}
with $\widetilde{B}:=-i\int \hat{h}(z)\int_0^z e^{i(y-z)Q} B_{-y}\ \mathrm{d}y\ \mathrm{d}z$ and $\widetilde{R}:=-i\int \hat{h}(z)\int_0^z  e^{i(y-z)Q} R_{-y}\ \mathrm{d}y\ \mathrm{d}z$. In the following we want to verify that $V(q) \left[\left(\mathbb{L}'-Q\right) \widetilde{B}^\dagger+\widetilde{R}^\dagger\right] \left(\mathbb{L}'-Q\right)=o_*\left(\mathbb{Q}_N\right)$. Since $\left(\mathbb{L}'-Q\right)^2=o_*\left(\mathbb{Q}_N\right)$ by Lemma \ref{Lemma: Taylor of the square root - order two}, it is enough to verify that $V(q) \widetilde{R}^\dagger  \widetilde{R} V(q)$ and $V(q) \left(\mathbb{L}'-Q\right) \widetilde{B}^\dagger \widetilde{B} \left(\mathbb{L}'-Q\right) V(q)$ are of order $o_*\left(\mathbb{Q}_N\right)$. Recall that we have the identity $R_y=i\int_0^y e^{ix\mathbb{L}'}\left[\mathbb{L}'\!-\!Q,e^{i(y\!-\!x)Q}\right]\mathrm{d}x=i\int_0^y e^{ix\mathbb{L}'}\left[(p')^\dagger\cdot p',\phi_x(q)\right]\mathrm{d}x$ with the function $\phi_x$ from Eq.~(\ref{Equation: Definition of phi_x}). We can further express $[(p')^\dagger\cdot p',\phi_x(q)]V(q)$ as
\begin{align*}
\sum_{j=1}^d\left(\frac{1}{iN}\partial_j \phi_{x}(q) V(q) p_j'- \frac{1}{2N^2}\partial_j \phi_{x}(q) \partial_j V(q)-\frac{1}{4N^2} \partial_j^2 \phi_{x}(q) V(q)\right).
\end{align*}
Similar to before, this leads to the estimate $\|R_z V(q) \pi_{M,N}\|\leq \frac{\widetilde{C}}{N}\left(1+|z|^3\right)$ for some constant $\widetilde{C}$, and consequently $\|\widetilde{R} V(q) \pi_{M,N}\|\leq \frac{\widetilde{C}_1}{N}$ for some constant $\widetilde{C}_1$. Hence we have $V(q) \widetilde{R}^\dagger  \widetilde{R} V(q)=o_*\left(\mathbb{Q}_N\right)$. Regarding the term $V(q) \left(\mathbb{L}'-Q\right) \widetilde{B}^\dagger \widetilde{B} \left(\mathbb{L}'-Q\right) V(q)$, note that $\|\widetilde{B}\|^2=:\tilde{C}_2<\infty$. Applying the Cauchy--Schwarz yields
\begin{align*}
&V(q) \left(\mathbb{L}'-Q\right) \widetilde{B}^\dagger \widetilde{B} \left(\mathbb{L}'-Q\right) V(q)\leq \tilde{C}_2\ V(q) \left(\mathbb{L}'-Q\right)^2 V(q)\\
&\!\leq\! 5 \tilde{C}_2 V(q)\!\left[\!f(q)^\dagger\!\cdot \! b_{>d}\, b_{>d}^\dagger\!\cdot\!  f(q) \!+\! b_{>d}^\dagger\!\cdot\!  f(q) f(q)^\dagger\!\cdot\!  b_{>d}\!+\!\left(\!( p')^\dagger\!\cdot\!  p'\!\right)^2\!+\!\left(\! b_{>d}^\dagger\!\cdot\!  b_{>d}\!\right)^2\! +\!\frac{d^2}{N^2}\!\right]\!V(q).
\end{align*}
Let us define the function $w(t):=V(t)f(t)$. By Lemma \ref{Lemma: b} we obtain that
\begin{align*}
V(q)\,  f(q)^\dagger\cdot  b_{>d}\  b_{>d}^\dagger\cdot  f(q)\, V(q)=w(q)^\dagger\cdot b_{>d}\  b_{>d}^\dagger\cdot w(q)&=o_*\left(\mathbb{Q}_N\right),\\
V(q)\,  b_{>d}^\dagger\cdot  f(q)\  f(q)^\dagger\cdot  b_{>d}\, V(q)= b_{>d}^\dagger\cdot  w(q)\,  w(q)^\dagger\cdot  b_{>d}&=o_*\left(\mathbb{Q}_N\right),
\end{align*}
and $V(q) \left( b_{>d}^\dagger\cdot  b_{>d}\right)^2 V(q)=o_*\left(\mathbb{Q}_N\right)$. Furthermore, $V(q) \left[\left( p'\right)^\dagger\cdot  p'\right]^2 V(q)=o_*\left(\mathbb{Q}_N\right)$ by Lemma \ref{Lemma: p'}. We conclude that $V(q) \left(\mathbb{L}'-Q\right) \widetilde{B}^\dagger \widetilde{B} \left(\mathbb{L}'-Q\right) V(q)=o_*\left(\mathbb{Q}_N\right)$. 

Let us now verify that the final term $\widetilde{V}(q) \left(\mathbb{L}'-Q- b_{>d}^\dagger\cdot  f(q)- f(q)^\dagger\cdot  b_{>d}\right)$ in Eq.~(\ref{Equation: Square Root Second order expansion}) is of order $o_*\left(\mathbb{Q}_N\right)$, where $\widetilde{V}(t):=V(t) \int \hat{h}(z)iz\ e^{iz\left(\sum_{j=1}^d t_j^2+f(t)^\dagger\cdot f(t)\right)}\ \mathrm{d}z$. By the definition of $\mathbb{L}'$ and $Q$, we have the identity
\begin{align*}
\widetilde{V}(q)\left(\mathbb{L}'-Q- f(q)^\dagger\cdot  b_{>d}- b_{>d}^\dagger\cdot  f(q)\right)=\widetilde{V}(q)  b_{>d}^\dagger\cdot  b_{>d}+\widetilde{V}(q)\left( p'\right)^\dagger\cdot   p'-\frac{d}{2N}V(q).
\end{align*}
The first term is of order $o_*\left(\mathbb{Q}_N\right)$ by Lemma \ref{Lemma: Function of q}, the second term is by Lemma \ref{Lemma: p'} and regarding the last term we know that $\frac{d}{2N}V(q)=o_*\left(\mathbb{Q}_N\right)$ by Lemma \ref{Lemma: Function of q}.
\end{proof}

Before we can verify the Taylor approximation for the operator $\left(1-\mathbb{L}\right)^{\frac{3}{2}}$ in Corollary \ref{Corollary: Taylor of the square root}, we need the following two results, which are of independent relevance for the proof of Theorem \ref{Theorem: Main B}. 

\begin{lem}
\label{Lemma: Lifting}
We have $\left(\mathbb{L}'\right)^2=o_*\left(1\right)$, and furthermore
\begin{align}
\sqrt{1-\mathbb{L}'}\, &\mathbb{Q}_N\, \sqrt{1-\mathbb{L}'}=O_*\left(\mathbb{Q}_N\right),\label{Equations: L' one}\\
\label{Equations: L' two}\mathbb{L}'\,  &\mathbb{Q}_N\, \mathbb{L}'=o_*\left(\mathbb{Q}_N\right).
\end{align}
\end{lem}
\begin{proof}
Note that $\|\mathbb{L}' \pi_{M,N}\|=\frac{M}{N}$ for all $M\leq N$, and therefore we immediately obtain $\left(\mathbb{L}'\right)^2=o_*\left(1\right)$. In order to verify Equations (\ref{Equations: L' one}) and (\ref{Equations: L' two}), it is enough to prove that $\sqrt{1-\mathbb{L}'} \left(\xi^\dagger\cdot  \xi\right) \sqrt{1-\mathbb{L}'}=O_*\left(\mathbb{Q}_N\right)$ and $\mathbb{L}' \left(\xi^\dagger\cdot  \xi\right) \mathbb{L}'=o_*\left(\mathbb{Q}_N\right)$ for $\xi\in \{p',b_{>d}\}$.\\

\underline{The case $\xi= p'$}: In order to verify $\sqrt{1-\mathbb{L}'} \left(\xi^\dagger\cdot  \xi\right) \sqrt{1-\mathbb{L}'}=O_*\left(\mathbb{Q}_N\right)$, observe that we have for all $\Psi\in \mathcal{W}_N \mathcal{F}_{\leq N-1}$ the commutation law
\begin{align*}
p_j' \sqrt{1-\mathbb{L}'}\Psi=\frac{\sqrt{1-\mathbb{L}'-\frac{1}{N}}+\sqrt{1-\mathbb{L}'+\frac{1}{N}}}{2} p_j'\Psi+\frac{\sqrt{1-\mathbb{L}'-\frac{1}{N}}-\sqrt{1-\mathbb{L}'+\frac{1}{N}}}{2} q_j\Psi.
\end{align*}
For $M\leq N-2$, let us define the operators $B_{M,N}:=\frac{\sqrt{1-\mathbb{L}'-\frac{1}{N}}+\sqrt{1-\mathbb{L}'+\frac{1}{N}}}{2} \pi_{M+1,N}$ and $\widetilde{B}_{M,N}:=\frac{\sqrt{1-\mathbb{L}'-\frac{1}{N}}-\sqrt{1-\mathbb{L}'+\frac{1}{N}}}{2} \pi_{M+1,N}$. Note that $\|B_{M,N}\|\leq 1$ and $\|\widetilde{B}_{M,N}\|^2\leq \frac{C}{N^2}$ for all $\frac{M}{N}\leq \delta_0$, where $C$ and $0<\delta<1$ are suitable constants. Consequently
\begin{align*}
&\pi_{M,N} \sqrt{1-\mathbb{L}'}\, \left(p'\right)^\dagger\!\cdot\! p'\,\sqrt{1-\mathbb{L}'}\pi_{M,N}=\left| \left(B_{M,N}\!\otimes\! 1_{\mathcal{H}}\!\cdot\! p'\!+\!\widetilde{B}_{M,N}\!\otimes\! 1_{\mathcal{H}}\!\cdot\!  q\right) \pi_{M,N}\right|^2\\
&\ \ \leq 2\left|B_{M,N}\otimes 1_{\mathcal{H}}\cdot p'\, \pi_{M,N}\right|^2+2\left|\widetilde{B}_{M,N}\otimes 1_{\mathcal{H}}\cdot  q\, \pi_{M,N}\right|^2\\
&\ \ \leq \pi_{M,N}\, \left(p'\right)^\dagger\cdot p'\, \pi_{M,N}+\frac{C(d+1)}{N^2},
\end{align*}
which concludes the proof of $\sqrt{1-\mathbb{L}'} \left(p'\right)^\dagger\cdot p'\, \sqrt{1-\mathbb{L}'}=O_*\left(\mathbb{Q}_N\right)$. The estimate $\mathbb{L}' \left(p'\right)^\dagger\cdot p'\, \mathbb{L}'=o_*\left(\mathbb{Q}_N\right)$ follows from an analogue commutation law.\\

\underline{The case $\xi= b_{>d}$}: In order to verify $\sqrt{1-\mathbb{L}'}\, b_{>d}^\dagger\cdot  b_{>d}\, \sqrt{1-\mathbb{L}'}=O_*\left(\mathbb{Q}_N\right)$, note that $\left[\sqrt{1-\mathbb{L}'}-\eta_1(q)\right]^2=o_*\left(\mathbb{Q}_N\right)$ by \ref{Lemma: Taylor of the square root}, i.e. there exists a function $\epsilon$ with $\epsilon(\delta)\underset{\delta\rightarrow 0}{\longrightarrow}0$ and $\pi_{M,N} \left[\sqrt{1-\mathbb{L}'}-\eta_1(q)\right]^2 \pi_{M,N}\leq \epsilon\left(\frac{M}{N}\right)\ \pi_{M,N}\, \mathbb{Q}_N\, \pi_{M,N}$. By Lemma \ref{Lemma: Auxiliary}, we know that $\hat{\pi}_{M,N}\,  b_{>d}^\dagger\cdot  b_{>d}\, \hat{\pi}_{M,N}\leq C$ for a constant $C$. Furthermore $\left[\sqrt{1-\mathbb{L}'}-\eta_1(q)\right] \hat{\pi}_{M,N}=\hat{\pi}_{M,N} \left[\sqrt{1-\mathbb{L}'}-\eta_1(q)\right]$. Let us define $S:=\left[\sqrt{1-\mathbb{L}'}-\eta_1(q)\right]  b_{>d}^\dagger\cdot  b_{>d} \left[\sqrt{1-\mathbb{L}'}-\eta_1(q)\right] $, and estimate
\begin{align*}
\pi_{M,N} S \pi_{M,N}&=\pi_{M,N} \left[\sqrt{1-\mathbb{L}'}-\eta_1(q)\right]  \hat{\pi}_{M,N}\,  b_{>d}^\dagger\!\cdot\!  b_{>d}\, \hat{\pi}_{M,N} \left[\sqrt{1-\mathbb{L}'}-\eta_1(q)\right] \pi_{M,N}\\
&\leq 4\ \pi_{M,N}\left[\sqrt{1-\mathbb{L}'}-\eta_1(q)\right] ^2 \pi_{M,N}\leq 4\ \epsilon\left(\frac{M}{N}\right)\ \pi_{M,N}\, \mathbb{Q}_N\, \pi_{M,N}.
\end{align*}
Hence, $S=o_*\left(\mathbb{Q}_N\right)$ and therefore
\begin{align*}
\sqrt{1-\mathbb{L}'}\, b_{>d}^\dagger\cdot  b_{>d}\, \sqrt{1-\mathbb{L}'}\leq 2\left(S+\eta_1(q)\, b_{>d}^\dagger\cdot b_{>d}\, \eta_1(q)\right)=O_*\left(\mathbb{Q}_N\right).
\end{align*}
The proof of $\mathbb{L}'\, b_{>d}^\dagger\cdot  b_{>d}\, \mathbb{L}'=o_*\left(\mathbb{Q}_N\right)$ can be carried out in a similar fashion.
\end{proof}

\begin{cor}
\label{Corollary: Lifting}
Let $X_N$ be a sequence with $X_N=O_*\left(\mathbb{Q}_N\right)$ and $Y_N$ a sequence with $Y_N=o_*\left(\mathbb{Q}_N\right)$. Then,
\begin{align}
\label{Equation: Lifting 1}
\sqrt{1-\mathbb{L}'} &X_N\sqrt{1-\mathbb{L}'}=O_*\left(\mathbb{Q}_N\right),\\
\label{Equation: Lifting 2}\sqrt{1-\mathbb{L}'}&Y_N\sqrt{1-\mathbb{L}'}=o_*\left(\mathbb{Q}_N\right),\\
\label{Equation: Lifting 3}\mathbb{L}'&X_N\mathbb{L}'=o_*\left(\mathbb{Q}_N\right).
\end{align}
\end{cor}
\begin{proof}
The Corollary follows from Lemma \ref{Lemma: Lifting} and the fact that $\pi_{M,N}$ commutes with $\sqrt{1-\mathbb{L}'}$ and $\mathbb{L}'$. For the purpose of illustration, let us verify Eq.~(\ref{Equation: Lifting 1}). By the assumptions of the Corollary we know that there exist constants $C$ and $\delta>0$, such that $\pi_{M,N}\, X_N\, \pi_{M,N}\leq C \pi_{M,N}\, \mathbb{Q}_N\, \pi_{M,N}$. Consequently
\begin{align*}
\pi_{M,N}&\, \sqrt{1-\mathbb{L}'}X_N\sqrt{1-\mathbb{L}'}\, \pi_{M,N}=\sqrt{1-\mathbb{L}'}\pi_{M,N}\, X_N\, \pi_{M,N}\sqrt{1-\mathbb{L}'}\\
&\leq C \pi_{M,N}\, \sqrt{1-\mathbb{L}'}\mathbb{Q}_N\sqrt{1-\mathbb{L}'}\, \pi_{M,N}=O_*\left(\mathbb{Q}_N\right),
\end{align*}
where we have used Eq.~(\ref{Equations: L' one}) from Lemma \ref{Lemma: Lifting} in the last equality.
\end{proof}

\begin{cor}
\label{Corollary: Taylor of the square root}
Let Assumption \ref{Assumption: Part II} hold and let $\eta_m$ be the functions from Eq.~(\ref{Equation: eta_m}), with $m\in \{0,\dots,3\}$. Then
\begin{align*}
\left[\left(1-\mathbb{L}'\right)^{\frac{m}{2}}-\eta_m(q)\right]^2=o_*\left(\mathbb{Q}_N\right).\\
\end{align*}
\end{cor}
\begin{proof}
The case $m=0$ is trivial. The case $m=1$ is the content of Lemma \ref{Lemma: Taylor of the square root} and the case $m=2$ follows from Lemma \ref{Lemma: Taylor of the square root - order two}. Let us now verify the statement in the case $m=3$. Using the fact that $\eta_3(t)=\eta_2(t)\eta_1(t)$, we obtain
\begin{align*}
&\left(1-\mathbb{L}'\right) \sqrt{1-\mathbb{L}'}-\eta_3(q)=\left[\left(1-\mathbb{L}'\right)-\eta_2(q)\right] \sqrt{1-\mathbb{L}'}+\eta_2(q) \left[\sqrt{1-\mathbb{L}'}-\eta_1(q)\right]\\
&=-\left(f(q)^\dagger\cdot  b_{>d} + b_{>d}^\dagger\cdot  f(q)+( p')^\dagger\cdot  p'+ b_{>d}^\dagger\cdot  b_{>d}-\frac{d}{2N}\right) \sqrt{1-\mathbb{L}'}\\
&\ \ \ \ +\tau_r(q) \eta_2(q) \left[\sqrt{1-\mathbb{L}'}-\eta_1(q)\right]+\left[1-\tau_r(q)\right] \eta_2(q) \left[\sqrt{1-\mathbb{L}'}-\eta_1(q)\right],
\end{align*}
where $\tau:\mathbb{R}^d\longrightarrow \mathbb{R}$ is a function with $\tau|_{B_1(0)}=0$, $\tau|_{\mathbb{R}^d\setminus B_2(0)}=1$ and $0\leq \tau \leq 1$. Since the function $\tilde{\eta}=(1-\tau)\eta_2$ is bounded by a constant $c$, we obtain using Lemma \ref{Lemma: Taylor of the square root}
\begin{align*}
\left[\sqrt{1-\mathbb{L}'}-\eta_1(q)\right] \tilde{\eta}^2(q) \left[\sqrt{1-\mathbb{L}'}-\eta_1(q)\right]\leq c\ \left[\sqrt{1-\mathbb{L}'}-\eta_1(q)\right]^2=o_*\left(\mathbb{Q}_N\right).
\end{align*}
Note that $\eta':=\left(\tau\ \eta_2\right)^2$ is zero in a neighborhood of zero. Therefore, $\eta'(q)=O_*\left(e^{-\delta N}\right)$ and $\eta^2_1(q) \eta'(q)=O_*\left(e^{-\delta N}\right)$ for some $\delta>0$ by Lemma \ref{Lemma: q function}. By Corollary \ref{Corollary: Lifting}, we obtain in particular that $\sqrt{1-\mathbb{L}'} \eta'(q) \sqrt{1-\mathbb{L}'}=o_*\left(\mathbb{Q}_N\right)$. Hence we have the estimate
\begin{align*}
\left[\! \sqrt{1\!-\!\mathbb{L}'}\!-\!\eta_1(q)\!\right] \eta'(q)\left[\sqrt{1\!-\!\mathbb{L}'}\!-\!\eta_1(q)\!\right]\!\leq\! 2\sqrt{1\!-\!\mathbb{L}'} \eta'(q)\sqrt{1\!-\!\mathbb{L}'}\!+\!2 \eta^2_1(q) \eta'(q)\!=\!o_*\left(\mathbb{Q}_N\right).
\end{align*}
By Lemma \ref{Lemma: b}, Lemma \ref{Lemma: p'} and Corollary \ref{Corollary: Lifting}, we know that the operators
\begin{align*}
\sqrt{1-\mathbb{L}'} \left( b_{>d}^\dagger\cdot  f(q)+ f(q)^\dagger\cdot  b_{>d}\right)^2 \sqrt{1-\mathbb{L}'}
\end{align*}
$\sqrt{1-\mathbb{L}'} \left(( p')^\dagger\cdot  p'\right)^2  \sqrt{1-\mathbb{L}'}$ as well as $\sqrt{1-\mathbb{L}'} \left( b_{>d}^\dagger\cdot  b_{>d}\right)^2  \sqrt{1-\mathbb{L}'}$ are of order $o_*\left(\mathbb{Q}_N\right)$ as well, and therefore
\begin{align*}
\sqrt{1-\mathbb{L}'} \left( b_{>d}^\dagger\cdot  f(q)+ f(q)^\dagger\cdot  b_{>d}+( p')^\dagger\cdot  p'+ b_{>d}^\dagger\cdot  b_{>d}-\frac{d}{2N}\right)^2 \sqrt{1-\mathbb{L}'}=o_*\left(\mathbb{Q}_N\right).
\end{align*}
We conclude that $\left(1-\mathbb{L}'\right) \sqrt{1-\mathbb{L}'}-\eta_3(q)=T_1+T_2+T_3$ is a sum of terms with $T_i^\dagger T_i=o_*\left(\mathbb{Q}_N\right)$, and therefore $\left[\left(1-\mathbb{L}'\right) \sqrt{1-\mathbb{L}'}-\eta_3(q)\right]^2=o_*\left(\mathbb{Q}_N\right)$.
\end{proof}

\begin{center}
\textsc{Acknowledgments}
\end{center}

We are grateful to Rupert Frank for helpful discussions at an early stage of this project.
Funding from the European Union’s Horizon 2020 research and innovation programme
under the ERC grant agreement No 694227 is  acknowledged.

\bibliographystyle{plain}

\begin{thebibliography}{10}

\bibitem{BBSS2}
C. Boccato, C.  Brennecke, S. Cenatiempo, and B. Schlein.
\newblock Bogoliubov theory in the Gross–Pitaevskii limit.
\newblock {\em Acta Math.}, 222:219--335, 2019.

\bibitem{BBSS1}
C. Boccato, C.  Brennecke, S. Cenatiempo, and B. Schlein.
\newblock The excitation spectrum of Bose gases interacting through singular potentials.
\newblock {\em J. Eur. Math. Soc.}, 22:2331--2403, 2020.

\bibitem{B}
N.~Bogoliubov.
\newblock On the theory of superfluidity.
\newblock {\em Journal of Physics (USSR)}, 11:23--32, 1947.

\bibitem{BCS}
C. Brennecke, M. Caporaletti, and B. Schlein.
\newblock Excitation Spectrum for Bose Gases beyond the Gross-Pitaevskii Regime.
\newblock preprint arXiv:2104.13003.

\bibitem{BSS}
C. Brennecke, B. Schlein, and S. Schraven.
\newblock Bogoliubov Theory for Trapped Bosons in the Gross-Pitaevskii Regime.
\newblock {\em Annales Henri Poincar{\'e}}, 23:1583--1658, 2022. 

\bibitem{DN}
J. Derezi{\'n}ski and M. Napi{\'o}rkowski.
\newblock Excitation Spectrum of Interacting Bosons in the Mean-Field Infinite-Volume Limit.
\newblock {\em Annales Henri Poincar{\'e}}, 15:2409--2439, 2014.

\bibitem{FLS}
R.~Frank, E.~Lieb, and R. Seiringer.
\newblock Symmetry of bipolaron bound states for small coulomb repulsion.
\newblock {\em Communications in Mathematical Physics}, 319:557–573, 2013.

\bibitem{GS}
P.~Grech and R.~Seiringer.
\newblock The excitation spectrum for weakly interacting {B}osons in a trap.
\newblock {\em Communications in Mathematical Physics}, 322:559–591, 2013.

\bibitem{G}
P.~Gross.
\newblock Particle-like solutions in field theory.
\newblock {\em Annals of Physics}, 19:219–233, 1962.

\bibitem{GZ}
Y.~Guo and X.~Zeng.
\newblock The Lieb--Yau conjecture for ground states of pseudo-relativistic
  boson stars.
\newblock {\em Journal of Functional Analysis}, 278:108510, 2020.

\bibitem{H}
I.~Herbst.
\newblock Spectral theory of the operator $(p^2+m^2)^{\frac{1}{2}}-ze^2/r$.
\newblock {\em Communications in Mathematical Physics}, 53:285--294, 1977.

\bibitem{HM}
R.~Hudson and G.~Moody.
\newblock Locally normal symmetric states and an analogue of de {F}inetti's
  theorem.
\newblock {\em Zeitschrift für Wahrscheinlichkeitstheorie und Verwandte
  Gebiete}, 33:343--351, 1976.

\bibitem{K}
T.~Kato.
\newblock {\em Perturbation theory for linear operators}.
\newblock Springer, Berlin, Heidelberg, 1966.


\bibitem{KM}
Y. Kato and N. Mugibayashi.
\newblock Friedrichs-Berezin Transformation and Its Application to the Spectral Analysis of the BCS Reduced Hamiltonian.
\newblock {\em Progress of Theoretical Physics}, 38:813--831, 1967.


\bibitem{Kw}
M.~Kwong.
\newblock Uniqueness of positive solutions of ${\Delta} u-u+u^p=0$ in
  $\mathbb{R}^n$.
\newblock {\em Archive for Rational Mechanics and Analysis}, 105:243--266,
  1989.

\bibitem{Le}
E.~Lenzmann.
\newblock Uniqueness of ground states for pseudo-relativistic {H}artree
  equations.
\newblock {\em Analysis and PDE}, 2:1--27, 2009.

\bibitem{LL}
E.~Lenzmann and M.~Lewin.
\newblock On singularity formation for the ${L}^2$-critical {B}oson star
  equation.
\newblock {\em Nonlinearity}, 24:3515--3540, 2011.

\bibitem{Lew}
M.~Lewin.
\newblock Geometric methods for nonlinear many-body quantum systems.
\newblock {\em Journal of Functional Analysis}, 260:3535--3595, 2011.

\bibitem{LNR}
M.~Lewin, P.~Nam, and N.~Rougerie.
\newblock Derivation of {H}artree's theory for generic mean-field {B}ose
  systems.
\newblock {\em Advances in Mathematics}, 254:570--621, 2014.

\bibitem{LNS}
M.~Lewin, P.~Nam, and B.~Schlein.
\newblock Fluctuations around {H}artree states in the mean-field regime.
\newblock {\em American Journal of Mathematics}, 137:1613--1650, 2015.

\bibitem{LNSS}
M.~Lewin, P.~Nam, S.~Serfaty, and J.P. Solovej.
\newblock Bogoliubov spectrum of interacting {B}ose gases.
\newblock {\em Communications on Pure and Applied Mathematics}, 68:413--471, 2015.

\bibitem{Li}
E.~Lieb.
\newblock Existence and uniqueness of the minimizing solution of {C}hoquard's
  nonlinear equation.
\newblock {\em Studies in Applied Mathematics}, 57:93--105, 1977.

\bibitem{LS}
E.~Lieb and J.~Solovej.
\newblock Ground state energy of the one-component charged {B}ose gas.
\newblock {\em Communications in Mathematical Physics}, 217:127–163,
  2001.

\bibitem{LY}
E.~Lieb and H.~Yau.
\newblock The {C}handrasekhar theory of stellar collapse as the limit of
  quantum mechanics.
\newblock {\em Communications in Mathematical Physics}, 112:147–174, 1987.


\bibitem{MR}
P. Martin and F. Rothen.
\newblock {\em Many-Body Problems and Quantum Field Theory}.
\newblock Springer, Berlin, Heidelberg, 2002.

\bibitem{M}
J.B.~McGuire.
\newblock Study of Exactly Soluble One‐Dimensional $N$‐Body Problems.
\newblock {\em Journal of Mathematical Physics}, 5:622--636, 1964.


\bibitem{NNS}
P.T.~Nam, M.~Napíorkowski, and J.P.~Solovej.
\newblock Diagonalization of bosonic quadratic {H}amiltonians by {B}ogoliubov
  transformations.
\newblock {\em Journal of Functional Analysis}, 270:4340--4368, 2015.

\bibitem{NS}
P.T.~Nam and R.~Seiringer.
\newblock Collective excitations of {B}ose gases in the mean-field regime.
\newblock {\em Archive for Rational Mechanics and Analysis}, 215:381–417,
  2015.

\bibitem{NT}
P.T. Nam and A. Triay.
\newblock Bogoliubov excitation spectrum of trapped Bose gases in the Gross-Pitaevskii regime.
\newblock preprint arXiv:2106.11949.

\bibitem{N}
E.~Nelson.
\newblock Interaction of nonrelativistic particles with a quantized scalar
  field.
\newblock {\em Journal of Mathematical Physics}, 5:1190–1197, 1964.

\bibitem{S}
R.~Seiringer.
\newblock The excitation spectrum for weakly interacting {B}osons.
\newblock {\em Communications in Mathematical Physics}, 306:565–578, 2011.

\bibitem{St}
E.~Størmer.
\newblock Symmetric states of infinite tensor products of ${C}^*$-algebras.
\newblock {\em Journal of Functional Analysis}, 3:48--68, 1969.

\bibitem{We}
A.~Wehrl.
\newblock Three theorems about entropy and convergence of density matrices.
\newblock {\em Reports on Mathematical Physics}, 10:159--163, 1976.

\bibitem{W}
M.~Weinstein.
\newblock Modulational stability of ground states of nonlinear {S}chrödinger
  equations.
\newblock {\em SIAM Journal on Mathematical Analysis}, 16:472--491, 1985.


\end{thebibliography}

\end{document}